\documentclass[pdflatex,sn-mathphys-num]{sn-jnl}
\usepackage[T1]{fontenc}
\usepackage{silence}
\usepackage{hyperref}
\usepackage{hhline}
\usepackage{multicol, multirow}
\hypersetup{
    colorlinks=true,
    linkcolor=blue,
    filecolor=magenta,      
    urlcolor=cyan,
    citecolor=blue,
    pdftitle={TaRDiS},
    }
    \usepackage{lmodern} % to get rid of warnings about font

\AtBeginDocument{%
    \DeclareFontShape{T1}{lmr}{m}{scit}{<->ssub*lmr/m/scsl}{}%
    \DeclareFontShape{T1}{lmr}{bx}{sc}{<->ssub*lmr/m/sc}{}%
}

\usepackage{caption}
\usepackage{subcaption}
\usepackage{multirow}
\usepackage{tabularray}
\usepackage{graphicx} 
\usepackage{xcolor}
\usepackage{framed}
\usepackage{amsmath, amssymb, float, amsthm}
\usepackage{tabularray} % for table
\usepackage{bm} % bold math, for operators
\usepackage{bold-extra} % trying to get \textsc and \textbf to work together
\usepackage[T1]{fontenc} % needed for author names using fun diacritics/alphabets

\usepackage{array}
\newtheorem{theorem}{Theorem}
\newtheorem{lemma}{Lemma}
\newtheorem{corollary}{Corollary}

\newtheorem{definition}{Definition}
\def\qed{\hfill\hbox{${\vcenter{\vbox{
	\hrule height 0.4pt\hbox{\vrule width 0.4pt height 6pt						
	\kern5pt\vrule width 0.4pt}\hrule height 0.4pt}}}$}}

\usepackage{vmargin}
\setmarginsrb{3.5cm}{2.25cm}{3.5cm}{3.05cm}{0.3cm}{0.3cm}{-0.3cm}{1.0cm}
\begin{document}

\title[Temporal Reachability Dominating Sets: contagion in temporal graphs]{Temporal Reachability Dominating Sets: contagion in temporal graphs}

\author*[1]{\fnm{David C.} \sur{Kutner}}\email{david.c.kutner@durham.ac.uk}
\equalcont{Authors contributed equally to this work}
% DK Orcid: orcidID{0000-0003-2979-4513}

\author*[2]{Laura Larios-Jones}\email{laura.larios-jones@glasgow.ac.uk}
\equalcont{Authors contributed equally to this work.}
% LLJ Orcid: \orcidID{0000-0003-3322-0176}

\affil[1]{\orgdiv{Department of Computer Science}, \orgname{Durham University}} 
% , \orgaddress{Upper Mountjoy Campus, Stockton Road, Durham DH1 3LE, UK}

\affil[2]{\orgdiv{School of Computing Science}, \orgname{University of Glasgow}}
% , \orgaddress{Lilybank Gardens, Glasgow, G12 8RZ, UK}

\abstract{
    Given a population with dynamic pairwise connections, we ask if the entire population could be (indirectly) infected by a small group of $k$ initially infected individuals. We formalise this problem as the \textsc{Temporal Reachability Dominating Set} (\textsc{TaRDiS}) problem on temporal graphs. We provide positive and negative parameterized complexity results in four different parameters: the number $k$ of initially infected, the lifetime $\tau$ of the graph, the number of locally earliest edges in the graph, and the treewidth of the footprint graph $\mathcal{G}_\downarrow$.   
    We additionally introduce and study the \textsc{MaxMinTaRDiS} problem, where the aim is to schedule connections between individuals so that at least $k$ individuals must be infected for the entire population to become fully infected. We classify three variants of the problem: Strict, Nonstrict, and Happy. We show these to be coNP-complete, NP-hard, and $\Sigma_2^P$-complete, respectively. Interestingly, we obtain hardness of the Nonstrict variant by showing that a natural restriction is exactly the well-studied \textsc{Distance-3 Independent Set} problem on static graphs.
}


\maketitle

\section{Introduction}\label{sec:Introduction}

A natural problem in the study of networks is that of finding a small set of individuals which together can affect the entire network. This problem is practically relevant in identifying influential entities in social networks, in gauging the risk of viral spread in biological networks, or in choosing sources to broadcast from in wireless networks. We study this problem through the lens of \emph{temporal graphs}.  That is, graphs which change over time, capturing the dynamic nature intrinsic to many real-world networks. 

We formalize the notion of a set of sources which together reach the whole temporal graph as a \emph{Temporal Reachability Dominating Set}, or TaRDiS. We study the complexity of the problem of finding a TaRDiS of a given size in a given temporal graph. Later, we ask: if we can choose when connections between individuals exist, can we maximize the size of the minimum TaRDiS? That is, for viral spread, can we guarantee that the entire population will not be contaminated by at most $k$ initial infections. We refer to this problem as \textsc{MaxMinTaRDiS}.

Our problems sit at the intersection of reachability and covering, which are two rich classes of (temporal) graph problems. 
% \todo[inline]{Introduce citations below, some (2) more problems to table}
Reachability is a fundamental concept in temporal graph theory. Temporal reachability and related problems have been the subject of extensive study since the seminal work in the field by Kempe, Kleinberg and Kumar \cite{kempe_connectivity_2000}. Several works focus on the complexity of choosing or modifying times to optimize reachability \cite{kempe_connectivity_2000,deligkas_potapov_optimizing_2022,enright_assigning_2021,molter_temporal_2021,klobas_mertzios_molter_spirakis_complexity_2022}, while some consider as input temporal graphs in which times are immutably fixed \cite{casteigts_temporal_2021,erlebach_spooner_cops_2019,erlebach_spooner_exploration_2018,whitbeck_temporal_2012,casteigts_simple_2022_bugfree}. A comparison between our problems and some of those in the literature is shown in Table~\ref{table:problems}. 

\textsc{MaxMinTaRDiS} is closely related to the \textsc{MinMaxReach} problem studied by Enright, Meeks and Skerman \cite{enright_assigning_2021}, in which the objective is to minimize the maximum number of individuals reached by any single vertex. \textsc{TaRDiS}, on the other hand, vastly generalizes Casteigts's \cite{casteigts_journey_2018_bugfree} notion of $\mathcal{J}^{1\forall}$ connectivity, in which the question is whether any single vertex reaches the entire network. This framing of reachability asks what the worst-case spread is from a single source in the temporal graph. In reality, studied populations are often infected by several individuals. For example, SARS-CoV-2 had been independently introduced to the UK at least 1300 times by June 2020 \cite{pybus2020preliminary}. This, a dynamic population infected by many sources, is precisely the setting motivating our \textsc{TaRDiS} problem.

\def\arraystretch{1.5}
\begin{table}[!ht]
        % \footnotesize % uncomment for smaller text in table
		\centering
		\setlength{\tabcolsep}{1pt}
		\begin{tabular}{
				| m{.2\linewidth}
				| >{\centering\arraybackslash}m{.2\linewidth} 
                | >{\centering\arraybackslash}m{.2\linewidth} 
				| >{\centering\arraybackslash}m{.2\linewidth} 
				| >{\centering\arraybackslash}m{.2\linewidth} |}
			\hline
            \textbf{Problem} 
            & \textbf{Choose} 
            & \textbf{such that}
            & \textbf{reach} 
            & \textbf{Aim/motivation}
            \\
			\hhline{|=|=|=|=|=|}
			{\bf \textsc{MaxMinTaRDiS}}
            & times for edges
            & no coalition of $k$ nodes
            & all nodes
            & Minimize connectivity
            \\
            \hline
            {\bf \textsc{TaRDiS}}            
            & nothing
            & a coalition of $k$ nodes
            & all nodes
            & Assess connectivity
            \\
			\hline
            \textsc{MinMaxReach} \cite{enright_assigning_2021}
            & times for edges
            & no single node
            & more than $k$ nodes
            & Minimize connectivity
            \\
			\hline
            $\mathcal J^{1\forall}$ connectedness \cite{casteigts_journey_2018_bugfree}
            & nothing
            & a single node
            & all nodes
            & Assess connectivity
            \\
			\hline
           {\textsc{Reachability \newline Inference} }\cite{kempe_connectivity_2000}
            & times for edges
            & a designated root node
            & all ``good'' nodes and no ``bad'' nodes
            & Network design
            \\
			\hline
            \textsc{ReachFast} \cite{deligkas_eiben_skretas_reachability_2023}, 
            & edges to delay
            & each source in the designated set
            & all nodes
            & Maximize connectivity
            \\
            \hline
            \textsc{MinReachDelete} \cite{molter_temporal_2021, enright_deleting_2021}
            & edges to delete
            & the specified coalition
            & at most $r$ vertices
            & Minimize connectivity
            \\
            \hline
            \textsc{MinReachDelay} \cite{molter_temporal_2021}
            & edges to delay
            & the specified coalition
            & at most $r$ vertices
            & Minimize connectivity
            \\
            \hline
            
		\end{tabular}\medskip
		\caption{Comparison of our problems \textsc{TaRDiS} and \textsc{MaxMinTaRDiS} to problems in the literature.}\label{table:problems}
	\end{table}

Our work is focused on the computational (parameterized) complexity of \textsc{TaRDiS} and \textsc{MaxMinTaRDiS}, which depends heavily on which formalisation of the problems is considered. In particular, instantaneous spread through a large swath of the population, while realistic in some computer networks, is inconsistent with viral spread in a biological system. Further, should multiple interactions between the same pair of individuals be allowed? Lastly, should it be possible for a single individual to simultaneously interact with several others?

We consider all combinations of answers to these questions, and in all cases for both problems: show that the problem is computationally hard in general; identify the maximum number of discrete times at which edges can appear such that the problem remains tractable; and provide (parameterized) algorithms. Interestingly, we also show through a nontrivial proof that \textsc{MaxMinTaRDiS} generalizes the well-studied \textsc{Distance-3 Independent Set} problem.

\subsection{Problem Setting}\label{subsec:problem_setting}

We begin with some standard definitions. We denote by $[i,j]$ the set $\{i, i+1, \ldots, j\}$ and $[j]$ the set $[1,j]$. Let $G=(V,E)$ be an undirected graph with $(u,v) \in E$. We say that $u$ and $v$ are \emph{adjacent} (also \emph{neighbours}) and that the edge $(u,v)$ is \emph{incident} to both $u$ and $v$. If $S \subseteq V$ is a set of vertices, we say an edge $(u,v)$ is incident to $S$ if one of its endpoints $u$ or $v$ is in $S$. For any vertex $v \in V$,  the \emph{closed neighbourhood} of $v$ is denoted $N[V]:= \{v\} \cup \{u : u $ is adjacent to $v\}$. We say $G$ is \emph{planar} if it can be drawn in a plane without any edges intersecting. 

A \emph{temporal graph} $\mathcal{G}=(V,E,\lambda)$ consists of a set of vertices $V$, a set of edges $E$ and a function $\lambda:E\to 2^{\tau} \setminus \{\emptyset\}$ which maps each edge to a discrete set of times where the \emph{lifetime} $\tau\in \mathbb N$ of a temporal graph is the value of the latest timestep. We refer to $\lambda$ as the \emph{temporal assignment} of $\mathcal{G}$. If $t\in\lambda(e)$ then we call the pair $(e,t)$ a \emph{time-edge}, and say $e$ is \emph{active} at time $t$. The set of all time-edges is denoted $\mathcal{E}$. We abuse notation and write $\lambda(u,v)$ to mean $\lambda((u,v))$.

For a static graph $G=(V,E)$, we denote the temporal graph $(V,E,\lambda)$ by $(G,\lambda)$. We also use $V(\mathcal{G}),E(\mathcal{G})$ to refer to the vertex and edge sets of $\mathcal{G}$, respectively, and use $E_t(\mathcal{G})$ to refer to the set of edges active at time $t$, and call $G_t(\mathcal{G})=(V(\mathcal{G}),E_t(\mathcal{G}))$ the \emph{snapshot at time $t$}. When $\mathcal{G}$ is clear from the context we may omit it. Also, we use the convention that no snapshot is empty, which guarantees $\tau \le |\mathcal{E}|$. The static graph $\mathcal{G}_{\downarrow}=(V,E)$ is referred to as the \emph{footprint} of $\mathcal{G}$.

\begin{figure}[!ht]
     \centering
     \begin{subfigure}[b]{0.22\textwidth}
         \centering
         \includegraphics[width=\textwidth,page=1]{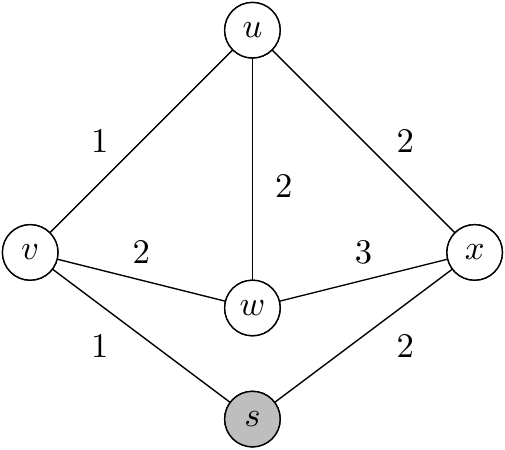}
         \caption{$\mathcal G$}
         \label{fig:snapshots_g}
     \end{subfigure}
     \hfill     \vrule      \hfill
     \begin{subfigure}[b]{0.22\textwidth}
         \centering
         \includegraphics[width=\textwidth,page=2]{IPE_figs_a4.pdf}
         \caption{$G_1$}
         \label{fig:snapshots_g1}
     \end{subfigure}
     \hfill     \vrule      \hfill
     \begin{subfigure}[b]{0.22\textwidth}
         \centering
         \includegraphics[width=\textwidth,page=3]{IPE_figs_a4.pdf}
         \caption{$G_2$}
         \label{fig:snapshots_g2}
     \end{subfigure}
     \hfill     \vrule      \hfill
     \begin{subfigure}[b]{0.22\textwidth}
         \centering
         \includegraphics[width=\textwidth,page=4]{IPE_figs_a4.pdf}
         \caption{$G_3$}
         \label{fig:snapshots_g3}
     \end{subfigure}
        \caption{Spread in a temporal graph from source $s$ through snapshots. Vertices are shaded (half-shaded) when reached from $s$ by a strict (nonstrict) temporal path.}
        \label{fig:snapshots}
\end{figure}

A \emph{strict} (respectively \emph{nonstrict}) \emph{temporal path} from a vertex $u$ to a vertex $v$ is a sequence of time-edges $(e_1,t_1),\ldots, (e_l,t_l)$ such that $e_1 \ldots e_l$ is a static path from $u$ to $v$ and $t_i< t_{i+1}$ (resp. $t_i\leq t_{i+1}$) for $i\in [1,l-1]$. A vertex $u$ \emph{temporally reaches} (or just ``reaches'') a vertex $v$ if there is a temporal path from $u$ to $v$. The \emph{reachability set} $R_u(\mathcal{G})$ of a vertex $u$ is the set of vertices reachable from $u$. When the graph is clear from context, we may simply use $R_u$ to refer to the reachability set of a vertex $u$. We say a vertex $u$ is reachable from a set $S$ if for some $v\in S$, $u\in R_v(\mathcal{G})$. A set of vertices $T$ is \emph{temporal reachability dominated} by another set of vertices $S$ if every vertex in $T$ is reachable from $S$. Domination of and by single vertices is analogously defined. Strict and nonstrict reachability are illustrated in Figure~\ref{fig:snapshots}. We differentiate between strict and nonstrict reachability by introducing a superscript $<$ or $\leq$ to the appropriate operators. For example, in Figure \ref{fig:snapshots}, $u$ is in $R_s^{\leq}$, but is not in $R_s^<$. Note that any strict temporal path is also a nonstrict temporal path, but the converse does not necessarily hold.

Casteigts, Corsini, and Sarkar \cite{casteigts_simple_2022_bugfree} define three useful properties a temporal graph may exhibit. A temporal graph is called \emph{simple} if each edge is active exactly once, \emph{proper} if each snapshot has maximum degree one, and \emph{happy} if it is both simple and proper. Figure \ref{fig:graph_properties} provides examples of the different types of temporal graph. For simple graphs, we define the temporal assignment as $\lambda:E\to [\tau]$ for convenience. We also use these three terms to describe the temporal assignment of a graph with the corresponding property. Happy temporal graphs have the property that any nonstrict temporal path is also a strict temporal path. Part of the utility of happy temporal graphs is that hardness results on them generalize to the strict and nonstrict settings. We can now introduce our protagonist.

\begin{figure}[!ht]
     \centering
     \begin{subfigure}[b]{0.22\textwidth}
         \centering
         \includegraphics[width=\textwidth,page=5]{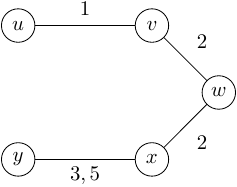}
         \caption{\footnotesize Not simple, not proper}
         \label{fig:notsimple_notproper}
     \end{subfigure}
     \hfill     \vrule      \hfill
     \begin{subfigure}[b]{0.22\textwidth}
         \centering
         \includegraphics[width=\textwidth,page=6]{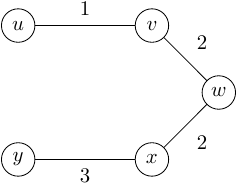}
         \caption{\footnotesize Simple, not proper}
         \label{fig:simple_notproper}
     \end{subfigure}
     \hfill     \vrule      \hfill
     \begin{subfigure}[b]{0.22\textwidth}
         \centering
         \includegraphics[width=\textwidth,page=7]{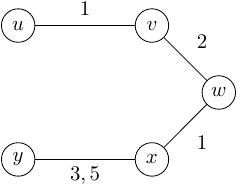}
         \caption{\footnotesize Not simple, proper}
         \label{fig:notsimple_proper}
     \end{subfigure}
     \hfill     \vrule      \hfill
     \begin{subfigure}[b]{0.22\textwidth}
         \centering
         \includegraphics[width=\textwidth,page=8]{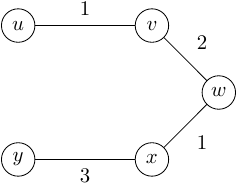}
         \caption{\footnotesize Happy (simple \& proper)}
         \label{fig:happy}
     \end{subfigure}
        \caption{Four small examples of (non)simple and (non)proper temporal graphs.}
        \label{fig:graph_properties}
\end{figure}

\smallskip

\begin{definition}[TaRDiS]
    In a temporal graph $\mathcal{G}$, a (strict/nonstrict) \emph{Temporal Reachability Dominating Set} (TaRDiS) is a set of vertices $S$ such that every vertex $v\in V(\mathcal{G})$ is temporally reachable from a vertex in $S$ by a (strict/nonstrict) temporal path.
\end{definition}

A minimum TaRDiS is a TaRDiS of fewest vertices in $\mathcal{G}$. We emphasize that, just as every strict temporal path is a nonstrict temporal path, every strict TaRDiS is a nonstrict TaRDiS. It is possible that the smallest nonstrict TaRDiS is strictly smaller than the smallest strict TaRDiS; for example, in Figure \ref{fig:simple_notproper}, there is a nonstrict TaRDiS of size $1$ (namely $\{u\}$) but every strict TaRDiS has size at least $2$.
We now formally define our problems.

\begin{framed}
    \noindent
    \textbf{\textsc{(Strict/Nonstrict) TaRDiS}}\\
    \emph{Input:} A temporal graph $\mathcal{G}=(V,E,\lambda)$ and an integer $k$.\\
    \emph{Question:} Does $\mathcal G$ admit a (strict/nonstrict) TaRDiS of size at most $k$?
\end{framed}

The restriction to \emph{happy} inputs $\mathcal G$ is a subproblem of both \textsc{Strict TaRDiS} and \textsc{Nonstrict TaRDiS}. 
\begin{framed}
    \noindent
    \textbf{\textsc{Happy TaRDiS}}\\
    \emph{Input:} A happy temporal graph $\mathcal{G}=(V,E,\lambda)$ and an integer $k$.\\
    \emph{Question:} Does $\mathcal G$ admit a TaRDiS of size at most $k$?
\end{framed}
\textsc{MaxMinTaRDiS} is an extension of our problem in which we look to find a temporal assignment such that no TaRDiS of cardinality less than $k$ exists.
As seen in Figure~\ref{fig:xkcd}, scheduling social events is natural combinatorial problem.
In a similar manner to how \textsc{Edge Colouring} corresponds straightforwardly to scheduling meetings between one pair of people at a time to avoid a scheduling conflict, the \textsc{MaxMinTaRDiS} problem can be thought of as scheduling interactions (potentially simultaneously) so that the risk of widespread contagion is limited. 

\begin{figure}[!ht]
    \centering
    \includegraphics[width=0.75\linewidth]{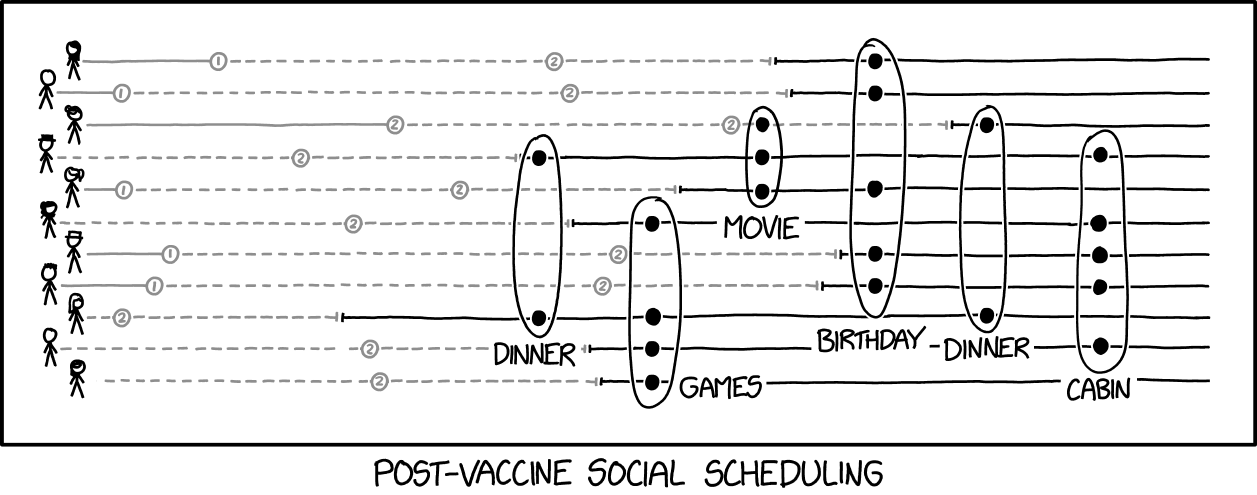}
    \caption{xkcd\#2450 \cite{xkcd} depicts the scheduling of social events to minimize the risk of contagion. }
    \label{fig:xkcd}
\end{figure}

% \todo[inline]{IAS: can you better motivate? when we wish to find a temporal assignment? the other 2 problems are fine: we have a temporal graph and ask questions of it}

% An organization needs to schedule many one-on-one meetings during flu season.

% There is a set of cattle trades which are due to occur over a week/month/whatever and the traders want to minimize the risk of widespread infection. Because biological spread we are talking about Strict TaRDiS. As we shall see, making all trades simultaneously is optimal, though this is highly impractical, which partly motivates our study of the Happy setting.

% Scheduling communications for the round in a distributed system supposing there is a (fast) adversary with a computer virus. 
% i.e. each round, the adversary can choose $k$ nodes to be contaminated, then comms happen and infection spreads instantly, and at the end of the round if at least one node is uninfected we win, the adversary loses (things reboot). 

% \todo[inline]{Attempts above}

% \todo[inline]{Below, allow assignment of multiple times.}
\begin{framed}
    \noindent
    \textbf{\textsc{(Strict/Nonstrict) MaxMinTaRDiS}}\\
    \emph{Input:} A static graph $H=(V,E)$ and integers $k$ and $\tau$.\\
    \emph{Question:} Does there exists a temporal assignment $\lambda:E\to 2^\tau \setminus \{\emptyset\}$ such that every strict/nonstrict TaRDiS admitted by $(H,\lambda)$ is of size at least $k$?
\end{framed}

Likewise, the variant of this problem in which the temporal assignment $\lambda$ is required to be happy is referred to as \textsc{Happy MaxMinTaRDiS}. Note this is not a subproblem of \textsc{Strict MaxMinTaRDiS} or \textsc{Nonstrict MaxMinTaRDiS}. We will later show that \textsc{Happy TaRDiS} generalizes \textsc{Edge Colouring} (and hence conflict-free scheduling in our earlier analogy).

\begin{framed}
    \noindent
    \textbf{\textsc{Happy MaxMinTaRDiS}}\\
    \emph{Input:} A static graph $H=(V,E)$ and integers $k$ and $\tau$.\\
    \emph{Question:} Does there exists a happy temporal assignment $\lambda:E\to [\tau]$ such that every TaRDiS admitted by $(H,\lambda)$ is of size at least $k$?
\end{framed}
In this work, by ``each variant'' we refer to the Strict, Nonstrict and Happy variants of the problems.

% \begin{definition}[(Weakly) locally earliest, (weakly) canonical TaRDiS]
%     In a temporal graph $\mathcal{G}$, an edge $(u,v)$ is \emph{locally earliest} if every other edge incident\footnote{If $e=(u,v)$ is an edge, $u$ is said to be incident to vertex $v$ and to edge $e$. Also, $u$ is incident to a set $S$ of vertices or edges if and only if it is incident to some element of $S$.}
%     to either $u$ or $v$ is at a time $t > \lambda(u,v)$. If the weaker constraint $t \geq \lambda(u,v)$ holds then we call the edge weakly locally earliest. A (weakly) \emph{canonical TaRDiS} consists exclusively of vertices which are incident to a (weakly) locally earliest edge.
% \end{definition}

\subsection{Our Contribution}

Our work highlights the complexity intrinsic to the dynamic behavior of spreading processes as soon as time-varying elements are incorporated into natural models. At a high level, we identify the minimum lifetime $\tau$ for which each problem is computationally hard. The existence of hardness results even with bounded lifetime justifies the need for parameters other than $\tau$ to obtain tractability results. We provide a fixed-parameter tractable (fpt)\footnote{We use fpt (lowercase) as a descriptor for algorithms witnessing the inclusion of a problem in the parameterized complexity class FPT. A full definition is given in Section~\ref{sec:param-tardis}.} algorithm for \textsc{TaRDiS} with parameters $\tau$ and the treewidth\footnote{Informally, the treewidth of a graph is a measure of its likeness to a tree. A formal definition is given in Section~\ref{sec:param-tardis}.} of the footprint graph $\text{tw}(\mathcal{G}_\downarrow)$, and show existence of such an fpt algorithm for \textsc{MaxMinTaRDiS} with parameters $\tau$, $\text{tw}(\mathcal{G}_\downarrow)$, and $k$.

\begin{table}[ht]
\resizebox{.9\linewidth}{!}{
\begin{tabular}{l}
\begin{tblr}{
  cell{1}{1} = {r=2}{},
  cell{1}{2} = {c=3}{c},
  cell{1}{5} = {c=3}{c},
  cell{3}{2} = {r=4}{},
  cell{3}{4} = {r=2}{},
  cell{3}{5} = {r=4}{},
  cell{3}{7} = {r=2}{},
  cell{4}{3} = {r=3}{},
  cell{5}{4} = {r=2}{},
  cell{5}{6} = {r=2}{},
  vlines,
  hline{1,3,7} = {-}{},
  hline{2} = {2-7}{},
  hline{4} = {1,3,6}{},
  hline{5} = {1,4,6-7}{},
  hline{6} = {1,7}{},
}
 \hspace{.5em}$\tau$&
 \textsc{TaRDiS}&&&
 \textsc{MaxMinTaRDiS}&&
\\
 & 
 Strict&
 Nonstrict&
 Happy&
 Strict&
 Nonstrict&
 Happy\\

\hspace{.5em}1&
 {NP-c \\(Cor.~\ref{cor:tardis-ds})}& % NP from domset
 {Linear \\(Lem.~\ref{lem:smalltau_tardis})}& %nonstrict connected components
 {Linear \\(Lem.~\ref{lem:smalltau_tardis})}& % happy paths and cycles
 {coNP-c \\(Cor.~\ref{cor:ds-strict-mmtardis})}& % coNP from domset
 {Linear \\(Lem.~\ref{lem:smalltau})}& % nonstrict maxmin tau=1
 {Linear \\ (Lem.~\ref{lem:smalltau})}\\ % happy maxmin paths and cycles

\hspace{.5em}2&
&
{NP-c \\(Thm.~\ref{thm:setcover_nonstrict})}& % set cover nonstrict tardis NP-h
&
&
{NP-c \\(Cor.~\ref{cor:d3is})}& % nonstrict tardis d3is
\\

\hspace{.5em}3&
&
&
{NP-c \\(Thm.~\ref{thm:happytardis})}& % happy tardis np-c
&
{$\in \Sigma_2^P$ \\(Lem.~\ref{lem:insigma})}& % known in Sigma2P
{$\Sigma_2^P-c$ \\(Thm.~\ref{thm:sigma})}\\ % happy mmtardis sigma2p

$\geq 4$&
&
&
&
&
&
 {$\in$ coNP-h $\cap~\Sigma_2^P$\\ (Cor.~\ref{cor:edgecol}, Lem.~\ref{lem:insigma})} % coNP-h from edgecol, in s2p from insigma 
\end{tblr}
\end{tabular}}
\caption{Computational complexity of our problems and its dependence on $\tau$.}\label{table:res}%
\end{table}

\def\arraystretch{1.5}
\begin{table}[!ht]
        % \footnotesize % uncomment for smaller text in table
		\centering
		\setlength{\tabcolsep}{1pt}
          \begin{tabular}{
				| m{.12857\linewidth}
				| >{\centering\arraybackslash}m{.12857\linewidth} 
				| >{\centering\arraybackslash}m{.12857\linewidth} 
				| >{\centering\arraybackslash}m{.12857\linewidth} 
                | >{\centering\arraybackslash}m{.12857\linewidth} 
				| >{\centering\arraybackslash}m{.12857\linewidth} 
				| >{\centering\arraybackslash}m{.12857\linewidth} |}
			\hline
			 
			& \multicolumn{3}{c|}{\textsc{TaRDiS}}
            & \multicolumn{3}{c|}{\textsc{MaxMinTaRDiS}}
            \\
			\hline
			Parameter  
			& Strict
            & Nonstrict
            & Happy
			& Strict
            & Nonstrict
            & Happy
            \\
			\hline
            $\Delta + \tau$  
			& para-NP-h  \newline (Cor.~\ref{cor:tardis-ds})
			& para-NP-h  \newline (Thm.~\ref{thm:setcover_nonstrict})
			& para-NP-h  \newline (Thm.~\ref{thm:happytardis})
			& para-NP-h  \newline (Cor.~\ref{cor:ds-strict-mmtardis})
			& para-NP-h  \newline (Cor.~\ref{cor:d3is})
			& para-NP-h  \newline (Thm.~\ref{thm:sigma})
            \\ 
            \hline
            $\mathrm{\#LEE}$  
			& n/a
            & \multicolumn{2}{c|}{FPT \newline (Lem.~\ref{lem:fpt_lee})}
			& \multicolumn{3}{c|}{n/a}
            \\ 
			\hline
            $k$  
			& W[2]-h  \newline (Cor.~\ref{cor:tardis-ds})
           & W[2]-h  \newline (Thm.~\ref{thm:setcover_nonstrict}) 
           & W[2]-h  \newline (Cor.~\ref{cor:setcover_happy}) 
			& co-W[2]-c  \newline (Cor.~\ref{cor:ds-strict-mmtardis}) % from co-DomSet
			& W[1]-h  \newline (Cor.~\ref{cor:d3is}) % from D3IS
			& para-NP-h  \newline (Cor.~\ref{cor:edgecol}) % edge coloring when k=0
            \\ 
            \hline
			$k+\tau$  
			& W[2]-h  \newline (Cor.~\ref{cor:tardis-ds})
			& W[2]-h  \newline (Thm.~\ref{thm:setcover_nonstrict})
			& FPT \newline (Lem.~\ref{lem:fpt_finite_lang_tardis})
			& co-W[2]-c  \newline (Cor.~\ref{cor:ds-strict-mmtardis}) % from co-DomSet
			& W[1]-h  \newline (Cor.~\ref{cor:d3is}) % from D3IS
			& para-NP-h  \newline (Cor.~\ref{cor:edgecol}) % edge coloring when k=0
            \\ 
            \hline
            $k+\tau+\Delta$  
			& FPT \newline (Lem.~\ref{lem:fpt_finite_lang_tardis})
			& W[2]-h  \newline (Thm.~\ref{thm:setcover_nonstrict})
			& FPT \newline (Lem.~\ref{lem:fpt_finite_lang_tardis})
			& FPT \newline (Lem.~\ref{lem:fpt_finite_lang})
			& ???
			& para-NP-h  \newline (Cor.~\ref{cor:edgecol})
            \\ 
            \hline
            $\mathrm{tw}\mathcal{G}_\downarrow + \tau$  
			& \multicolumn{3}{c|}{{FPT (Thm.~\ref{thm:tree-decomp-tardis})}}
			& \multicolumn{3}{c|}{???}
            \\ 
            \hline
            $\mathrm{tw}\mathcal{G}_\downarrow + \tau + k$
            & \multicolumn{3}{c|}{{FPT (Thm.~\ref{thm:tree-decomp-tardis})}}
            & \multicolumn{3}{c|}{{FPT (Thm.~\ref{thm:emso})}}
            \\ 
            \hline
		\end{tabular}\medskip
		\caption{Summary of our parameterized complexity results. Parameters are: maximum degree of the footprint graph $\Delta$, lifetime $\tau$, number of (weakly) locally earliest edges $\mathrm{\#LEE}$, input $k$, and treewidth of the footprint graph $\mathrm{tw}\mathcal{G}_\downarrow$. Problems which are W[1]-, (co-)W[2]-, or para-NP-hard are ones for which there presumably exists no fpt algorithm.}\label{table:fpt_results}
  
	\end{table}

Our results relating lifetime and computational complexity are highlighted in Table~\ref{table:res}, and our parameterized complexity results are summarized in Table~\ref{table:fpt_results}. For the case of happy temporal graphs, we exactly characterize the complexity of both \textsc{TaRDiS} and \textsc{MaxMinTaRDiS} with lifetime $\tau\leq 3$. Both problems are trivially solvable in linear time for $\tau \leq 2$. We show NP-completeness of \textsc{Happy TaRDiS} and $\Sigma_2^P$-completeness of \textsc{Happy MaxMinTaRDiS} when $\tau=3$ - even when restricted to planar inputs of bounded degree.

For \textsc{MaxMinTaRDiS}, membership of NP is nontrivial since the existence of a polynomial-time verifiable certificate is uncertain. 
Indeed, the $\Sigma_2^P$-completeness of \textsc{Happy MaxMinTaRDiS} indicates no such certificate exists in general unless the Polynomial Hierarchy collapses. 
Interestingly, we show equivalence\footnote{Our definition of equivalence can be found in Section \ref{sec:classic-maxmintardis}.} of \textsc{Nonstrict MaxMinTaRDiS} restricted to inputs where $\tau=2$ and \textsc{Distance-3 Independent Set}, which is NP-complete \cite{eto_distance-dindependent_2014_bugfree}, in Section~\ref{sec:NonStrictMTRDSA}. Given the uncertain membership of NP for the general problem, NP-completeness of this restriction of \textsc{Nonstrict MaxMinTaRDiS} indicates it is possibly \emph{easier} than the unrestricted problem.
    
Having shown $\tau$ and planarity of the footprint graph alone are insufficient for tractability, we give an algorithm which solves \textsc{TaRDiS} on trees in time polynomial in the number of time-edges $|\mathcal{E}|$ in the input graph. In addition, we give an algorithm for \textsc{TaRDiS} on nice tree decompositions \cite{cygan_parameterized_2015_bugfree}. This gives us tractability with respect to lifetime and treewidth of the footprint of the input graph combined. We also show existence of an algorithm for \textsc{MaxMinTaRDiS} which is fixed-parameter tractable with respect to $\tau$, $k$, and treewidth. This is achieved by applying Courcelle's theorem \cite{courcelle_1997_expression}.

\subsection{Related Work}\label{sec:RelatedWork}

Reachability and connectivity problems on temporal graphs have drawn significant interest in recent years. These have been studied in the context of network design \cite{akrida_temporal_2020,bilo_blackout-tolerant_2022_bugfree,casteigts_temporal_2021} and transport logistics \cite{fuchsle_delay-robust_2022_bugfree} (where maximizing connectivity and reachability at minimum cost is desired), and the study of epidemics \cite{braunstein_inference_2016,enright_deleting_2021,enright_assigning_2021,valdano_analytical_2015} and malware spread \cite{tang_applications_2013_bugfree}(where it is not).

In research on networks, broadcasting refers to transmission to every device. In a typical model, communication rather than computation is at a premium, and there is a single source in a graph which does not vary with time\cite{peleg_time_2007_bugfree}. Broadcasting-based questions deviating from this standard have been studied as well. Namely, there is extensive study of the complexity of computing optimal broadcasting schedules for one or several sources \cite{erlebach_np-hardness_2004,jakoby_complexity_1995_bugfree}, broadcasting in ad-hoc networks or time-varying graphs \cite{casteigts_time-varying_2011_bugfree}, and the choice of multiple sources (originators) for broadcasting in minimum time in a static graph \cite{chia_multiple_2007,grigoryan_problems_2013}. Ours is the first work to focus on the hardness of choosing multiple sources in a temporal graph to minimize the number of sources, in an offline setting.

One metric closely related to a temporal graph's vulnerability to contagion is its \emph{maximum reachability}; that is, the largest number $k$ such that some vertex reaches $k$ vertices in the graph. In Enright, Meeks and Skerman \cite{enright_assigning_2021} and Enright, Meeks, Mertzios and Zamaraev's \cite{enright_deleting_2021} works, the problems of deleting and reordering edges in order to minimize $k$ are shown to be NP-complete. Problems on temporal graphs are often more computationally complex than their static counterparts \cite{akrida_temporal_2019,bumpus_edge_2021,mertzios_computing_2023,klobas2023interference}. In such cases, efficient algorithms may still be obtained on restricted inputs. 
In work exploring disease spread through cattle, Enright and Meeks find that these real-world networks naturally have low treewidth \cite{enright_deleting_2018}. Treewidth and other structural parameters have been used with varying degrees of success to give tractability on some of these temporal problems \cite{akrida_temporal_2019,erlebach_temporal_2015,haag_feedback_2022,mertzios_computing_2023}. 
 
A powerful tool in parameterized complexity is Courcelle's theorem, which gives tractability on graphs of bounded treewidth for problems that can be represented in monadic second order logic \cite{courcelle_monadic_2012}.
Unfortunately, there are many temporal problems which remain intractable even when the underlying graph is very strongly restricted, for example when it is a path, a star, or a tree \cite{akrida_temporal_2019,bumpus_edge_2021,mertzios_computing_2023,klobas2023interference}, all of which have treewidth 1. When this is the case, it is sometimes sufficient to additionally bound the lifetime of the temporal graph in addition to its underlying structure. This motivates our use of treewidth in combination with lifetime as parameters for the study of our problems.  

The problems we study generalize two classical combinatorial graph problems, namely \textsc{Dominating Set} and \textsc{Distance-3 Independent Set (D3IS)}. For a static graph $G=(V,E)$, a dominating set is a set of vertices $S$ such that every vertex is either adjacent to a vertex in $S$ or is in $S$ itself, and a D3IS is a set of vertices $S$ such that no pair of vertices $u,v\in S$ has a common neighbour $w$ (so all pairs are at distance at least three). The corresponding decision problems ask, for an integer $k$, if there exists a dominating set (respectively D3IS) of size at most (resp. at least) $k$, and are W[2]- (resp. W[1]-) complete \cite{downey_fixed-parameter_1995, eto_distance-dindependent_2014_bugfree}; that is, even if $k$ is fixed it is unlikely there exists an algorithm solving the problem with running time $f(k)\cdot n^{O(1)}$. However, \textsc{Dominating Set} can be solved in polynomial time on graphs of bounded treewidth \cite{bjorklund_fourier_2006,cygan_parameterized_2015_bugfree}.

\textsc{TaRDiS} is exactly the problem of solving the directed variant of \textsc{Dominating Set} on a Reachability Graph \cite{bhadra_complexity_2003}. A Reachability Graph is also referred to as the transitive closure of journeys in a temporal graph and is shown to be efficiently computable by Whitbeck, Amorim, Conan, and Guillaume \cite{whitbeck_temporal_2012}. Temporal versions of dominating set and other classical covering problems have been well studied \cite{akrida_temporal_2020,casteigts_journey_2018_bugfree,verheije_algorithms_2021_bugfree}, however these interpretations do not allow a chosen vertex to dominate beyond its neighbours. Furthermore, many other problems looking to optimally assign times to the edges of a static graph have been studied \cite{kempe_connectivity_2000,enright_assigning_2021,klobas_mertzios_molter_spirakis_complexity_2022}. \textsc{TaRDiS} also generalises \textsc{Temporal Source}, which asks whether a single vertex can infect every other vertex in the graph, which is equivalent to the graph being a member of the class $\mathcal{J}^{1\forall}$ \cite{casteigts_journey_2018_bugfree} mentioned earlier. There has also been extensive research into modifying an input temporal graph (subject to some constraints) to achieve some desired reachability objective \cite{casteigts_journey_2018_bugfree, deligkas_eiben_skretas_reachability_2023, molter_temporal_2021}. Deligkas, Eiben, and Skretas \cite{deligkas_eiben_skretas_reachability_2023} have the objective of modifying $\lambda$ through a delaying operation so that \emph{each vertex} in a designated set reaches all vertices in the graph; that is $\bigcap_{u\in S} R_u = V(\mathcal{G})$. Molter, Renken and Zschoche \cite{molter_temporal_2021} consider both delay and deletion operations to modify $\lambda$ with the objective of minimizing the cardinality of the set $\bigcup_{u\in S}R_u$. Contrast these with our problem \textsc{MaxMinTaRDiS}, which has the objective of maximizing the minimum cardinality of {\bf any} set $S$ satisfying $\bigcup_{u\in S}R_u=V(\mathcal{G})$.

\subsection{Organization}

This paper is organised as follows. We begin with classical complexity results for \textsc{TaRDiS} in Section~\ref{sec:classic-tardis}. This consists of some preliminary observations pertaining to our problems in Section~\ref{sec:prelim-tardis}, which also convey some of the intuition for \textsc{TaRDiS} and provide tools for later technical results. Here we also consider temporal graphs with very small lifetime, showing that  \textsc{Nonstrict TaRDiS} and \textsc{Happy TaRDiS} are efficiently solvable in temporal graphs with lifetime $\tau=1$ and $\tau=2$ respectively. In Sections~\ref{sec:np-happy-lifetime-3} and~\ref{sec:np-nonstrict-lifetime-2} we show that these are the largest lifetime for which the problem is efficiently solvable in general. Finally, in Section~\ref{sec:tree-alg}, we show that \textsc{TaRDiS} is efficiently solvable when the footprint of the temporal graph is a tree. 

Section~\ref{sec:classic-maxmintardis} presents classical complexity results for our \textsc{MaxMinTaRDiS} problems. We again begin with some preliminary results in Section~\ref{sec:prelim-mmtardis} which provide the scaffolding for our later proofs as well as some easy bounds against which we compare our more technical results. This section allows us to draw the comparison between our problems and the classical problems \textsc{Edge Colouring} and \textsc{Dominating Set}. The remainder of this section includes our proofs of the contrasting results that \textsc{Nonstrict MaxMinTaRDiS} generalizes \textsc{D3IS} (and is in NP when restricted to $\tau=2$), and of the $\Sigma_2^P$-completeness of \textsc{Happy MaxMinTaRDiS}. 

Having established the hardness of both problems in the general case, and provided an algorithm for \textsc{TaRDiS} when the footprint graph is a tree, we turn to parameterized complexity in Sections~\ref{sec:param-tardis} and~\ref{sec:param-maxmintardis}. The majority of these sections focus on the structural parameter of treewidth of the footprint graph. We give an algorithm that solves \textsc{TaRDiS} on a nice tree decomposition of a graph in Section~\ref{sec:tree-decomp}. In Section~\ref{sec:courcelle} we show existence of such an algorithm for \textsc{MaxMinTaRDiS} by leveraging Courcelle's theorem.
Finally, we give some concluding remarks and future research directions in Section~\ref{sec:conclusion}.

\section{Classical complexity results for \textsc{TaRDiS}}\label{sec:classic-tardis}

In this section, we establish NP-completeness of each variant of \textsc{TaRDiS}, and characterize the maximum lifetime $\tau$ for which the problem tractable. We also give an algorithm solving the problem in polynomial time when the footprint graph $\mathcal G_\downarrow$ is a tree. A more general algorithm for footprint graphs of bounded treewidth is later given in Section~\ref{sec:param-tardis}.

\subsection{Containment in NP, useful tools, and small lifetime}\label{sec:prelim-tardis}

We begin this section by showing containment of \textsc{Strict, Nonstrict} and \textsc{Happy TaRDiS} in NP and introducing the notions of a (weakly) locally earliest edge and (weakly) canonical TaRDiS. We then show that there always exists a canonical minimum TaRDiS in happy temporal graphs. This allows us to reduce the number of cases we must consider when solving \textsc{Happy MaxMinTaRDiS} in Section \ref{sec:Sigma}. Finally, we consider the restrictions of the lifetime to $\tau=1$ and $\tau\le 2$, where \textsc{Nonstrict TaRDiS} and \textsc{Happy TaRDiS} respectively are easily shown to be efficiently solvable. We later show that these are the largest values of $\tau$ for which these problems are tractable.

\begin{lemma}\label{lem:innp}
    Each variant of \textsc{TaRDiS} is in NP.
\end{lemma}
\begin{proof}
The reachability set of a vertex can be computed in polynomial time by a modification of breadth-first search. Therefore, we can verify whether a set temporal reachability dominates a temporal graph in polynomial time.
\end{proof}

We now introduce the notion of a canonical TaRDiS. The following results will then allow us to make assumptions about the composition of a minimum TaRDiS when working with proper or happy temporal graphs, or on \textsc{Nonstrict TaRDiS}. 

\smallskip
\begin{definition}[(Weakly) locally earliest, (weakly) canonical TaRDiS]
    In a temporal graph $\mathcal{G}$, a time-edge $((u,v),t)$ is \emph{locally earliest} if every other time-edge incident
    to either $u$ or $v$ is at a time $t' > t$. If the weaker constraint $t' \geq t$ holds, then we call the time-edge weakly locally earliest. We say an edge $(u,v)$ is (weakly) locally earliest if, for some $t$, the time-edge $((u,v),t)$ is (weakly) locally earliest. A (weakly) \emph{canonical TaRDiS} consists exclusively of vertices which are incident to a (weakly) locally earliest edge.
\end{definition}

In Figure~\ref{fig:intro_canonical}, each of $\{b,d\}$, $\{b,c,d\}$ and $\{a,c,e\}$ is a TaRDiS, but only the former two are canonical. 

\begin{figure}[!ht]
     \centering
     \begin{subfigure}[t]{0.25\textwidth}
         \centering
    \includegraphics[height=2.5cm,page=11]{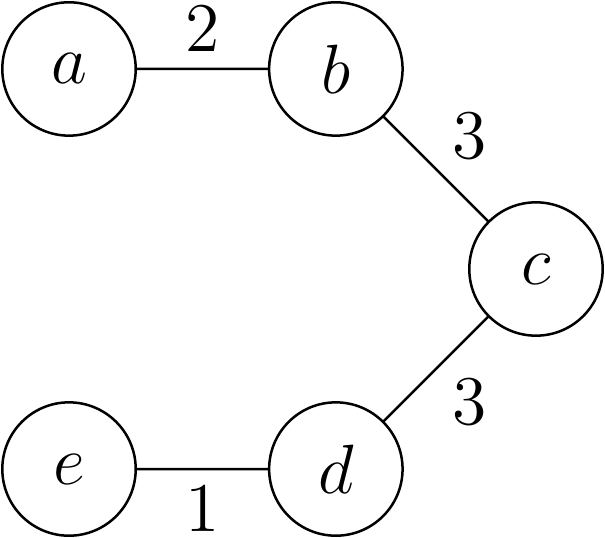}
         \caption{$R_a^{\leq} \neq R_a^<$ and $R_e^{\leq} \neq R_e^<$.}
         \label{fig:intro_strict_nonstrict}
     \end{subfigure}
     \hfill     \vrule      \hfill
     \begin{subfigure}[t]{0.25\textwidth}
         \centering
         \includegraphics[height=2.5cm,page=31]{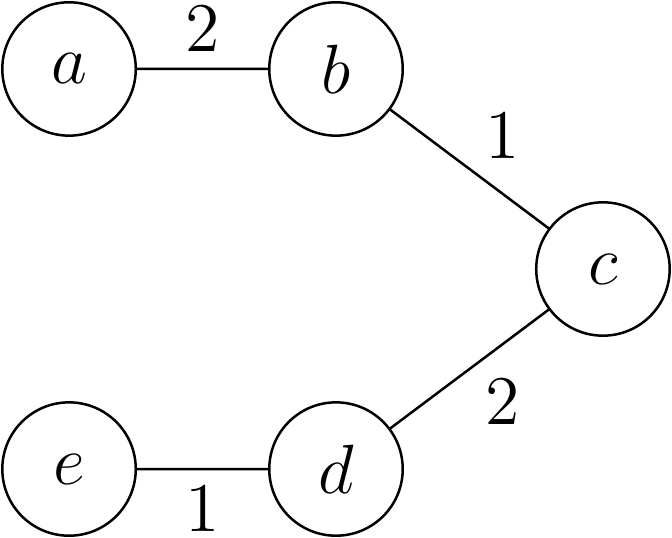}
         \caption{$\{a,e\}$ is not canonical.}
         \label{fig:intro_canonical}
     \end{subfigure}
     \hfill     \vrule      \hfill
     \begin{subfigure}[t]{0.4\textwidth}
         \centering
         \includegraphics[height=2.5cm,page=7]{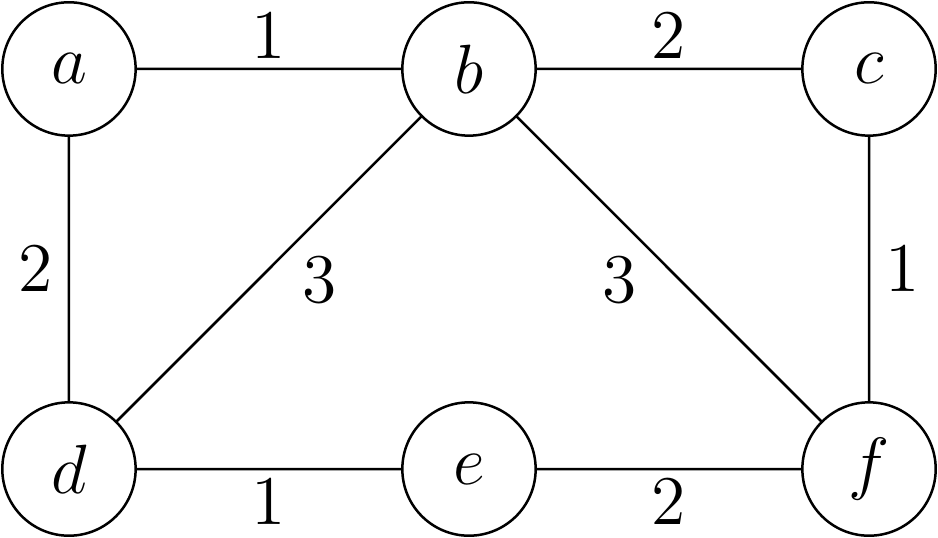}
         \caption{Every TaRDiS is canonical.}
         \label{fig:intro_d3iscounter}
     \end{subfigure}     
        \caption{Three temporal graphs, all admitting $\{a,e\}$ as a minimum TaRDiS.}
        \label{fig:canonical}
\end{figure}

\begin{lemma}\label{lem:canonical}
    In a temporal graph, there always exists a minimum weakly canonical nonstrict TaRDiS. In a proper temporal graph, there always exists a canonical TaRDiS.
\end{lemma}
\begin{proof}
We may assume without loss of generality that $\mathcal{G}_\downarrow$ is a connected graph. Note that a weakly locally earliest edge in a proper temporal graph is necessarily a locally earliest edge. Consequently, for the remainder of the proof, we focus on weakly locally earliest edges.

Given a TaRDiS $S$, we can find a weakly canonical TaRDiS $S'$ by replacing vertices not incident to weakly locally earliest edges by vertices that are. To see this, suppose $x\in S$ is not incident to a weakly locally earliest edge. Then there exist vertices $u,v\in V(\mathcal{G})$ and time-edges $((x,u),t),((u,v),t')$ such that $t \ge t'$. Suppose that $t$ is the earliest time that this is true for. Then, since all vertices reachable from $x$ must be temporally reachable from $v$ by a path appending the time-edges $((v,u),t'),((u,x),t)$, $R_x\subseteq R_{v}$. Therefore, we can swap $x$ for $v$ without losing domination of the whole graph. We perform this replacement repeatedly for every such $x$ until we obtain a set of vertices each incident to at least one weakly locally earliest edge. This is precisely a weakly canonical TaRDiS.
\end{proof}
This gives us a result that mirrors Lemma 54 from \cite{balev_temporally_2023}.

\smallskip
\begin{corollary}\label{cor:lee_balev}
  The number of weakly locally earliest edges upper-bounds the size of a minimum nonstrict TaRDiS.
\end{corollary}
\smallskip

We now show that the classical problem \textsc{Dominating Set} is a special case of \textsc{Strict TaRDiS}. \textsc{Dominating Set} is known to be NP-complete, even on planar graphs with maximum degree $\Delta=3$ (hence para-NP-hard with respect to $\Delta$) \cite{garey_simplified_1974}, and W[2]-hard with respect to the size of the dominating set \cite{downey_fixed-parameter_1995}.

\smallskip
\begin{lemma}\label{lem:tar2ds}
    For all positive integers $\tau$ and instances $(G,k)$ of \textsc{Dominating Set}, a temporal graph $\mathcal{G}$ of lifetime $\tau$ can be found in linear time, where $(\mathcal{G},k)$ is a yes-instance of \textsc{Strict TaRDiS} if and only if $(G,k)$ is a yes-instance of \textsc{Dominating Set}.
\end{lemma}

\begin{proof}
    The construction of $\mathcal{G}$ can be seen in Figure \ref{fig:trivial_domset_reduction}. 
    We append a path $P$ of length $\tau-1$ to an arbitrary vertex $v$ in $G$. The temporal assignment of $\mathcal{G}$ is defined as follows: all edges in $E(G)$ are assigned time 1, and the edges in $P$ are assigned times in ascending order from $v$. It is clear that there exists a dominating set of cardinality $k$ in $G$ if and only if there exists a strict TaRDiS of cardinality $k$ in $\mathcal{G}$.
\end{proof}
We note that, in our construction, $k$ is unchanged and the maximum degree of $\mathcal{G}_{\downarrow}$ is bounded by $\Delta(G)+1$; giving us the following corollary.

\smallskip
\begin{corollary}\label{cor:tardis-ds}
    \textsc{Strict TaRDiS} is NP-complete, W[2]-hard with respect to $k$, and para-NP-hard with respect to $\Delta+\tau$, where $\Delta$ is the maximum degree of the footprint graph and $\tau$ is the lifetime of the temporal graph.
\end{corollary}

\begin{figure}[!ht]
    \centering
    \includegraphics[page=45, width=\linewidth]{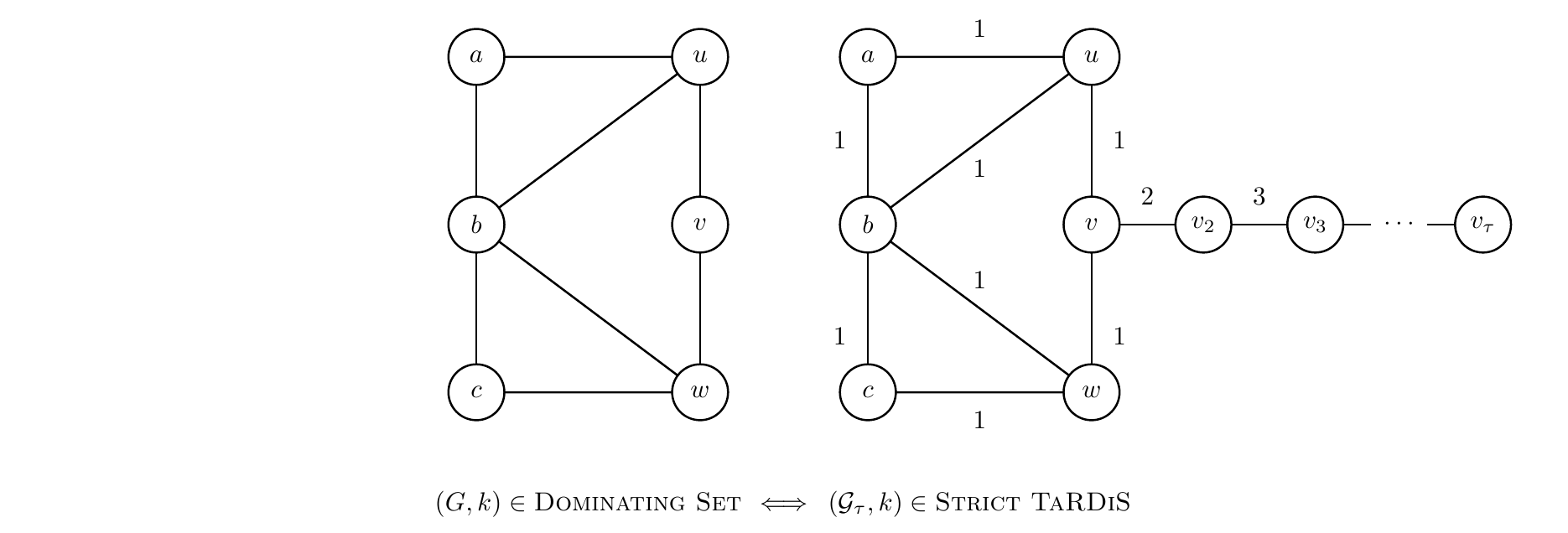}
    \caption{Illustration of the construction used to reduce from \textsc{Dominating Set} to \textsc{Strict TaRDiS} with arbitrary lifetime $\tau$.}
    \label{fig:trivial_domset_reduction}
\end{figure}

We now consider the problems \textsc{Nonstrict TaRDiS} and \textsc{Happy TaRDiS} with very restricted lifetimes.

\smallskip
\begin{lemma}\label{lem:smalltau_tardis}
    \textsc{Nonstrict TaRDiS} can be computed in linear time when $\tau=1$. When $\tau\le 2$, \textsc{Happy TaRDiS} is solvable in linear time. 
\end{lemma}

\begin{proof}
If every edge is active at the same time, every vertex can be temporally reached by a nonstrict path from any vertex in the same connected component. Therefore, if $\tau=1$, we find a minimum TaRDiS by choosing exactly one vertex in every connected component of the graph.

If a temporal graph is happy and the lifetime of the graph is at most 2, the footprint of the graph must be 2-edge colourable, and so consist of only paths or even cycles. Any TaRDiS must include every vertex which has no neighbours. Then, for any paths, we choose a leaf and add its neighbour to the TaRDiS. Following this, we work along the path adding an endpoint of an edge if it is active before the edge previously considered. We use a similar method on even cycles by simply picking an arbitrary vertex $v$ incident to an edge in $E_1(\mathcal{G})$ to be in the TaRDiS. From here, we find the subgraph induced by removing the vertices reachable from $v$, giving us a path to which we can apply the previous procedure.
\end{proof}

\subsection{NP-completeness of \textsc{Happy TaRDiS} with lifetime 3}\label{sec:np-happy-lifetime-3}

Above, we establish that \textsc{Happy TaRDiS} can trivially be solved in linear time when the input has lifetime $\tau \leq 2$ by Lemma~\ref{lem:smalltau_tardis}. Here we show that the problem immediately becomes NP-complete for inputs where $\tau=3$, even when $\mathcal{G}_{\downarrow}$ is planar.

We give a reduction from the problem \textsc{Planar exactly 3–bounded 3–SAT}, which asks whether the input Boolean formula $\phi$ admits a satisfying assignment. We are guaranteed that $\phi$ is a planar formula in 3-CNF (a disjunction of clauses each containing at most 3 literals) with each variable appearing exactly thrice and each literal at most twice. This problem is shown to be NP-hard by Tippenhauer and Muzler \cite{tippenhauer2016planar}.

\subsubsection{Construction} Let $\phi$ be an instance \textsc{Planar exactly 3–bounded 3–SAT} consisting of clauses $\{c_1, \ldots, c_m\}$ over variables $X=\{x_1, \ldots, x_n\}$. 
Let $l_{2i}$ (resp. $l_{2i-1}$) denote the positive (resp. negative) literal for variable $x_i$, and $l_{2n+1}$ denote the special literal $\bot$, which is always \texttt{False}.
We rewrite every 2-clause of $\phi$ to include the special literal $\bot$ (e.g. the clause $(x_i\lor \lnot x_j)$ becomes $(x_i\lor \lnot x_j \lor \bot)$ then $(l_{2i} \lor l_{2j-1} \lor l_{2n+1})$). Let $\phi '$ denote the new formula obtained by doing this. Note that $\phi'$ admits a satisfying assignment (in which $\bot$ evaluates to \texttt{False}) if and only if $\phi$ admits a satisfying assignment. 

\begin{figure}[!ht]
    \centering
    \includegraphics[page=9, width =\textwidth]{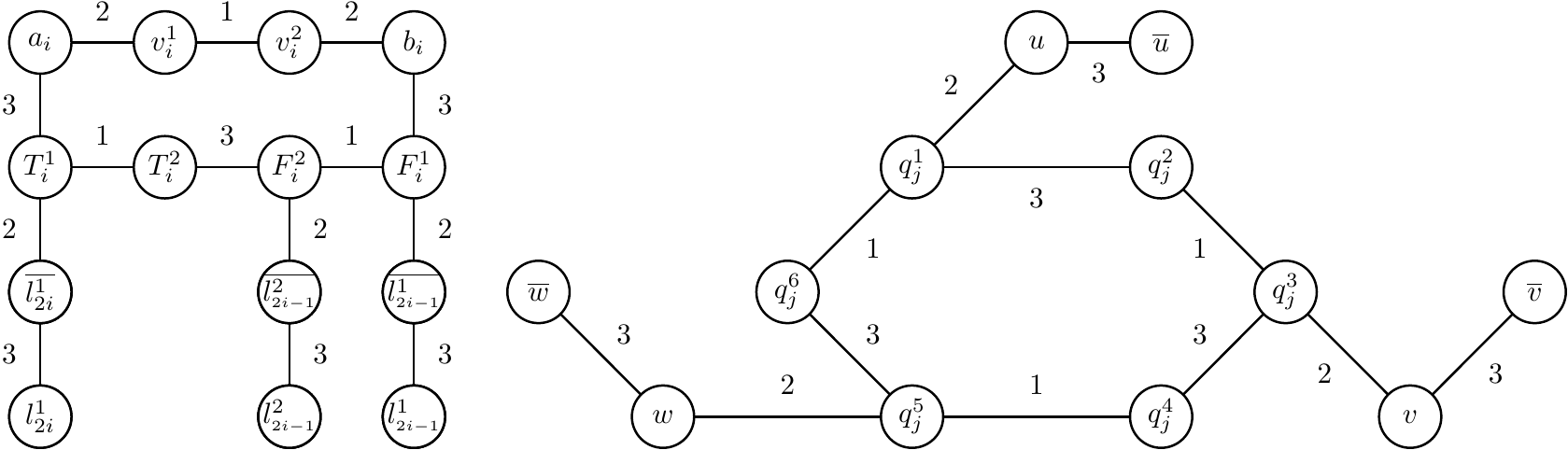}
    \caption{Gadgets and adjacent literal vertices in the reduction from \textsc{Planar exactly 3–bounded 3–SAT} to \textsc{Happy TaRDiS}. Left: the variable gadget for $x_i$, which appears twice negatively. Right: the clause gadget for the clause $c_j$.}
    \label{fig:happy_tardis}
\end{figure}

In this construction $E(\mathcal{G})=E_1\cup E_2 \cup E_3$; $\lambda$ is implicitly defined on this basis. We introduce \emph{literal vertices} as follows. Iterate over $\phi'$: at the $a$th appearance of each literal $l_i$, create two new vertices $l_i^a$ and $\overline{l_i^a}$, and let $(l_i^a, \overline{l_i^a})\in E_3$. We refer to the set of all literal vertices as $L$. We say the vertex $l_i^a$ \emph{belongs} to the clause $c_j$ in which $l_i$ appears for the $a$th time.

For each variable $x_i$, we introduce vertices $V_i=\{T_i^1,T_i^2,F_i^1,F_i^2,v_i^1,v_i^2,a_i,b_i\}$, connected in a cycle labeled as shown in Figure~\ref{fig:happy_tardis}. Specifically, we have $(T_i^1,T_i^2),(F_i^1,F_i^2),(v_i^1,v_i^2)\in E_1$, $(a_i,v_i^1),(b_i,v_i^2) \in E_2$, and $(a_i,T_i^1),(b_i,F_i^1) \in E_3$. Then, for each positive (resp. negative) literal vertex $\overline{l_{2i}^a}$ (resp. $\overline{l_{2i-1}^a}$), we add $(\overline{l_{2i}^a}, T_i^a)$ to $E_2$.

Every three literals vertices $u,v,w$ belonging to the same clause $c_j$ are connected with the \emph{clause gadget} as shown in Figure~\ref{fig:happy_tardis}. Namely, we introduce the vertices $Q_j=\{q_j^1,\ldots,q_j^6\}$, with $(q_j^1,q_j^2),(q_j^3,q_j^4),(q_j^5,q_j^6)\in E_3$, $(q_j^2,q_j^3),(q_j^4,q_j^5),(q_j^6,q_j^1)\in E_1$, and $(q_j^1,u),(q_j^3,v),(q_j^5,w)\in E_2$. If, for example, clause $c_j=(l_{5}\lor l_{17} \lor l_{20})$ is the first appearance of $l_{5}$ and $l_{20}$ and the second appearance of $l_{17}$ in $\Phi$, then $u=l_{5}^1$, $v=l_{17}^2$ and $w=l_{20}^1$.

Lastly, we let $k=2m + 2n$. This concludes our construction of the \textsc{Happy TaRDiS} instance $(\mathcal{G}, k)$.

\subsubsection{Properties of the construction} We will show that $\mathcal G$ admits a TaRDiS $S$ of size $k$ if $\phi$ is satisfiable, and that any TaRDiS is of size at least $k+1$ otherwise. It is worth noting that our construction preserves planarity of $\phi$. That is, $\mathcal G_\downarrow$ is a planar graph.

We now show that the existence of a TaRDiS of size at most $k$ implies that of a satisfying assignment.

\smallskip
\begin{lemma}\label{lem:canonical_happy_tardis}
    Any \emph{canonical} TaRDiS of $\mathcal G$ includes at least 2 vertices from each variable gadget and at least 2 vertices from each clause gadget, and so has size at least $k$.
\end{lemma}
\begin{proof}
    Recall that a \emph{canonical} TaRDiS $S$ by definition is incident only to locally earliest edges. In the case of $\mathcal G$, none of the \emph{literal vertices} can belong to $S$, i.e. $S\subseteq V \setminus L$ where $L$ is the set of literal vertices. 
    For each variable $x_i$ and clause $c_j$, no vertex from $V_i$ reaches any vertex from $Q_j$, and vice versa. Note that, for all $j\in [m]$, every vertex in $Q_j$ reaches exactly 4 other vertices in $Q_j$, for example $R_{q_j^3}=\{q_j^1,q_j^2,q_j^3,q_j^4\}$. Thus, $|S\cap Q_j|\geq 2$, else $S$ would not reach every vertex in $Q_j$. Similarly, for all $i\in[n]$, every vertex in $V_i$ reaches at most 6 other variable gadget vertices, so $|S\cap V_i|\geq 2$. 
  \end{proof}

\begin{lemma}\label{lem:if_happy_tardis_k_then_sat}
If $\mathcal G$ admits a TaRDiS $S$ of size $k$ then $\phi$ is satisfiable.
\end{lemma}

\begin{proof}
Suppose $\mathcal G$ admits a TaRDiS $S$ of size at most $k$. We first apply Lemma~\ref{lem:canonical}, which allows us to assume that $S$ is canonical. Then, by Lemma~\ref{lem:canonical_happy_tardis}, we have that $S$ is of size \emph{exactly} $k$, and that for all $i$ and $j$, $|S\cap V_i|= 2$ and $|S\cap Q_j|=2$.

We show that a satisfying assignment $X$ to the variables of $\phi$ can be obtained from $S$. For each variable $x_i$, assign $x_i=\texttt{True}$ if $T_i^1$ or $T_i^2\in S$, and $x_i=\texttt{False}$ otherwise. Note that, since $S$ is a canonical TaRDiS, $a_i$ and $b_i$ are not in $S$. Hence $F_i^1\in S$ or $F_i^2\in S$ if and only if $x_i = \texttt{False}$ in our assignment. Suppose for contradiction there is some clause $c_j$ which is not satisfied. Then, under this assumption, the gadget $c_j$ is incident to three \emph{literal vertices} $\{u,v,w\}$, none of which are reached from their respective \emph{variable gadgets}. Since $S$ is a TaRDiS, $q_j^1 $ or $q_j^6$ is in $S$ (as $u$ must be reached from within $Q_j=\{q_j^1,\ldots,q_j^6\}$), and likewise $q_j^2$ or $q_j^3$ and $q_j^4$ or $q_j^5$ are in $S$ (as $v$ and $w$ respectively must be reached from within $Q_j$). Hence $|Q_j\cap S|\geq 3$. By applying Lemma~\ref{lem:canonical_happy_tardis} we get that $|S|\geq 2m + 2n + 1 = k+1$, a contradiction. Hence $X$ must satisfy $\phi$.
  \end{proof}

\begin{lemma}\label{lem:if_sat_then_happy_tardis_k}
    If $\phi$ is satisfiable then $\mathcal G$ admits a TaRDiS $S$ of size $k$.
\end{lemma}
\begin{proof}
    Given a satisfying assignment to $\phi$, we construct $S$ as follows. We start with $S=\cup_i\{v_i^2\}$. Then if $x_i=\texttt{True}$ (respectively $x_i=\texttt{False}$) under the assignment, add $T_i^1$ to $S$ (resp. $F_i^1\in S$). Now $S$ has size $2n$ and a vertex in $S$ reaches the literal vertices for every literal set to $\texttt{True}$ in the assignment, but none of the clause gadget vertices.

    For each clause $c_j$, there is some literal in $c_j$ which is assigned \texttt{True}, and hence at least one literal vertex incident to $Q_j$ is reached from the corresponding variable gadget. Denote this literal vertex $u$, and the other literal vertices incident to $Q_j$ (which may or may not be reached from their respective variable gadgets) $v$ and $w$, respectively. We add the neighbours of $v,w$ in $Q_j$ to $S$. That is, $S=S\cup ((N[v]\cup N[w])\cap Q_j)$. Now $S$ has size $2n+2m$ and reaches: all variable gadget vertices; all clause gadget vertices; and all literal vertices (since every literal vertex is incident to a clause gadget). Hence $S$ is a TaRDiS of size exactly $k$.
  \end{proof}

\begin{theorem}\label{thm:happytardis}
    \textsc{Happy TaRDiS} is NP-complete, even restricted to instances where the footprint of the temporal graph is planar and its lifetime is 3.
\end{theorem}

\begin{proof}
    We have membership of NP from Lemma~\ref{lem:innp}. The construction above produces an instance $(\mathcal G, k)$ of \textsc{Happy TaRDiS} from an instance $\phi$ of \textsc{Planar Exactly 3–Bounded 3–SAT} in polynomial time. 
    Combining Lemmas~\ref{lem:if_happy_tardis_k_then_sat} and~\ref{lem:if_sat_then_happy_tardis_k} shows that $\mathcal G$ admits a TaRDiS of size at most $k$ if and only if $\phi$ is satisfiable.
  \end{proof}

This result generalises to both \textsc{Strict} and \textsc{Nonstrict TaRDiS}, giving us NP-completeness of all variants with bounded lifetime and a planar footprint.

\subsection{NP-completeness of \textsc{Nonstrict TaRDiS} with lifetime 2}\label{sec:np-nonstrict-lifetime-2}
We show hardness of \textsc{Nonstrict TaRDiS} by reducing from \textsc{Set Cover}, which is known to be NP-complete \cite{book_richard_1975} and W[2]-hard with respect to the parameter $k$ \cite{cygan_parameterized_2015_bugfree}. The \textsc{Set Cover} problem is defined as follows.
\begin{framed}
        \noindent
    \textbf{\textsc{Set Cover}}\\
    \emph{Input:} Universe $\mathcal{U}=\{x_1,\ldots x_n\}$, a family $\mathcal{S}=\{s_i|s_i\subseteq \mathcal{U}\}$ of subsets in $\mathcal{U}$ and an integer $k$.\\
    \emph{Question:} Is there a set $\{s_{i\in I}\}$ with at most $k$ elements such that $\cup_{i\in I}s_i=\mathcal{U}$?
\end{framed}
Given an instance of \textsc{Set Cover}, we build a temporal graph $\mathcal{G}$ with vertex set $V(\mathcal{G})= \mathcal{U} \cup \mathcal{S} \cup \{a_i^j|~\exists x_i\in \mathcal{U}, s_j\in \mathcal{S}: x_i \in s_j\}$. To this vertex set we add the edges:
\begin{itemize}
    \item connecting $s_j$ in a path at time 1 to all $a_i^j$ for all $i$ such that $a^j_i$ exists;
    \item connecting $x_i$ in a path at time 2 to all $a_i^j$ for all $j$ such that $a^j_i$ exists;
    \item $(s_j,s_{j+1})$ at time 2 for all $j\in[1,m-1]$.
\end{itemize}
A sketch of $\mathcal{G}$ can be found in Figure~\ref{fig:setcover2nonstrict}. Note that the maximum degree of $\mathcal{G}$ is 4. We keep the same value of $k$ in our construction; that is, we construct an instance $(\mathcal G, k)$ of \textsc{Nonstrict TaRDiS} from the instance $(\mathcal U, \mathcal S, k)$ of \textsc{Set Cover}.

\smallskip
\begin{lemma}\label{lem:SCtardis}
    Any TaRDiS in $\mathcal{G}$ can be reduced to a TaRDiS of the same or smaller size in $V(\mathcal{G})\cap\mathcal{S}$.
\end{lemma}
\begin{proof}
    Suppose there is a TaRDiS $S$ in $\mathcal{G}$ containing a vertex $v\in V(\mathcal{G})\setminus \mathcal{S}$. We have two cases:
    \begin{enumerate}
        \item $v=x_i$ for some $1\leq i\leq n$, or
        \item $v=a^j_i$ for some $1\leq i\leq n$ and $1\leq j\leq m$.
    \end{enumerate}
In the first case, we can swap $x_i$ for any neighbour (or delete $x_i$ if $N[x_i] \subset S$). This must be a vertex $a_i^j$ for some $1\leq j\leq m$. Since the reachability set of $a^j_i$ is a strict superset of that of $x_i$, $S$ remains a TaRDiS when $x_i$ is replaced by its neighbour. This leaves us with case 2. The reachability sets of $a^j_i$ are equal to the reachability set of $s_j$ for all $a_i^j \in V(\mathcal G) \setminus (\mathcal{U} \cup \mathcal{S})$. Therefore, we can replace vertices in $S$ with those corresponding to sets in $\mathcal{S}$ without changing the reachability set of the TaRDiS.
  \end{proof}

\begin{lemma}
    The temporal graph $\mathcal{G}$ admits a TaRDiS of cardinality $k$ if and only if $(\mathcal U, \mathcal S, k)$ is a yes-instance of \textsc{Set Cover}.
\end{lemma}
\begin{proof}
    Suppose there is a set $\{s_{j\in I}\}$ of cardinality $k$ which is a set cover of $\mathcal{U}$. Then we claim the same set $S$ is a TaRDiS in $\mathcal{G}$. If an element $x_i$ is in the set $s_j$, then by construction of $\mathcal{G}$, there is a path from $s_j$ to $x_i$ via $a_i^j$. Furthermore, all vertices $s_{j'}$ for $1\leq j'\leq m$ are temporally reachable from any $s_j$. Since all $x_i$ appear in at least one of $\{s_{j\in I}\}$, we know that all $x_i$ and all $s_j$ are temporal reachability dominated by $S=\{s_{j\in I}\}$. What remains to check is that all $a_i^j$ are also reached by a vertex in $S$. Each $a_i^j$ is connected to $x_i$ at time 2. Therefore, since all $x_i$ are reached at time 2 from $S$, each $a_i^j$ must also be temporally reachable from $S$.

    Now suppose that there is a TaRDiS $S$ of cardinality $k$ of $\mathcal{G}$. By Lemma~\ref{lem:SCtardis}, we can assume that $S\subseteq \mathcal{S}$. Observe that there is a nonstrict temporal path from a vertex $s_j$ reaches a vertex $x_i$ if and only if $x_i\in s_j$. Hence, a vertex $x_i$ is temporally reachable from some vertex in $S$ if and only if a set $s_j$ containing $x_i$ is in $S$. Since every $x_i$ is reachable from some vertex in $S$, then, $S$ is a set cover. 
  \end{proof}

\begin{figure}[ht]
    \centering
    \includegraphics[page=10, width=0.6\textwidth, height=0.6\textwidth]{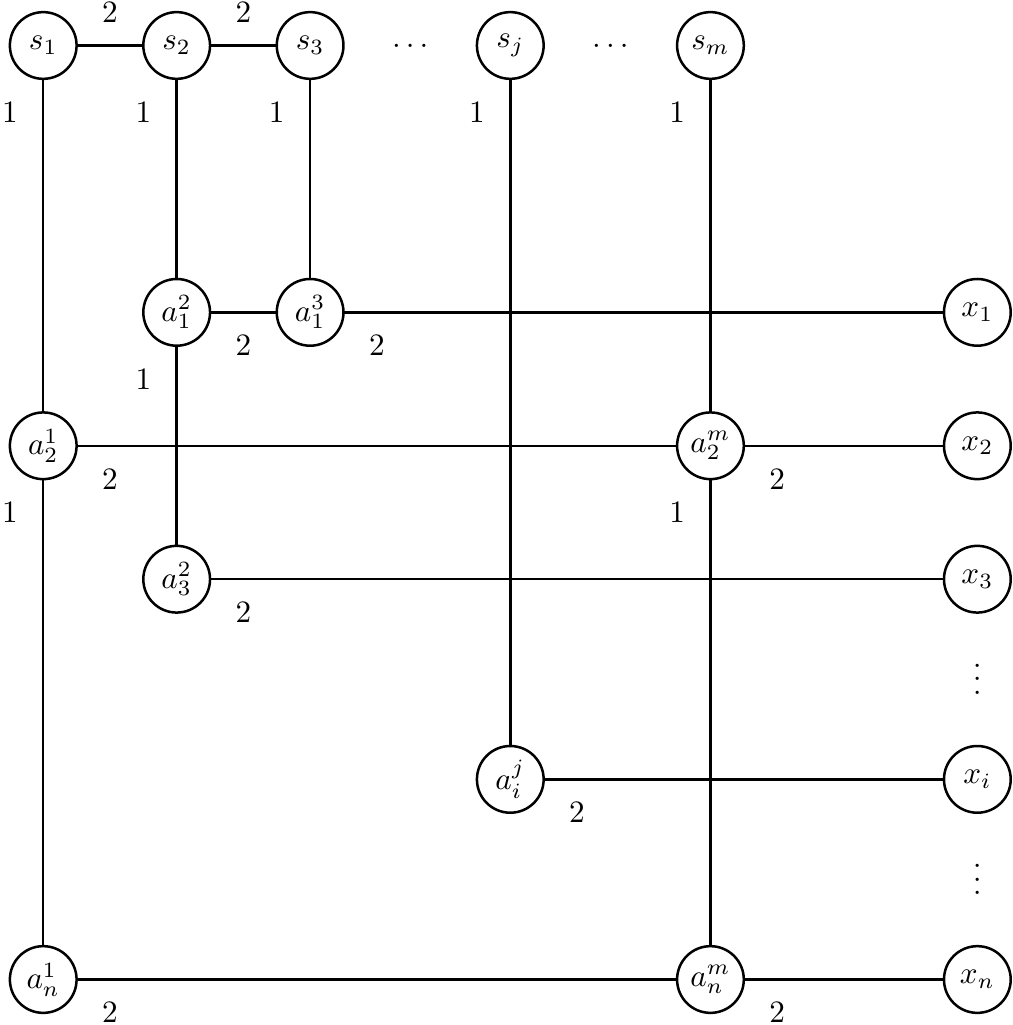}
    \caption{Sketch for the reduction from \textsc{Set Cover} to \textsc{Nonstrict TaRDiS}}
    \label{fig:setcover2nonstrict}
\end{figure}

We note that since we use the same input $k$ for both problems, \textsc{Nonstrict TaRDiS} is $W[2]$-hard with respect to $k$. This proves the following theorem.

\smallskip
\begin{theorem}\label{thm:setcover_nonstrict}
    For every lifetime $\tau\geq 2$ \textsc{Nonstrict TaRDiS} is NP-complete and W[2]-hard with respect to $k$. This holds even on graphs with maximum degree 4. 
\end{theorem}

We now extend our construction to obtain W[2]-hardness of \textsc{Happy TaRDiS}. Consider the following construction for a happy temporal graph $\mathcal{G'}$. Begin with the temporal graph $\mathcal{G}$ as constructed earlier in Figure~\ref{fig:setcover2nonstrict}. Then, for each $i$, add edges such that the set $\{s_i\}\cup \bigcup_{j}a^i_j$ forms a clique. For each $j$, we also add edges such that the set $\{x_j\}\cup \bigcup_{i}a^i_j$ forms a clique. Then we order all time-edges arbitrarily such that all edges in each clique containing edges formerly in $E_1$ are active before the edges in each clique containing all time-edges formerly in $E_2$. This concludes the description of $\mathcal{G}'$. Observe that $\mathcal{G}'$ is a happy temporal graph, since each snapshot contains just one edge.

\smallskip
\begin{lemma}
      The temporal graph $\mathcal{G'}$ admits a TaRDiS of cardinality $k$ if and only if $(\mathcal G, k)$ is a yes-instance of \textsc{Happy TaRDiS}.
\end{lemma}
\begin{proof}
    We show that a vertex $u$ temporally reaches a vertex $v$ in $\mathcal{G}$ if and only if $u$ also reaches $v$ in $\mathcal{G}'$. Thus, proving that $\mathcal{G'}$ admits a TaRDiS of cardinality $k$ if and only if $\mathcal G$ admits a TaRDiS of cardinality $k$.

    Suppose we have a path $p$ from a vertex $u$ to a vertex $v$ in $\mathcal{G}$. We have two cases; namely, the edges in the path are assigned the same time in $\mathcal{G}$ or they are not. In the former, there must still be a (one-edge) path from $u$ to $v$ in $\mathcal{G}'$ since we have added edges such that all connected components in either snapshot of $\mathcal{G}$ form a clique. In the latter, we can break $p$ into a path at time one $p_1$ and a path $p_2$ at time two. These paths can each be replaced by an edge as in the case where $p$ consists of edges which occur at the same time. Therefore, for every path in $\mathcal{G}$, there exists a path with the same starting and terminal vertex in $\mathcal{G'}$.

    Now suppose we have a path $p'$ in $\mathcal{G}'$ from $u'$ to $v'$. Since we have cliques for each connected component in the snapshots of $\mathcal{G}$, the path consisting of fewest edges in $\mathcal{G}'$ from $u'$ to $v'$ must consist of at most 2 edges. Suppose without loss of generality that $p'$ is the shortest path from $u'$ to $v'$. If $p'$ consists of a single edge, then $u'$ and $v'$ must be in the same connected component of $\mathcal{G}$, thus there exists a path from $u'$ to $v'$ in $\mathcal{G}$. Now suppose that $p'$ consists of two edges. We know that $v'$ and $u'$ are not adjacent in $\mathcal{G}'$, else $p'$ would not be the shortest path from $u'$ to $v'$. Therefore, $u'$ and $v'$ must be members of two different cliques in $\mathcal{G}'$. Let $w'$ be the vertex adjacent to both $u'$ and $v'$ which $p'$ traverses. Then, since $u'$ and $w'$ share an edge which is earlier than the edge $(w',v')$, they must be in the same connected component of $\mathcal{G}_1$. Similarly, $w'$ and $v'$ must be in the same connected component of $\mathcal{G}_2$. Therefore, there is a path from $u'$ to $v'$ in $\mathcal{G}$. Hence, a vertex $u$ temporally reaches a vertex $v$ in $\mathcal{G}$ if and only if $u$ also reaches $v$ in $\mathcal{G}'$.
\end{proof}

\begin{corollary}\label{cor:setcover_happy}
    \textsc{Happy TaRDiS} is NP-complete and W[2]-hard with respect to $k$.
\end{corollary}

\subsection{Algorithm for \textsc{TaRDiS} on Trees}\label{sec:tree-alg}
In this section, we show that each variant of \textsc{TaRDiS} can be solved in polynomial time on inputs which are trees, even when $k$ and lifetime are unbounded. We introduce some preliminary notions needed for the algorithm.

In this section, we deal with \emph{non-simple} temporal graphs; each edge $e\in E(\mathcal{G})$ may appear several times.  
We define the functions $\lambda_{\min}:E\to [\tau]$ and $\lambda_{\max}:E\to [\tau]$ which return the earliest and latest appearance of an edge, respectively. 

In the following we refer to \emph{rooted trees}. A rooted tree is a tree $T=(V,E)$ with a special vertex $r$; the \emph{root}. A vertex $w$ is a \emph{descendant} of a vertex $v$ if $v$ lies on the unique path from $w$ to $r$. We refer to the vertices neighbouring $v$ which are descendants of $v$ as its \emph{children}. If a vertex $w$ is a child of $v$, then $v$ is its \emph{parent}. Similarly, the parent of a parent vertex is referred to as a \emph{grandparent}. The subgraph induced by a vertex $v$ and its descendants is referred to as the \emph{subtree rooted at $v$}.
The input of the algorithm is a \emph{rooted temporal tree} $\mathcal{G}$. By this, we mean a temporal graph whose footprint $\mathcal{G}_{\downarrow}$ is a rooted tree. We define the function $C(p)$ to return the set of children for a parent vertex $p$. 

\subsubsection{Intuition}
The algorithm operates as follows. We initialise a counter $k$ and two empty sets. The first set $S$ will become a TaRDiS if a TaRDiS of size $k$ exists. The second set $M$ is a set of vertices which are temporally reachable from those in $S$ or will be temporally reachable from a later vertex added to $S$. Roughly, these are the vertices we do not need to worry about. We refer to them as \emph{marked}. Vertices not in $M$ will be referred to as \emph{unmarked}. We refer to the distance of a vertex in the tree from the root as its \emph{depth}.

We work from the leaves of $\mathcal{G}$ to the root $r$. At each iteration, we either mark vertices or reduce the number of appearances of an edge. We continue until the unmarked graph is a star or empty. The star graph is the complete bipartite graph $K_{1,k}$. Its central vertex is a vertex incident to every edge the graph if $k\geq 1$, and is the sole vertex in the graph otherwise. 

If a vertex $v$ with parent $p$ and grandparent $g$ is reachable from its grandparent using the earliest time on the edge $(v,p)$ and the latest appearance of $(p,g)$, it will be reached by a vertex in $S$ if its parent is. Therefore, we can mark that vertex. If a vertex $v$ is not marked and also not temporally reachable from its grandparent, we must add a vertex in $N[v]$ to $S$. We solve \textsc{Strict TaRDiS} or \textsc{Nonstrict TaRDiS} using the variable $s$ which depends on the Boolean flag \texttt{Strict}.
\begin{framed}
\textbf{Algorithm: \texttt{TaRDiS on Trees}}\\
\emph{Input}: A rooted temporal tree $\mathcal{G}$ and a Boolean flag \texttt{Strict}.\\
\emph{Output:} A TaRDiS $S$ of minimum size.
\begin{enumerate}
    \item Initialise $S=\emptyset,M=\emptyset$, and $s=0$. 
    \item If \texttt{Strict}:
    \begin{enumerate}
        \item Set $s=1$.
    \end{enumerate}
    \item While $M\neq V(\mathcal{G})$:
            \begin{enumerate}
            \item If $\mathcal{G}\setminus M$ is a star:
            \begin{enumerate}
                \item Add the central vertex to $S$ and return $S$.
            \end{enumerate}
            \item Denote $L$ the set of vertices in $V(\mathcal{G})\setminus M$ at maximum depth. 
            \item Let $p$ be the parent of an arbitrary vertex from $L$. 
            \item Choose the vertex $l\in C(p)\setminus M$ which minimizes $\lambda_{\max}(p,l)$. 
            \item Let $g$ be the parent of $p$ (and hence the grandparent of $l$).
            \item  While $l$ is not in $M$:\label{alg:while}
            \begin{enumerate}
                \item If $\lambda_{\max}(l,p)< \lambda_{\min}(p,g)+s$:
                \begin{enumerate}
                    \item Add $p$ to $S$. \label{alg:addp}
                    \item Find $R_p$. 
                    \item Add $R_p$ to $M$: $M=M\cup R_p$.
                \end{enumerate}
                \item Else, if $\lambda_{\max}(l,p)\geq \lambda_{\max}(p,g)+s$:\label{alg:else}
                \begin{enumerate}
                    \item  Do $M:= M \setminus \{p\}$ (note this does nothing if $p\notin M$ already).
                    \item Add all children of $p$ to $M$.          
                \end{enumerate}
                \item Else, update $\lambda$: do $\lambda(p,g)=\lambda(p,g)\setminus \{\lambda_{\max}(p,g)\}$.
            \end{enumerate}
            \end{enumerate}
            \item Return $S$.
        \end{enumerate}
\end{framed}

\smallskip
\begin{lemma}\label{lem:terminate}
    The algorithm \texttt{TaRDiS on Trees} always terminates in time $O(|\mathcal{E}|^2)$, where $\mathcal{E}$ is the set of time-edges in $\mathcal{G}$.
\end{lemma}
\begin{proof}
    The algorithm operates by adding vertices to a set $S$ until the whole tree is marked. Observe that each iteration of the inner while-loop (step:~\ref{alg:while}) results in either an increase in size of the set $M$ of marked vertices or removal of an appearance of the time-edge between parent and grandparent vertices of the vertex in question. Note that, if there is only one appearance of the edge between parent and grandparent vertices, a vertex must be added to $M$. Therefore, there are at most $|\mathcal{E}|$ iterations of the inner while loop. Computing $R_u$ is achievable in linear time when $\mathcal G_\downarrow$ is a tree. Thus, each iteration can be completed in time linear in the number of time-edges. Hence, the algorithm terminates in time $O(|\mathcal{E}|^2)$. 
  \end{proof}

We let $\mathcal{T}_v$ denote the subtree rooted at a vertex $v$.
\smallskip
\begin{lemma}\label{lem:inv1}
    For any vertex $v$ added to the set $S$ by the algorithm \texttt{TaRDiS on Trees}, $S\cap \mathcal{T}_v$ is a TaRDiS of the subtree $\mathcal{T}_v$ rooted at $v$. 
\end{lemma}
\begin{proof}
    We show that all descendants of $v$ are temporal reachability dominated by $S\cap \mathcal{T}_v$. Note that the children of $v$ are necessarily temporal reachability dominated by $S$ since they are adjacent to the vertex $v$. Before $v$ is added to $S$, it must have a child $l$ which is an unmarked vertex at maximum depth. Therefore, all vertices deeper in the subtree must be in $M$.
    
    We will show that all vertices in $M\cap \mathcal{T}_v$ are temporally reachable from a vertex in $S$. We claim that, if a vertex $u$ is in $M\cap \mathcal{T}_v$ but not temporally reachable from $S\setminus \{v\}$, then there is a temporal path from $v$ to $u$. This is the because edge from $u$ to its parent $p_u$ is active strictly later (or possibly at the same time in the nonstrict case) than the last appearance of the edge from $p_u$ to the grandparent vertex $g_u$. Otherwise, $u$ would not be marked. Note that this implies that neither $g_u$ nor $p_u$ are temporally reachable from a vertex in $S\setminus \{v\}$; otherwise $u$ would also be reachable by appending the edges $(g_u,p_u)$ and $(p_u,u)$ to such a path. 
    
    This either gives us an ancestor of $u$ which is not in $M$, or an ancestor of $u$ in $M$ which is not temporally reachable from a vertex in $S\setminus\{v\}$. If the latter is true, we repeat this logic until we reach an ancestor $w$ which is not in $M$. This must be either $v$ or a child of $v$ since $l$ is an unmarked vertex of maximum depth. For all descendants $w'$ of $w$ that are $u$ or ancestors of $u$, we have shown that if the parent of $w'$ is temporally reached by a path from an ancestor, then $w'$ must be too. Since $v$ is either $w$ or the parent of $w$, $u$ and all of its ancestors in $\mathcal{T}_v$ must be temporally reachable from $v$ and thus $S\cap \mathcal{T}_v$ is a TaRDiS for $\mathcal{T}_v$.
  \end{proof}

\begin{lemma}\label{lem:inv2}
    For any vertex $v$ added to the set $S$ by the algorithm \texttt{TaRDiS on Trees}, $S\cap \mathcal{T}_v$ is a minimum TaRDiS of the subtree $\mathcal{T}_v$ rooted at $v$. 
\end{lemma}
\begin{proof}
    Suppose for contradiction that such a vertex $v\in S$ and set $S'$ exist where $S'\subseteq T_v$ is a TaRDiS of $\mathcal{T}_v$ and strictly smaller than $S\cap \mathcal{T}_v$. We assume without loss of generality that $S'$ is a minimal TaRDiS. 
    
    Consider the set $S'\cap S$. If this is non-empty, let $R$ be the union of reachability sets of vertices in $S'\cap S$. Let $\mathcal{T}'$ be the subgraph of $\mathcal{T}_v$ induced by removing all vertices in $R\setminus (S'\cap S)$. Since we have assumed that $|S'|<|S|$, there must be a connected component $\mathcal{T}''$ of $\mathcal{T}'$ where $|\mathcal{T}''\cap S'|<|\mathcal{T}''\cap S|$. 

    Let $u\in \mathcal{T}''\cap S$ be the vertex in $(\mathcal{T}''\cap S)\setminus S'$ of maximum depth. Observe that $u$ is not a leaf in $\mathcal{T}''$, otherwise it would not have an unmarked child before its addition to $S$ by the algorithm. Denote by $c$ the child of $u$ which minimises $\lambda_{\max}(u,c)$. Since $u$ is added to $S$, we must have that $\lambda_{\max}(c,u)<\lambda_{\min}(u,p)+s$ where $p$ is the parent of $u$. Therefore, the path consisting of time-edges $((u,p),\lambda_{\min}(u,p)), ((c,u),\lambda_{\max}(c,u))$ is not a valid temporal path. We now have two cases to consider. The case where appearances of the edge $(c,u)$ have been deleted and the case where no appearances of the edge $cu$ have been deleted by the algorithm. 

    If no appearances of $(u,c)$ have been removed by the algorithm, then $c$ is not temporally reachable from $p$ or any ancestors of $p$. Since we know that $u$ is in $S$ and not $S'$, $c$ must be temporally reachable from a descendant $c'$ of $u$ in $S'\setminus S$. We claim that we can swap $c'$ with $u$ without reducing the reachability set of vertices in $S'$. Suppose otherwise; that there is a vertex $w$ no longer reached by $S'$. Since we assumed that $u$ was a vertex of maximum depth in $S'\setminus S$, this contradicts that $S\cap\mathcal{T}_u$ is a TaRDiS of $\mathcal{T}_u$. Thus, this contradicts Lemma~\ref{lem:inv1}.

    Now suppose that there are appearances of the edge $(u,c)$ which have been removed by the algorithm. Then, a child of $c$ must have been an unmarked vertex of maximum depth. Call this vertex $g$. This implies that $g\in \mathcal{T}''$. If the algorithm were to add $c$ to $S$ when $g$ is unmarked, then $c$ would not be an unmarked vertex of maximum depth in any later iterations. Before deletion of appearances of $(u,c)$, we must have that $\lambda_{\max}(g,c)< \lambda_{\max}(c,u)+s$ and $\lambda_{\max}(g,c)\geq \lambda_{\min}(c,u)+s$. Following deletion of the appearance $\lambda_{\max}(c,u)$ (potentially multiple times), $\lambda_{\max}(g,c)\geq \lambda_{\max}(c,u)+s$ and $\lambda_{\max}(c,u)<\lambda_{\min}(u,p)+s$. Therefore, there is no temporal path from $p$ or an ancestor of $p$ to $g$. Since we know that $u$ is in $S$ and not $S'$, $g$ must be temporally reachable from a descendant $g'$ of $u$ in $S'\setminus S$. We claim that we can swap $g'$ with $u$ without reducing the reachability set of vertices in $S'$. Suppose otherwise; that there is a vertex $w$ no longer reached by $S'$. Since we assumed that $u$ was a vertex of maximum depth in $S'\setminus S$, this contradicts that $S\cap\mathcal{T}_u$ is a TaRDiS of $\mathcal{T}_u$. Thus, this also contradicts Lemma~\ref{lem:inv1}.

    In both cases, we have shown that there is a vertex in $S'\setminus S$ that can be replaced by $u$ without reducing the reachability set of vertices in $S'$. We repeat this with the vertex at maximum depth in $\mathcal{T}''\cap S\setminus S'$ until either we contradict the assumption that $S'$ is a TaRDiS of $\mathcal{T}_v$, or obtain $S=S'$; contradicting our assertion that $|S'|<|S|$. Therefore $S$ is a minimum TaRDiS of $\mathcal{T}_v$.
  \end{proof}
\begin{lemma}\label{lem:algTaRDiS}
    The set $S$ output by the algorithm \texttt{TaRDiS on Trees} is a minimum TaRDiS of $\mathcal{G}$.
\end{lemma}
\begin{proof}
    If the root $r$ of $\mathcal{G}$ is in $S$ as output by the algorithm, then combining Lemmas~\ref{lem:inv1} and~\ref{lem:inv2} gives us the desired result. Suppose $r$ is not in $S$. Then, rooting the tree at the final vertex added to $S$ gives the result without changing the vertices in $S$.
    % By Lemma~\ref{lem:terminate}, $S$ is a TaRDiS. Therefore, $r$ must be reached by a vertex in $S$. Let $\cup_{s\in S}\mathcal{T}_s$ be the union of subtrees rooted at vertices in $S$. By Lemma~\ref{lem:inv2}, $S$ is a minimum TaRDiS of $\cup_{s\in S}\mathcal{T}_s$. Furthermore, the set $S\setminus(\cup_{s\in S}\mathcal{T}_s)$ is empty. If there were a smaller TaRDiS $S'$ of $\mathcal{G}$, 
    % \begin{align*}
    %     &|S|>|S'|\\
    %     \implies &|S\setminus(\cup_{s\in S}\mathcal{T}_s)|+|S\cap(\cup_{s\in S}\mathcal{T}_s)|>|S'|\\
    %     \implies &|S\cap(\cup_{s\in S}\mathcal{T}_s)|>|S'|
    % \end{align*}
    %since $S\setminus(\cup_{s\in S}\mathcal{T}_s)$ is empty. The above inequality implies that $S'$ dominates $\cup_{s\in S}\mathcal{T}_s$ with fewer vertices than $S$. This contradicts Lemma~\ref{lem:inv2}.
  \end{proof}

This gives us the following theorem and corollary.
\smallskip
\begin{theorem}\label{thm:tree-alg}
    When the footprint of the graph is a tree, \textsc{TaRDiS} is solvable in $O(|\mathcal{E}|^2)$ time.
\end{theorem}

We obtain a running time dependent only on the number of vertices for simple temporal graphs, where $|\mathcal{E}|=|E(\mathcal G_\downarrow)|=V(\mathcal{G})-1$.

\smallskip
\begin{corollary}
    When the input temporal graph is simple and the footprint of the graph is a tree, \textsc{TaRDiS} is solvable in $O(n^2)$ time.
\end{corollary}

In Section~\ref{sec:param-tardis}, we give a more general algorithm which solves any variant of \textsc{TaRDiS} on a nice tree decomposition of the underlying graph. When the footprint graph is a tree, we have treewidth 1. Our algorithm for \textsc{TaRDiS} on trees runs faster than the tree decomposition algorithm (whose runtime also grows with $\tau$), so this does not consume Theorem~\ref{thm:tree-alg}.

\section{Classical complexity results for \textsc{MaxMinTaRDiS}}\label{sec:classic-maxmintardis}

We now shift our focus to \textsc{MaxMinTaRDiS}, where the objective is to find a temporal assignment which precludes the existence of a small TaRDiS. For each variant, we characterize the minimum lifetime such that the problem becomes intractable. Intriguingly, the expected leap in complexity to $\Sigma_2^P$ is only seen with \textsc{Happy MaxMinTaRDiS}.

We give a full definition of $\Sigma_2^P$ in Section~\ref{sec:Sigma}. It is widely believed that $\Sigma_2^P$ is a strict superclass of NP $\cup$ coNP.
In this sense, \textsc{Happy MaxMinTaRDiS} is a harder problem than both \textsc{Strict MaxMinTaRDiS} (which is coNP-complete), and all \textsc{TaRDiS} variants (which are NP-complete as shown in Section~\ref{sec:classic-tardis}). We show \textsc{Nonstrict MaxMinTaRDiS} is at least as hard as all \textsc{TaRDiS} variants and at most as hard as \textsc{Happy MaxMinTaRDiS}, but leave open its exact complexity. Interestingly, this is achieved by proving that the well-studied \textsc{Distance-3 Independent Set} problem is a subproblem of \textsc{Nonstrict MaxMinTaRDiS}. 

\subsection{Containment in \texorpdfstring{$\Sigma_2^P$}{Sigma 2 P}, useful tools, and small lifetime} \label{sec:prelim-mmtardis}

Here we give some preliminary complexity results for each variant of \textsc{MaxMinTaRDiS}. We begin by showing that they are all contained in $\Sigma_2^P$. We then show that we need only consider simple temporal assignments when trying to solve \textsc{MaxMinTaRDiS}. Finally, we draw comparisons between \textsc{Happy MaxMinTaRDiS} and \textsc{Strict MaxMinTaRDiS} and the classical problems \textsc{Edge Colouring} and \textsc{Dominating Set} respectively, which allow us to show NP-hardness and coNP-hardness of the respective problems.

\begin{lemma}\label{lem:insigma}
    Each variant of the problem \textsc{MaxMinTaRDiS} is contained in $\Sigma_2^P$.
\end{lemma}
\begin{proof}
    \textsc{(Strict/Nonstrict/Happy) MaxMinTaRDiS} admits a triple $(H, k, \tau)$ as a yes-instance if and only if \textit{there exists} a (happy) temporal assignment $\lambda:E(H)\to \tau$ such that \textit{for every} (Strict/Nonstrict) TaRDiS $S$ of $(H,\lambda)$, $|S|\geq k$. It follows that a formula $\Phi$ over variable sets $X$ and $Y$ can be produced such that $(H,k,\tau)$ is a yes-instance if and only if there exists an assignment to the variables in $X$ such that for every assignment to the variables of $Y$, $\Phi(X,Y)$ is \texttt{True}. We refer the reader to Section~\ref{sec:Sigma} for a more detailed construction of such a Boolean formula.
  \end{proof}

We begin with the observation that we need only consider \emph{simple} temporal assignments for the input graph.

\smallskip
\begin{lemma}\label{lem:max-min-always-simple}
    Let $(H, k, \tau)$ be a yes-instance of \textsc{MaxMinTaRDiS}. Then there exists a \emph{simple} temporal assignment $\lambda: E \to [\tau]$ such that the cardinality of the minimum TaRDiS on $(H,\lambda)$ is at least $k$.
\end{lemma}
\begin{proof}
    We begin by supposing, for a contradiction, that $(H,k,\tau)$ is an instance of \textsc{MaxMinTaRDiS} such that any optimal solution is non-simple. Denote such a solution by $\lambda^*$. By our assumption, there is at least one edge $e^*\in E(H)$ such that $|\lambda^*(e^*)|>1$. Let $\lambda$ be a simple temporal assignment such that, for all edges $e\in E(H)$, $\lambda(e)\in\lambda^*(e)$. Then, under $\lambda$, the reachability set of any vertex $v\in V(H)$ is a subset of the reachability set of $v$ under $\lambda^*$. Therefore, a minimal TaRDiS of $(H,\lambda)$ must be at most the cardinality of a minimal TaRDiS of $(H,\lambda^*)$. Since $\lambda$ is a simple temporal assignment, we have contradicted our assumption. Thus, for every instance of \textsc{MaxMinTaRDiS}, there exists an optimal solution which is simple.
 \end{proof}

\begin{lemma}[\cite{casteigts_simple_2022_bugfree}]\label{lem:edgecol}
    A static graph $H$ admits a happy temporal assignment $\lambda:E(G)\to [\tau]$ if and only if $H$ is $\tau$-edge colourable. 
\end{lemma}

\begin{proof}
    To see this, we assign each of the times a colour. Then, if the graph cannot be $\tau$-edge coloured, we cannot assign times to the edges such that no two adjacent edges share a time. Furthermore, if we can give $H$ a happy temporal assignment, then the corresponding assignment of colours to edges is a proper edge-colouring.
  \end{proof}

\textsc{Happy MaxMinTaRDiS} restricted to instances with $k=0$ asks only if there exists a happy assignment $\lambda$ with lifetime $\tau$ for the input graph $G$. This is equivalent to the \textsc{Edge colouring} problem with $\tau$ colours. \textsc{Edge Colouring} is NP-complete, even when the number of colours is 3 \cite{holyer1981np}. 

\smallskip
\begin{corollary}\label{cor:edgecol}
     \textsc{Happy MaxMinTaRDiS} is NP-hard for any $\tau\geq 3$, even when $k=0$. 
\end{corollary}

\smallskip
 \begin{lemma}\label{lem:smalltau}
   \textsc{Nonstrict MaxMinTaRDiS} can be computed in linear time when $\tau=1$. When $\tau\le 2$, \textsc{Happy MaxMinTaRDiS} is solvable in linear time. 
\end{lemma}

\begin{proof}
If every edge is active at the same time, we have no choice in the temporal assignment. As shown in Lemma~\ref{lem:smalltau_tardis}, we can find the size of the minimum nonstrict TaRDiS in linear time by counting the number of connected components.

As stated in Lemma~\ref{lem:smalltau_tardis}, the footprint of a happy temporal graph with lifetime 2 necessarily consists of disconnected paths and even cycles. In paths and cycles of even length, there is only one happy temporal assignment up to symmetry. In odd paths, the optimal ordering is that wherein both edges incident to leaves are assigned time 1. This forces one of the endpoints of these edges to be in a TaRDiS, which is not the case if they are assigned time 2.
\end{proof}

\begin{lemma}\label{lem:mmtar2ds}
    For any static graph $H$ and $k\in \mathbb N^+$, $(H,k)$ is a yes-instance of \textsc{Strict MaxMinTaRDiS} if and only if $(H,k-1)$ is a no-instance of \textsc{Dominating Set}.
\end{lemma}

\begin{proof}
    We first show that the constant function $\lambda$ is an optimal one (i.e. one which maximizes the size of the minimum TaRDiS for $(H,\lambda)$).
    Suppose otherwise, that we have an optimal temporal assignment $\lambda'$ where there exist edges $e$, $e'$ such that $\lambda'(e)\neq \lambda'(e')$. Then, under a temporal assignment $\lambda(e)=c$ for all edges $e\in E$ and $c\in\mathbb{N}$, all vertices have reachability sets whose cardinality are bounded above by the size of their reachability set under $\lambda'$. Therefore a minimum TaRDiS under $\lambda$ must be at least the size of a minimum TaRDiS under $\lambda'$. Applying Lemma~\ref{lem:tar2ds}, we see that $(H,k)$ is a yes-instance of \textsc{Strict MaxMinTaRDiS} if and only if $(H,k-1)$ is a no-instance of \textsc{Dominating Set}.
\end{proof}
\begin{corollary}\label{cor:ds-strict-mmtardis}
    \textsc{Strict MaxMinTaRDiS} is coNP-complete, coW[2]-complete with respect to $k$, and para-NP-hard with respect to $\Delta+\tau$.
\end{corollary}

\subsection{\texorpdfstring{$\Sigma_2^P$-completeness}{Sigma2P-completeness} of \textsc{Happy MaxMinTaRDiS} with lifetime 3}\label{sec:Sigma}

As stated in Corollary~\ref{cor:edgecol}, \textsc{Happy MaxMinTaRDiS} is trivially NP-hard even for lifetime $\tau=3$. Lemma~\ref{lem:insigma} gives us that the problem is in $\Sigma_2^P$. We characterize the problem's complexity exactly, showing it is $\Sigma_2^P$-complete.

\begin{definition}[$\Sigma_2^P$ - adapted from \cite{arora_computational_2009_bugfree}, Definition 5.1]
    $\Sigma_2^P$ is the set of all languages $L$ for which there exists a polynomial-time Turing Machine $M$ and a polynomial $q$ such that
    $$x\in L\iff \exists u\in\{0,1\}^{q(|x|)} \forall v\in\{0,1\}^{q(|x|)} M(x,u,v)=1$$ 
    for every $x\in\{0,1\}^*$.
\end{definition}
Arora and Barak \cite{arora_computational_2009_bugfree} note that $\text{NP} \cup \text{co-NP} \subseteq \Sigma_2^P$ and that $\Sigma^P_1=\text{NP}$. The complement of a $\Pi_i^P$-complete problem is necessarily $\Sigma_i^P$-complete \cite{papadimitriou_computational_1994}. 

We begin by presenting its problem \textsc{Restricted Planar Satisfiability}, which is $\Pi_2^P$ complete \cite{gutner_complexity_2008}.

\begin{framed}
    \noindent
    \textbf{\textsc{Restricted Planar Satisfiability (RPS)}}\\
    \emph{Input:} An expression of form $(\forall X) (\exists Y) \Phi(X,Y)$ with $\Phi$ a CNF formula over the set $X\cup Y$ of variables. Each clause contains exactly 3 distinct variables, each variable occurs in exactly three clauses, every literal appears at most twice, and the graph $G_\Phi$ is planar.\\
    \emph{Question:} Is the expression true?
\end{framed}

We also introduce the complement problem \textsc{co-RPS}, which we will reduce from.

\begin{framed}
    \noindent
    \textbf{\textsc{co-Restricted Planar Satisfiability (co-RPS)}}\\
    \emph{Input:} An expression of form $(\exists X) (\nexists Y) \Phi(X,Y)$ with $\Phi$ a CNF formula over the set $X\cup Y$ of variables. Each clause contains exactly 3 distinct variables, each variable occurs in exactly three clauses, every literal appears at most twice, and the graph $G_\Phi$ is planar.\\
    \emph{Question:} Is the expression true?
\end{framed}

\smallskip
\begin{lemma}\label{lem:corps_sigma2p}
    The problem co-\textsc{RPS} is $\Sigma_2^P$ complete. 
\end{lemma}

\begin{proof}
    Gutner shows \textsc{Restricted Planar Satisfiability (RPS)} to be $\Pi_2^P$-complete in \cite{gutner_complexity_2008}. Note that they do not state the restriction that each literal appears at most twice in the lemma stating their result, but this can be seen from their construction. 
  \end{proof}

\subsubsection{Intuition}
We can imagine problems in $\Sigma_2^P$ as games with two players. In \textsc{co-RPS} the first player chooses an assignment to variables in $X$, then the second player chooses an assignment to the variables in $Y$. If the resulting assignment to $X\cup Y$ satisfies $\Phi$ then the second player wins; otherwise the first player wins. Analogously, in \textsc{Happy MaxMinTaRDiS} the first player chooses a happy temporal assignment $\lambda$ with lifetime $\tau$ for the edges of $G$, then the second player chooses a set $S$ of vertices in $G$ of size at most $k-1$. If $S$ is a TaRDiS of $(G,\lambda)$ then the second player wins; otherwise the first player wins.

We perform this reduction by creating gadgets which force the first player to choose a function $\lambda$ which is \emph{nice}. Essentially, $\lambda$ must assign specific times to certain edges. This means that if the first player does not do this, the second player can always win. We can then replicate some techniques from the proof of Theorem~\ref{thm:happytardis} to encode clauses being satisfied. The choice of TaRDiS by the second player then encodes a truth assignment to variables in $Y$. By allowing a small amount of freedom in the choice of $\lambda$, we allow the first player to encode a truth assignment to the variables of $X$. We now give the formal construction.

\subsubsection{Construction}
Given an instance $(\Phi,X,Y)$ of \textsc{co-RPS}, we will produce an instance $(G, k, \tau=3)$ of \textsc{MaxMinTaRDiS}. The number of clauses in $\Phi$ is denoted by $m$. We label the number of variables in $X$ (respectively $Y$) with $n_X$ (resp. $n_Y$). The total number of variables is $n$. Note that $3m \geq n\geq m$, so the size of the \textsc{co-RPS} instance is linear in $n$. Further, we denote $X = \{x_1,\ldots, x_{n_x}\}$ and $Y=\{y_{n_x+1}, \ldots, y_{n_x+n_y}\}$. We say that the literal $l_{2i}$ (respectively, $l_{2i-1}$) is the positive (resp., negative) literal for variable $x_i$ if $i\leq n_x$, and the positive (resp. negative) literal for variable $y_i$ if $n_x<i\leq n_y$. 

We now describe the construction of $G$ and $k$. The construction of our gadgets is intended to guarantee that any optimal temporal assignment $\lambda$ will have certain properties. First, we define the Uncovered 2-Gadget, Uncovered 3-Gadget, and Covered 2-Gadget. Each of these is treated in our construction as a vertex of degree one. The construction for each of these makes use of large values $\alpha$ and $\beta$. Let $\beta=100n$ and $\alpha=600n\beta+600n+2$. In particular we require that $\alpha >> \beta >> n$, $\beta$ is even, and $\alpha \equiv 2 \mod 12$.

\subsubsection{Uncovered 2-Gadget (U2G)}
This construction is illustrated in Figure~\ref{fig:u2g}. We use a U2G by connecting it to some vertex $x$ elsewhere in the construction. Given some such $x$, we create vertices $y$ and $w$ and a ladder graph on $2\alpha$ vertices $\{u_1, \ldots, u_\alpha, v_1, \ldots, v_\alpha\}$, with the edges $(x,y), (y,u_1),(y,v_1),(w,u_{\alpha}), (w,v_{\alpha})$. The \emph{ladder graph} $L_n$ on $2n$ vertices has vertex set $V(L_n)=\{u_1,\ldots,u_n,v_1,\ldots,v_n\}$, and edge set $E(L_n)=\{(u_i,u_{i+1})|i\in \mathbb [n-1]\}\cup \{(v_i,v_{i+1})|i\in \mathbb [n-1]\} \cup \{(u_i,v_i)|i\in \mathbb [n]\}$.

\subsubsection{Covered 2-Gadget (C2G)}
The construction is very similar to that of the U2G - the only difference is the incrementation of the ladder length by 1. Given $x$, we create vertices $y$ and $w$ and a ladder graph on $2\alpha+2$ vertices $\{u_1, \ldots, u_{\alpha+1}, v_1, \ldots, v_{\alpha+1}\}$, and edges $(x,y), (y,u_1),(y,v_1),(w,u_{\alpha+1}), (w,v_{\alpha+1})$. 

\begin{figure}[!ht]
    \centering
    \includegraphics[page=1, width=\textwidth]{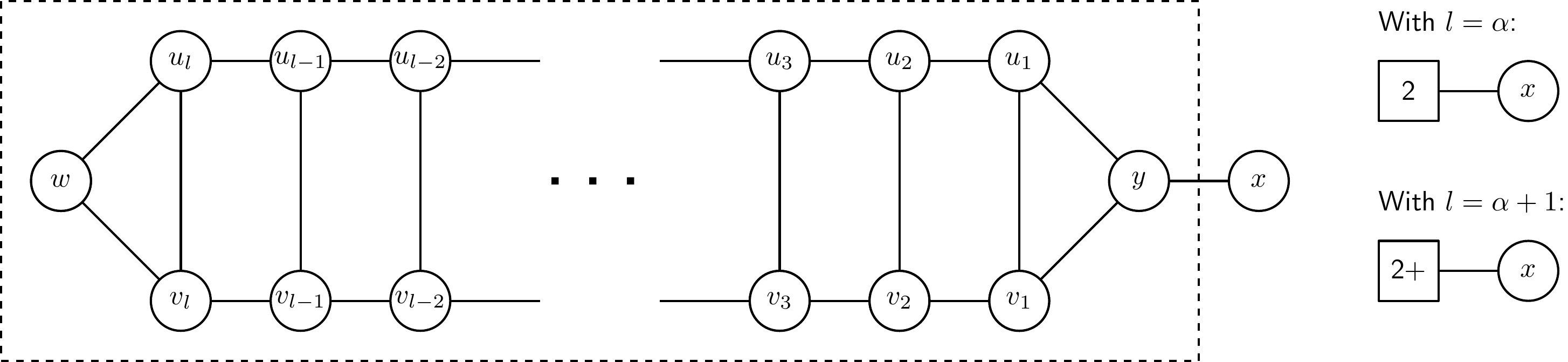}
    \caption{A vertex $x$ incident to the Uncovered 2 Gadget (U2G) or Covered 2 Gadget (C2G). Left: construction of the gadget where $l=\alpha$ for a U2G and $l=\alpha+1$ for a C2G. Right: usage in later constructions. Note that planarity is preserved.}
    \label{fig:u2g}
\end{figure}

\subsubsection{Uncovered 3-Gadget (U3G)}
The construction is illustrated in Figure~\ref{fig:u3g}. We start with the ladder graph on $2\alpha$ vertices $\{u_1, \ldots, u_{\alpha}, v_1, \ldots, v_{\alpha}\}$, and doubly subdivide the edge $(u_i,u_{i+1})$ (respectively $(v_i,v_{i+1})$) whenever $i$ is even (resp. odd) with two new vertices $a_i, b_i$. We say \emph{doubly subdivide} an edge $(u,v)$ to refer to the deletion of $(u,v)$ and introduction of two vertices $a,b$ and three edges $(u,a),(a,b),(b,v)$. We add the vertices $y_1,y_2,y_3,w$ and edges $(x,y_1)$, $(y_1,y_2)$, $(y_2,y_3)$, $(y_3,u_1)$, $(y_3,v_1)$, $(w,u_{\alpha})$. We denote the set of internal vertices of the U3G by $V_\text{U3G}=\{a_1, \ldots, a_{\beta-1}, b_1, \ldots b_{\beta-1}, y_1, y_2, y_3, u_1, \ldots, u_\beta, v_1, \ldots, v_\beta, w\}$. Internal vertices are highlighted by a red (solid) box in Figure~\ref{fig:u3g}. We make any vertex not already of degree 3 in $V_\text{U3G}$ incident to a U2G. 

\begin{figure}[!ht]
    \centering
    \includegraphics[page=4, width=\textwidth]{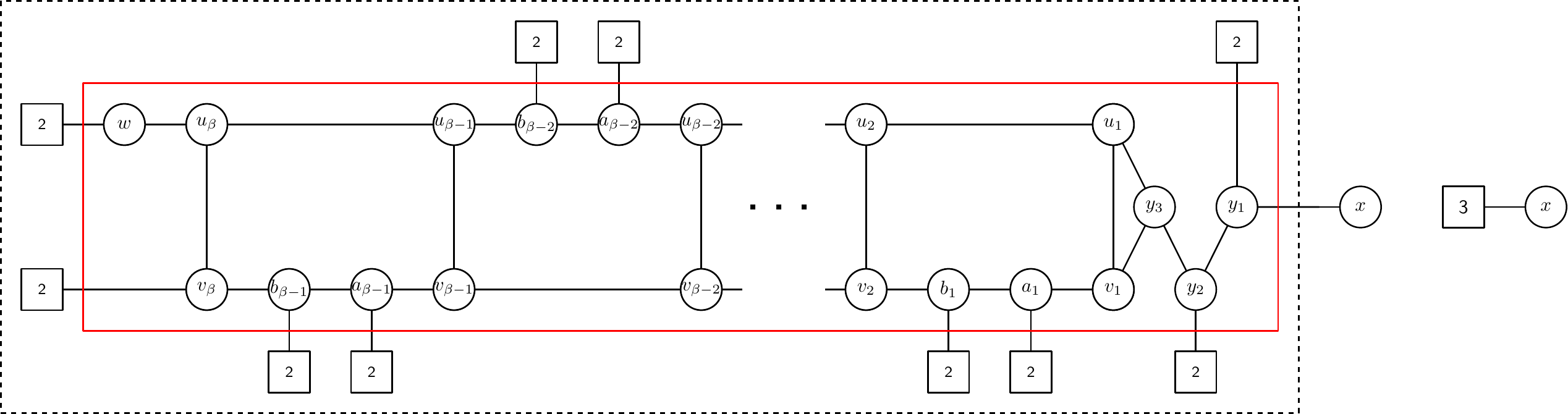}
    \caption{A vertex $x$ incident to the Uncovered 3 Gadget (U3G). Left: construction of the gadget; vertices of $V_\text{U3G}$ are in the red (solid) box. Right: usage in later constructions. Note that planarity is preserved.}
    \label{fig:u3g}
\end{figure}

\subsubsection{Uncovered 1-Gadget (U1G)}
For a U1G, we simply create a vertex $y$ incident to both a C2G and a U3G as shown in Figure~\ref{fig:u1g}, then add an edge from $y$ to the target vertex $x$.

\begin{figure}[!ht]
    \centering
    \includegraphics[page=2, width=.5\textwidth]{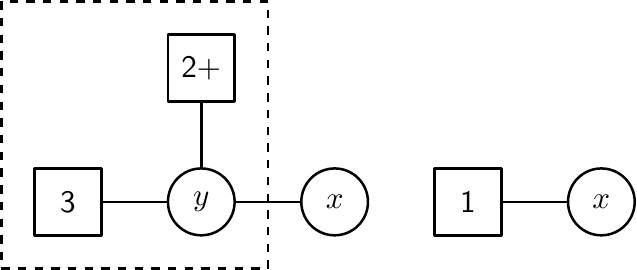}
    \caption{A vertex $x$ incident to the Uncovered 1 Gadget (U1G). Left: implementation of the gadget. Right: usage in other constructions.}
    \label{fig:u1g}
\end{figure}

\subsubsection{Construction: literal vertices and clause gadgets}
Exactly $3m$ \emph{literal vertices} and $m$ \emph{clause gadgets} are created.
We iterate over $\Phi$, and for each clause $c_j$ we create a cycle on 6 new vertices $Q_j=\{q_j^1,\ldots,q_j^6\}$, and make $q_j^2,q_j^4,q_j^6$ each incident to a U2G. Where the clause $c_j$ contains the $a$th appearance of some literal $l_i$ in $\Phi$, we create two new \emph{literal vertices} $l_i^a$ and $\overline{l_i^a}$ connected by an edge, and make each of $l_i^a$ and $\overline{l_i^a}$ incident to a new U1G. We make $l_i^a$ incident to a vertex from $\{q_j^1, q_j^3, q_j^5\}$, such that each of these is adjacent to exactly one literal vertex. Figure~\ref{fig:mmt_clause} illustrates this construction. Vertices $u,v,w, \overline{u},\overline{v},\overline{w}$ are literal vertices for which the corresponding literals appear in $c_j$. 

\begin{figure}[!ht]
    \centering
    \includegraphics[width=.7\textwidth, page=3]{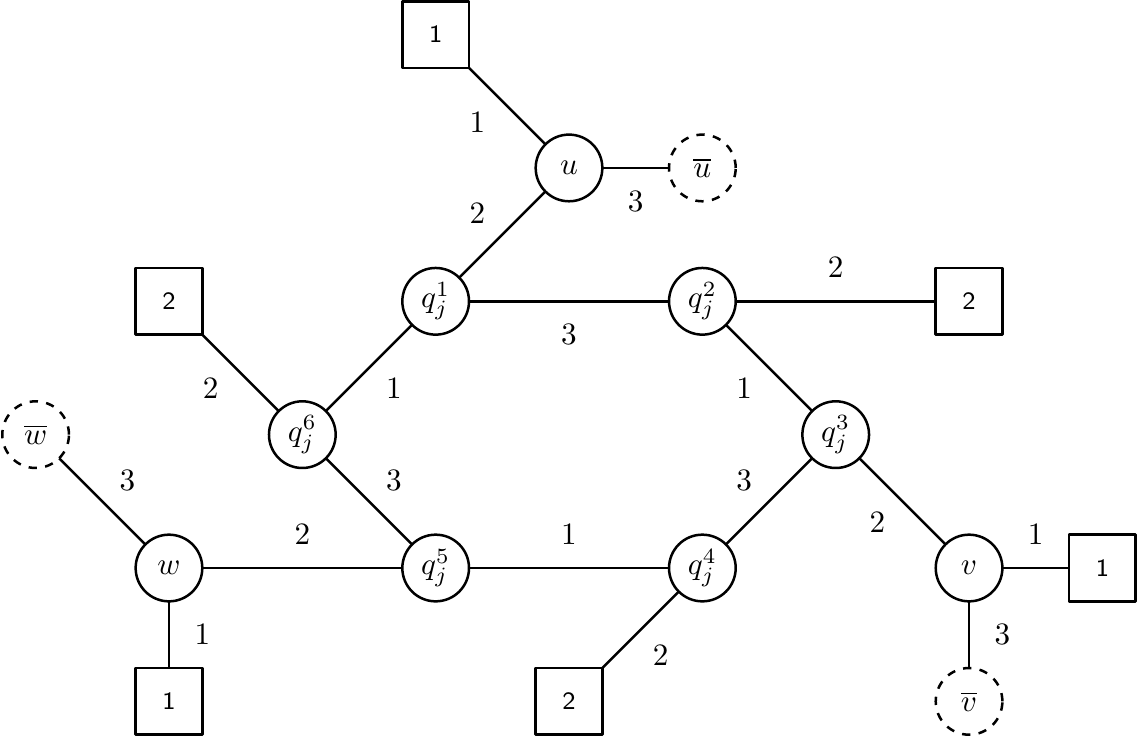}
    \caption{The clause gadget for clause $c_j$ together with an example \emph{nice} assignment. If, for example, clause $c_j=(l_{5}\lor l_{17} \lor l_{20})$ is the first appearance of $l_{5}$ and $l_{20}$ and the second appearance of $l_{17}$ in $\Phi$, then $u=l_{5}^1$, $v=l_{17}^2$ and $w=l_{20}^1$. The neighbourhoods of dashed vertices are shown in Figures \ref{fig:xvar} and \ref{fig:yvar} respectively depending on whether they correspond to variables in $X$ or $Y$.}
    \label{fig:mmt_clause}
\end{figure}
\begin{figure}[!hb]
    \centering
    \includegraphics[page=5,width=\textwidth]{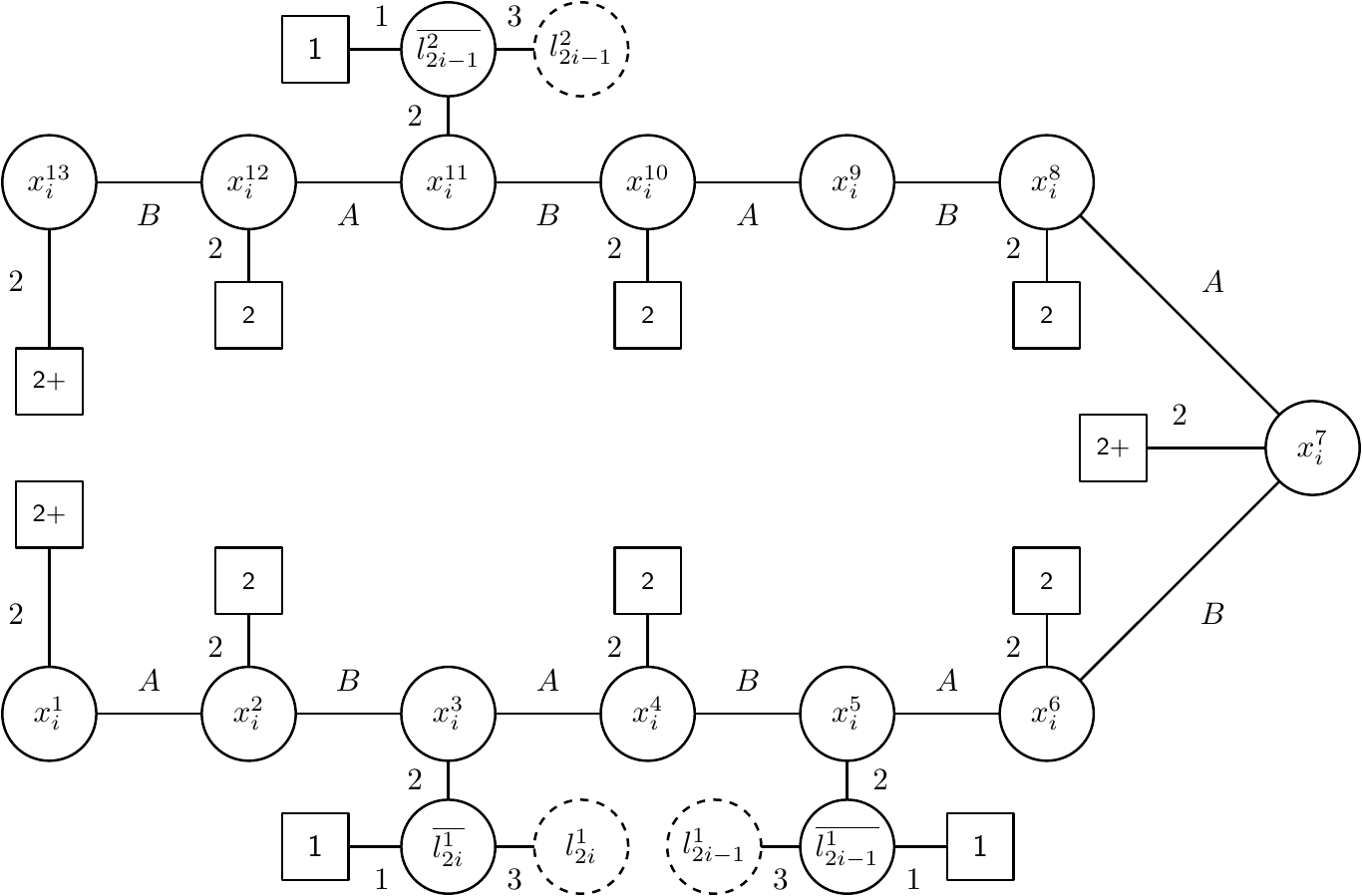}
    \caption{The $X$-variable gadget for variable $x_i$ together with an example gadget-respecting assignment; note that either $A=1$ and $B=3$ or $B=1$ and $A=3$. Only in the former case, which corresponds to setting $x_i$ to \texttt{True}, does $x_i^4$ have a temporal path to $l_{2i}^1$. Here $x_i$ appears twice negatively and once positively, hence $x_i^9$ is not incident to a literal vertex. The neighbourhoods of dashed vertices are shown in Figure \ref{fig:mmt_clause}.}
    \label{fig:xvar}
\end{figure}

\subsubsection{Construction: \texorpdfstring{$X$}{X}-variable gadget}
This construction is illustrated in Figure~\ref{fig:xvar}. For each variable $x_i\in X$, we create 13 vertices $\{x_i^1,\ldots,x_i^{13}\}$ connected in a path, with vertex $x_i^p$ incident to a U2G if $p$ is even, and incident to a C2G for $p\in\{1,7,13\}$. Lastly we add the edges $(x_i^3,\overline{l_{2i}^1}),(x_i^5,\overline{l_{2i-1}^1}),(x_i^9,\overline{l_{2i}^2})$, and $(x_i^{11},\overline{l_{2i-1}^2})$ whenever the appropriate literal vertex exists. Recall that, by the restrictions on $\Phi$, at most one of the vertices $l_{2i}^2$ and $l_{2i-1}^2$ exist in our construction for any input. 

\subsubsection{Construction: \texorpdfstring{$Y$}{Y}-variable gadget}
This construction is illustrated in Figure~\ref{fig:yvar}. For each variable $y_i\in Y$, we create 8 vertices $\{a_i, v_i^1, v_i^2, b_i, F_i^1, F_i^2, T_i^2, T_i^1\}$ connected in a cycle, with vertices $v_i^1$ and $v_i^2$ incident to U3Gs and vertices $a_i$ and $b_i$ incident to U1Gs. Then we add edges $(T_i^1,\overline{l_{2i}^1}),(T_i^2,\overline{l_{2i}^2}),(F_i^2, \overline{l_{2i-1}^2}),(F_i^1,\overline{l_{2i-1}^1})$ whenever the appropriate literal vertex exists. If vertex $T_i^2$ is not incident to a literal vertex, we connect it to a U2G. Otherwise, $F_i^2$ is connected to a U2G. 
\begin{figure}
    \centering
    \includegraphics[page=6,width=.7\textwidth]{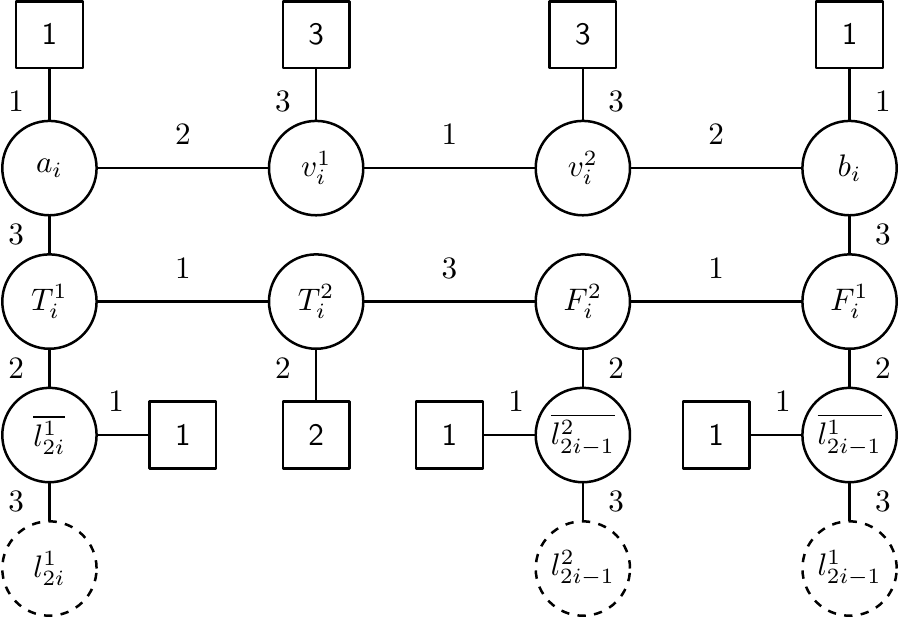}
    \caption{The $Y$-variable gadget for variable $y_i$ together with an example \emph{nice} assignment. Note the similarity to the variable gadget in Figure~\ref{fig:happy_tardis}. Here $x_i$ appears twice negatively and once positively, hence $T_i^2$ is connected to a U2G and not a literal vertex. The neighbourhoods of dashed vertices are shown in Figure \ref{fig:mmt_clause}.}
    \label{fig:yvar}
\end{figure}

\subsubsection{Construction: \texorpdfstring{$k$}{k}}
In order to state $k$, we first define some auxiliary variables. As stated earlier, the number of $X$-gadgets, $Y$-gadgets and clause-gadgets is $n_X,n_Y$ and $m$ respectively. We then define:
\begin{align*}
\#_\text{lit} = 2\cdot3m \hspace{10pt}&\text{(The number of literal nodes.)}\\
\#_\text{U1G} = \#_\text{lit} + 2\cdot n_Y \hspace{10pt}&\text{(The number of U1Gs.)}\\
\#_\text{U3G} = 2n_Y + \#_\text{U1G} \hspace{10pt}&\text{(The number of U3Gs.)}\\
\#_\text{C2G} = 3n_X + \#_\text{U1G} \hspace{10pt}&\text{(The number of C2Gs.)}\\
\#_\text{U2G} = \#_\text{U3G}(2\beta + 2) + 3m + 6n_X + n_Y \hspace{10pt}&\text{(The number of U2Gs.)}\\
\end{align*}

We now define $k$:
\begin{align*}
k & = \#_\text{U2G}\cdot\left(\frac{\alpha+1}{3}\right) + \#_\text{C2G}\cdot\left(\frac{\alpha+4}{3}\right) + \#_\text{U3G}\cdot(\beta) + 2m + 3n_X + 3n_Y + 1 
\end{align*}

This concludes the construction of the \textsc{Happy MaxMinTaRDiS} instance.

\subsubsection{Properties of the construction} We first define the temporal assignments of interest to us.

\begin{definition}[Gadget-respecting, Nice Temporal Assignment]
    A happy temporal assignment $\lambda$ for the graph $G$ described above is:
    \begin{itemize}
        \item \emph{2-gadget-respecting} if every edge $(x,y)$ incident to a U2G or C2G is assigned time 2 under $\lambda$.
        \item \emph{3-gadget-respecting} (resp. \emph{1-gadget-respecting}) if every edge $(x,y)$ incident to a U3G (resp. U1G) is assigned time 3 (resp. 1) under $\lambda$. 
    \end{itemize}
    If $\lambda$ is 1-, 2- and 3-gadget-respecting, we say that $\lambda$ is \emph{nice}.
\end{definition}

We will show that a happy temporal assignment $\lambda$ for $G$ satisfies that every TaRDiS of $(G,\lambda)$ has cardinality at least $k-1$ if and only if $\lambda$ is \emph{nice}. We then show that there exists a $\lambda$ such that every TaRDiS of $(G,\lambda)$ has cardinality at least $k$ if and only if $(\Phi, X, Y)$ is a yes-instance of \textsc{co-RPS}.

\begin{lemma}\label{lem:U2G}
    For any happy temporal assignment $\lambda$ on any U2G gadget incident to $x$ consisting of vertices $V_\text{U2G}=\{y,w,u_1, \ldots, u_\alpha, v_1, \ldots, v_\alpha\}$, if $\lambda(x,y)=2$ 
    \begin{itemize}
        \item exactly $\frac{\alpha+1}{3}$ vertices from $V_\text{U2G}$ are needed to temporal reachability dominate $V_\text{U2G}$;
        \item  no choice of $\frac{\alpha+1}{3}$ vertices temporally dominates $V_\text{U2G}\cup \{x\}$.
    \end{itemize}
     If $\lambda(x,y)\neq 2$, it is possible to temporally dominate $V_\text{U2G}$ with exactly $\frac{\alpha+2}{4}$ vertices. 
\end{lemma}

\begin{proof}
    It is easy to check that there is only one proper 3-colouring of the edges of the induced graph on $V_\text{U2G}\cup \{x\}$ up to isomorphism. Denote by $A$ the colour of $(x,y)$, $B$ the colour of $(y,u_1)$ and $C$ the colour of $(y,v_1)$. Note that all edges $(u_i, v_i)$ are given colour $A$, and the edges on the path $u_1,\ldots,u_{\alpha}$ (resp. $v_1,\ldots,v_{\alpha}$) alternate between $B$ and $C$.   
    
    Any possible happy temporal assignment then corresponds exactly to one of 6 possible assignments of times $\{1,2,3\}$ to the colours $\{A,B,C\}$. From this point, we abuse notation slightly and denote $\lambda(x,y)$ by $\lambda(A)$, $\lambda(y,u_1)$ by $\lambda(B)$ and $\lambda(y,v_1)$ by $\lambda(C)$. By symmetry of the gadget construction, we can always assume $\lambda(B)<\lambda(C)$. 
    
    If $\lambda(A)=2$, we have $\lambda(B)=1$ and $\lambda(C)=3$; then $|R_v|\leq 6$ for all $v\in V_\text{U2G}$. Note in particular that  $R_{v_1}=\{y, v_1, v_2, v_3, u_1, u_2\}$, and $R_x\cap V_\text{U2G}=\{y,v_1\}$. Thus for any TaRDiS $S$, at least $|V_\text{U2G}|-2=2\alpha$ vertices in $V_\text{U2G}$ are reached by vertices of $S\cap V_\text{U2G}$, each of which reaches at most 6 vertices including itself. Recall that $\alpha \equiv 2 \mod 12$. Hence, $|S\cap V_\text{U2G}|\geq \lceil \frac{2\alpha}{6} \rceil = \frac{\alpha+1}{3}$. Note that, by the same logic, any set of size at most $ \frac{\alpha+1}{3}$ reaches at most $2\alpha+2$ vertices and hence no such set reaches every vertex in $V_\text{U2G}\cup \{x\}$ since $|V_\text{U2G}\cup \{x\}|=2\alpha+3$. Further, the set $\{v_1,u_4,v_7,u_{10},\ldots v_{\alpha-1}\}$ is of size $\frac{\alpha+1}{3}$ exactly and reaches every vertex in $V_\text{U2G}$. Conversely, if $\lambda(A)\in \{1,3\}$ then $\{v_1,v_5,v_9,\ldots,v_{\alpha-1}\}$ is a set of size $\frac{\alpha+2}{4}$ which reaches all vertices of $V_\text{U2G}$.
  \end{proof}

\begin{lemma}\label{lem:C2G}
    For any happy temporal assignment $\lambda$ on any C2G gadget incident to vertex $x$ consisting of vertices $V_\text{C2G}=\{y,w,u_1, \ldots, u_{\alpha+1}, v_1, \ldots, v_{\alpha+1}\}$, if $\lambda(x,y)=2$
    \begin{itemize}
        \item exactly $\frac{\alpha+4}{3}$ vertices from $V_\text{C2G}$ are needed to temporally dominate $V_\text{C2G}$;
        \item it is possible to temporally dominate $V_\text{C2G}\cup \{x\}$ with the same number of vertices.
    \end{itemize} If $\lambda(x,y)\neq 2$ it is possible to temporally dominate $V_\text{C2G}$ with exactly $\frac{\alpha+6}{4}$ vertices.
\end{lemma}

\begin{proof}
    As in the case of the U2G, there is only one proper 3-colouring of the edges of the induced graph on $V_\text{U3G}\cup \{x\}$ up to isomorphism. We again denote $A$ the colour of $(x,y)$, $B$ the colour of $(y,u_1)$ and $C$ the colour of $(y,v_1)$. Since gadget symmetry is preserved, we also may assume without loss of generality that $\lambda(B)<\lambda(C)$. Then the assignment of the edges on vertices from $V_\text{C2G}\cup \{x\}$ depends only on the value of $\lambda(x,y)$. 

    The dependence of reachability on this choice remains; namely if $\lambda(A)=2$ then $|R_v|\leq 6$ for all $v\in V_\text{C2G}$ and $R_x\cap V_\text{C2G}=\{y,v_1\}$. Applying the same logic as above, for any TaRDiS $S$ at least $|V_\text{C2G}|-2=2\alpha \bm{+2}$ vertices in $V_\text{C2G}$ are reached by vertices of $S\cap V_\text{C2G}$. Hence $|S\cap V_\text{C2G}|\geq \lceil \frac{2\alpha+2}{6} \rceil = \frac{\alpha+4}{3}$. Now $\{y, u_2,v_5,u_8,v_{11},\ldots u_{\alpha}\}$ is of size $\frac{\alpha+4}{3}$ exactly and reaches every vertex in $V_\text{C2G}$ and additionally reaches $x$ at time 2. Conversely if $\lambda(A)\neq 2$ then covering $V_\text{C2G}$ requires at most $ \frac{\alpha+6}{4}$ vertices, for example $\{y, u_2, u_6, u_{10}, \ldots, u_{\alpha}\}$.
  \end{proof}

\begin{lemma}\label{lem:2g_respect}
    For any happy function $\lambda$ which is \emph{not} 2-gadget-respecting, $(G,\lambda)$ admits a TaRDiS of size less than $k-1$. 
\end{lemma}

\begin{proof}
    Recall that $n=n_Y+n_X$ and $n>m$. Observe that the number of vertices in $G$ which are not inside a U2G or C2G is exactly:
    \begin{align*}
     & ~~~~ \#_\text{U3G}\cdot(4\beta+2) + \#_\text{U1G} + \#_\text{lit} + 8n_Y +  13n_X + 6m \\
     & = (2n_Y+6m)(4\beta+2) + 6m + 6m + 12n_Y + 13n_X + 16m \\
     & \leq 32n\beta + 31n. \\
    \end{align*}
    which is strictly smaller than $\frac{\alpha}{12}\approx 50n\beta + 50n$. 
    Even if there exists some $\lambda_\bot$ such that every vertex in $G$ not belonging to a U2G or C2G is necessarily included in every TaRDiS of $(G, \lambda_\bot)$, any such TaRDiS would still have size less than $\#_\text{U2G}\cdot\frac{\alpha+1}{3}+\#_\text{C2G}\cdot\frac{\alpha+4}{3}$ and hence less than $k-1$. 
  \end{proof}

\begin{lemma}\label{lem:U3G}
    For any 2-gadget respecting happy temporal assignment $\lambda$ on any U3G gadget incident to vertex $x$:
    \begin{itemize}
        \item If $\lambda(x, y_1)=3$ then $\beta$ vertices from $V_\text{U3G}$ are needed to temporally dominate $V_\text{U3G}\setminus\{y_1\}$, and no choice of $\beta$ vertices temporally dominates $V_\text{U3G}\cup \{x\}$.
        \item If $\lambda(x,y_1)\neq 3$, it is possible to temporally dominate $V_\text{U3G}\setminus\{y_1\}$ with exactly $\frac{\beta}{2}+1$ vertices. 
    \end{itemize}
    Recall $V_{\text{U3G}}=\{a_1, \ldots, a_{\beta-1}, b_1, \ldots b_{\beta-1}, y_1, y_2, y_3, u_1, \ldots, u_\beta, v_1, \ldots, v_\beta, w\}$; the set of vertices in the red (solid) box in Figure~\ref{fig:u3g}.
\end{lemma}

\begin{proof}
    There are only two possible 2-gadget-respecting assignments for the edges of $V_\text{U3G}$. In any case, for all $i\in[\beta]$, $\lambda(u_i,v_i)=\lambda(x,y_1)$. We denote by $\lambda_3$ the assignment where $\lambda(x,y)=3$ and $\lambda_1$ the temporal assignment where $\lambda(x,y)=1$. 

    It can be manually verified that, under $\lambda_3$, $\beta$ vertices are necessary to temporally dominate $V_\text{U3G}\setminus \{y_1\}$. This is also sufficient: the set $\{v_1, u_2, v_3, u_4, \ldots, u_\beta\}$ is one possibility.  On the other hand, under $\lambda_1$ there exists a set of size $\frac{\beta}{2}+1$ which temporally dominates $V_\text{U3G}\setminus \{y_1\}$, namely $\{u_1, u_2, u_4, \ldots, u_\beta\}$.
  \end{proof}

\begin{lemma}\label{lem:3g_respect}
    For any 2-gadget-respecting $\lambda$, if for any edge $e$ incident to a U3G $\lambda(e)\neq 3$ then $(G,\lambda)$ admits a TaRDiS of size less than $k-1$. 
\end{lemma}
\begin{proof}
    Observe that the number of vertices in $G$ which are not inside a U2G, C2G or U3G is exactly:
    \begin{align*}
     & ~~~~ \#_\text{U1G} + \#_\text{lit} + 8n_Y +  13n_X + 6m \\
     & = 6m + 6m + 8n_Y + 13n_X + 6m \\
     & = 18m + 8n_Y + 13n_X \\
     & \leq 31n \\
    \end{align*}
    which is strictly smaller than $\frac{\beta}{2}$. Even if there exists some 2-gadget-respecting $\lambda_\bot$ such that every vertex in $G$ not belonging to a U2G, C2G, or U3G is necessarily included in every TaRDiS of $(G, \lambda_\bot)$, any such TaRDiS would still have cardinality less than $\#_\text{U2G}\cdot\frac{\alpha+4}{3}+\#_\text{C2G}\cdot\frac{\alpha+4}{3}+\#_\text{U3G}\cdot(\beta)$ and hence less than $k-1$. Recall we assigned $k = \#_\text{U2G}\cdot\left(\frac{\alpha+1}{3}\right) + \#_\text{C2G}\cdot\left(\frac{\alpha+4}{3}\right) + \#_\text{U3G}\cdot(\beta) + 2m + 3n_X + 3n_Y + 1$, and that by our choice of $\beta$ we have $\frac{\beta}{2} > 2m + 3n_X + 3n_Y + 1$.  
  \end{proof}

\begin{lemma}\label{lem:nicelambda}
    Any happy temporal assignment $\lambda$ such that every TaRDiS on $(G,\lambda)$ has size at least $k-1$ is \emph{nice}.
\end{lemma}

\begin{proof}
    By Lemmas~\ref{lem:2g_respect} and~\ref{lem:3g_respect}, any happy temporal assignment $\lambda$ such that every TaRDiS on $(G,\lambda)$ has size at least $k-1$ is 2-gadget and 3-gadget-respecting. 
    Observe that if $\lambda$ is happy, 2- and 3-gadget-respecting, it must also be 1-gadget-respecting by the edge coloring constraint. 
    Hence $\lambda$ is \emph{nice}.
  \end{proof}

We say that a nice function $\lambda$ \emph{encodes} a truth assignment to the variables of $X$, and define this assignment as follows: let variable $x_i\in X$ be set to \texttt{True} if $\lambda(x_i^1,x_i^2)=1$ and \texttt{False} otherwise. In Figure~\ref{fig:xvar}, \texttt{True} corresponds to the case where $A=1$.

\smallskip
\begin{lemma}
    A nice function $\lambda$ encodes an assignment to $X$ under which $\Phi(X,Y)$ is satisfiable if and only if $(G,\lambda)$ admits a TaRDiS of size $k-1$. 
\end{lemma}

\begin{proof}
    We first prove the forward direction. That is, if $\lambda$ encodes an assignment $\mathcal{X}$ to $X$ under which $\Phi(X,Y)$ is satisfiable then $(G,\lambda)$ admits a TaRDiS of size at most $k-1$. 

    In this case, there is an assignment $\mathcal{Y}$ to $Y$ such that $\Phi(\mathcal{X,Y})$ evaluates to \texttt{True}. We have, by Lemmas~\ref{lem:U2G},~\ref{lem:C2G} and~\ref{lem:U3G}, that exactly $\#_\text{U2G}\cdot \frac{\alpha+1}{3} + \#_\text{C2G}\cdot \frac{\alpha+4}{3} + \#_\text{U3G}\cdot \beta$ vertices are necessary and sufficient to cover all the vertices in U2G, U3G, C2G and U1G gadgets, and these vertices also reach every vertex adjacent to a C2G at time 2 exactly. Let $S$ include exactly those vertices. 

    For each $Y$-variable gadget we add the vertices $v_i^1$ and $T_i^1$ to $S$, if $y_i$ is \texttt{True} in $\mathcal{Y}$, and vertices $v_i^1$ and $F_i^1$ are added otherwise. Further, let $S$ include vertices $x_i^4$ and $x_i^{10}$ from each $X$-variable gadget. Note that irrespective of how that gadget is labeled under $\lambda$, this choice of vertices is sufficient to reach all vertices not already covered from some vertex in $S$ within a C2G. Additionally, note that now every literal vertex corresponding to a literal set to \texttt{True} under the combined assignment to $X\cup Y$ is reached from $S$. Since this is a satisfying assignment, every set $Q_j$ of clause gadget vertices as shown in Figure~\ref{fig:mmt_clause} must be incident to at least one literal vertex which is already reached from $S$. By symmetry we may assume this literal vertex is $u$ (otherwise, cycle the labels $q_i^1,\ldots, q_i^6$ such that it is). Then let $S$ include vertices $q_j^3$ and $q_j^5$. Observe that $S$ is a TaRDiS of size precisely $k-1$; all gadget vertices, literal vertices, clause vertices and variable vertices are reachable from $S$. 

    We now prove the other direction. Namely, if $(G,\lambda)$ admits a TaRDiS of size at most $k-1$ then $\lambda$ encodes an assignment to $X$ under which $\Phi(X,Y)$ is satisfiable. Recall that, by Lemma~\ref{lem:nicelambda}, $\lambda$ must be nice. Let $S$ be a minimum TaRDiS of size at most $k-1$ in $(G,\lambda)$. We may assume, by Lemma~\ref{lem:canonical}, that $S$ is canonical. 

    Since $S$ is minimum, we have, by Lemmas~\ref{lem:U2G},~\ref{lem:C2G} and~\ref{lem:U3G}, that exactly $\#_\text{U2G}\cdot \frac{\alpha+1}{3} + \#_\text{C2G}\cdot \frac{\alpha+4}{3} + \#_\text{U3G}\cdot \beta$ vertices in $S$ cover all the vertices in U2G, U3G, C2G and U1G gadgets. Then at most $2m+3n_X+3n_Y$ vertices of $S$ are not among these, and must be sufficient to cover the remaining vertices of $G$.

    We define from $S$ a truth assignment to $Y$ as follows: if vertex $T_i^1$ or $T_i^2$ is in $S$, then we assign $y_i$ to be \texttt{True}, and we assign it to be \texttt{False} otherwise. We argue that this assignment together with $\mathcal X$ necessarily satisfies $\Phi(X,Y)$. Note that, since $S$ is canonical, it must contain at least 2 vertices from every clause gadget, 2 vertices from every $Y$-variable gadget, and 2 vertices from every $X$-variable gadget. All these bounds must be tight, else $S$ has cardinality more than $k-1$.

    Suppose for contradiction that $\Phi(X,Y)$ is not satisfied by the assignment. Then some clause $c_j$ contains only False literals under the assignment, and at least 3 vertices in $Q_j$ must be in $S$, contradicting that at most 2 vertices from any clause gadget may be in the TaRDiS. Hence there is an assignment to $Y$ satisfying $\Phi(X,Y)$.
  \end{proof}

\begin{lemma}\label{lem:corps_reducible}
    \textsc{co-RPS} is polynomial-time reducible to \textsc{Happy MaxMinTaRDiS}.
\end{lemma}

\begin{proof}   
    The construction above can be achieved in polynomial time. We have shown:
    \begin{itemize}
        \item The constructed \textsc{MaxMinTaRDiS} instance $(G,k,\tau=3)$ is a yes-instance if and only if there is some \emph{nice} $\lambda$ such that every TaRDiS for $(G,\lambda)$ has size at least $k-1$.
        \item Every \emph{nice} $\lambda$ encodes a truth assignment to $X$.
        \item A nice function $\lambda$ encodes an assignment to $X$ under which $\Phi(X,Y)$ is satisfiable if and only if $(G,\lambda)$ admits a TaRDiS of size $k-1$. 
    \end{itemize}
    
    That is, $(G,k,\tau=3)$ is a yes-instance of \textsc{MaxMinTaRDiS} if and only if $(X,Y,\Phi)$ is a yes-instance of \textsc{co-RPS}.

  \end{proof}

\begin{theorem}\label{thm:sigma}
    \textsc{Happy MaxMinTaRDiS} is $\Sigma_2^P$-complete even restricted to inputs where $\tau=3$ and the input graph $G$ is planar. 
\end{theorem}

\begin{proof}
    By Lemma~\ref{lem:corps_reducible} we have that \textsc{co-RPS} is polynomially reducible to \textsc{Happy MaxMinTaRDiS}. \textsc{co-RPS} is $\Sigma_2^P$-complete by Lemma~\ref{lem:corps_sigma2p}. The reduction above preserves planarity of $\Phi$. We have containment of \textsc{MaxMinTaRDiS} in $\Sigma_2^P$ from Lemma~\ref{lem:insigma}.
  \end{proof}

\subsection{NP-completeness of \textsc{Nonstrict MaxMinTaRDiS} with lifetime 2}\label{sec:NonStrictMTRDSA}
Here we consider the restriction of \textsc{Nonstrict MaxMinTaRDiS} to instances with lifetime 2. We show the problem to be equivalent to the \textsc{Distance-3 Independent Set} (\textsc{D3IS}) decision problem. We say that two problems $X$ and $Y$ are \emph{equivalent} if they have the same language - that is, an instance $I$ is a yes-instance of $X$ if and only if the same instance $I$ is a yes-instance of $Y$. Where $X$ has a language consisting of triples $(G,k,\tau)$ and $Y$ has a language of tuples $(G,k)$, we may say that $Y$ is equivalent to $X$ with $\tau$ fixed to some value.

We define the distance $d(u,v)$ between two vertices $u,v$ in a temporal graph to be the number of edges in the shortest path between them in the footprint of the graph.

\smallskip
\begin{definition}
    A distance-3 independent set (D3IS) of a static graph $H$ is a set $S\subseteq V(H)$ such that for all distinct $u,v \in S$, $d(u,v)\geq 3$.
\end{definition}
The decision problem \textsc{D3IS} is defined as follows.
\begin{framed}
    \noindent
    \textbf{\textsc{Distance-3 Independent Set (D3IS)}}\\
    \emph{Input:} A static graph $H=(V,E)$ and an integer $k$.\\
    \emph{Question:} Is there a set $S\subseteq V(H)$  of cardinality $k$ that is a distance-3 independent set?
\end{framed}
We aim to show that a static graph $H$ and integer $k$ are a yes-instance of \textsc{Nonstrict MaxMinTaRDiS} with lifetime 2 if and only if the same graph $H$ and integer $k$ are a yes-instance of \textsc{D3IS}. 

We begin by showing that existence of a maximal D3IS of size $k$ in a graph $H$ implies that we can find a temporal assignment $\lambda:E(H)\to\{1,2\}$ such that a minimum TaRDiS in $(H,\lambda)$ is of cardinality $k$. Given such a D3IS $S$ of $H$, we assign $\lambda(u,v)=1$, when $u\in S$ or $v\in S$, and $\lambda(u,v)=2$, otherwise.

\smallskip
\begin{lemma}\label{lem:dis2tardis}
    Let $S$ be a maximal D3IS of a static graph $H$ and $\lambda$ be a temporal assignment of $H$ where $\lambda(u,v)=1$ when $u\in S$ or $v\in S$, and $\lambda(u,v)=2$ otherwise. Then $S$ is a minimum TaRDiS of $(H,\lambda)$.
\end{lemma}

\begin{proof}
    We first show that $S$ is a TaRDiS. We assume without loss of generality that $H$ is a single connected component. Suppose for contradiction some vertex $u$ is not reachable from any vertex in $S$. Note that every vertex in $S$ trivially reaches its neighbours. So, by construction of $\lambda$, $u$ is incident only to edges at time $2$. Since we have assumed that $H$ consists of a single connected component, there must be a static path in $H$ from each vertex in $S$ to $u$. Let $z$ be the closest vertex in $S$ to $u$. Then the shortest path from $z$ to $u$ must have length 2 and consist of an edge assigned time 1 followed by an edge assigned time 2. Else, $S$ is not maximal. Hence, the shortest path from $z$ to $u$ is a nonstrict temporal path and $u\in R_z$, contradicting our assumption. Thus $S$ is a TaRDiS of the constructed instance.

    We now show minimality of $S$. By construction of $\lambda$, every vertex $v\in S$ is temporally reachable only from its closed neighbourhood $N[v]$. No temporal path originating outside $N[v]$ can include any edge incident to $v$ since any such path must contain an edge assigned time 2 before the final edge which is assigned time 1. Since $S$ forms a \textsc{D3IS}, $N[u]\cap N[v]=\emptyset$ for all $u$, $v\in S$. Therefore, for $(H,\lambda)$ to be temporal reachability dominated, there must at least be a vertex from the neighbourhood of each vertex in $S$. These are disjoint sets, so any TaRDiS must have cardinality at least $k$. Hence $S$ is a minimum TaRDiS of $(H,\lambda)$.
  \end{proof}

  \begin{definition}[Sole Reachability Set]
    We define the \emph{sole reachability set} of a vertex $v$ in a TaRDiS $S$ as the set $SR(\mathcal{G}, S,v)=R_v(\mathcal{G})\setminus(\cup_{u\in S\setminus \{v\}}R_u(\mathcal{G}))$. Equivalently, it is the set of vertices reachable from $v$ and not any other vertex in $S$.
\end{definition}
When $\mathcal{G}$ is clear from context, we write $SR(S,v)$ for $SR(S,\mathcal{G},v)$. Note that, in a minimum TaRDiS, every vertex has a non-empty sole reachability set.

\smallskip
\begin{definition}
    We call a TaRDiS $S$ on a temporal graph $\mathcal G$ \emph{independent} if and only if every vertex in $S$ is in its own sole reachability set under $S$.
\end{definition}

\smallskip
\begin{lemma}\label{lem:independent_tardis}
    If a temporal graph $\mathcal G$ admits an independent nonstrict TaRDiS $S$, then $S$ is a D3IS in the footprint graph $\mathcal G_\downarrow$.
\end{lemma}

\begin{proof}
    Consider two vertices $u,v\in S$ at distance $d(u,v)$ from one another in $\mathcal G_\downarrow$. If $u$ and $v$ are adjacent, then $u$ and $v$ reach each other, which would contradict independence of $S$. If $d(u,v)=2$ then there is at least one vertex $w\in N[u]\cap N[v]$, and either $\lambda(u,w)\leq \lambda(w,v)$ or $\lambda(u,w) > \lambda(w,v)$. So one of $u$ and $v$ must reach the other. Hence any two vertices in $S$ must be distance at least 3 from one another, the definition of a D3IS. 
  \end{proof}

\begin{lemma}\label{lem:tardis2dis}
    If a temporal graph $\mathcal{G}$ with lifetime $2$ has a minimum nonstrict TaRDiS of cardinality $k$, then $\mathcal{G}_\downarrow$ admits a D3IS of size $k$. 
\end{lemma}

\begin{proof}
    Our proof is constructive; given a minimum TaRDiS $S$, we show existence of an \emph{independent} minimum TaRDiS $S^*$ of equal cardinality and then apply Lemma~\ref{lem:independent_tardis}. 

    First, we justify some simplifying assumptions about $\mathcal G$. Since \textsc{TaRDiS} can be computed independently in disconnected components of $\mathcal G_\downarrow$, we will assume $\mathcal G_\downarrow$ is connected. Further, if $E_1(\mathcal G)= \emptyset$ or $E_2(\mathcal G)= \emptyset$ we may choose any single vertex from $\mathcal G$ to be a minimum TaRDiS which is also independent; hence we assume $E_1(\mathcal G)\neq \emptyset$ and $E_2(\mathcal G)\neq \emptyset$. We also recall that $SR(S,x) \neq \emptyset$ for all $x\in S$ by minimality of $S$. 

    We construct $S^*$ by replacing every vertex $x$ in $S$ such that $x\notin SR(S,x)$ with some vertex $x^*$ with the property that $R_{x^*}=R_x$ and $x^*\in SR(S,x)$, as follows. Let $y\neq x$ be a vertex in $S$ such that $y$ reaches $x$. The path from $y$ to $x$ cannot arrive at time 1, else $S$ is not minimal as $R_x= R_{y}$ and $S\setminus\{x\}$ is a TaRDiS. Choose $x^*$ to be the closest vertex to $x$ in $SR(S,x)$. We know such a vertex exists by minimality of $S$. We claim that the path from $x$ to $x^*$ arrives at time 1 and so $R_x=R_{x^*}$. To see this, suppose otherwise. The path must begin with at least one edge at time 1, otherwise $x^*$ would be reachable from $y$. If the path arrives at time 2, then the last vertex on the path reached at time $1$ is closer to $x$ than $x^*$, and is in $SR(S,x)$. Else, some other vertex in $S$ reaches $x^*$. 

    This concludes our construction of $S^*$ as an independent minimum TaRDiS. By Lemma~\ref{lem:independent_tardis}, $S^*$ is also a D3IS. 
  \end{proof}
Combining Lemmas~\ref{lem:dis2tardis} and~\ref{lem:tardis2dis} gives us the following theorem.
\smallskip
\begin{theorem}\label{thm:d3is}
    \textsc{Nonstrict MaxMinTaRDiS} with lifetime $\tau=2$ is equivalent to \textsc{Distance-3 Independent Set}.
\end{theorem}

Interestingly, the same does not hold for $\tau \geq 3$. A counterexample is shown in Figure~\ref{fig:intro_d3iscounter}, where the minimum TaRDiS is larger than the maximum D3IS of the footprint.
The Petersen graph also gives us a 3-regular counterexample.

Eto, Guo and Miyano \cite{eto_distance-dindependent_2014_bugfree} show that \textsc{D3IS} is NP-complete even on planar, bipartite graphs with maximum degree 3. They also show \textsc{D3IS} to be $W[1]$-hard on chordal graphs with respect to the size of the distance-3 independent set. \textsc{D3IS} is also shown to be APX-hard on $r$-regular graphs for all integers $r\geq 3$ and admit a PTAS on planar graphs by Eto, Ito, Liu, and Miyano \cite{eto_approximability_2016_bugfree}. Agnarsson, Damaschke and Halldórsson \cite{agnarsson_powers_2003} show that \textsc{D3IS} is tractable on interval graphs, trapezoid graphs and circular arc graphs. 

\smallskip
\begin{corollary}\label{cor:d3is}
 \textsc{Nonstrict MaxMinTaRDiS} restricted to inputs with $\tau=2$ remains NP-complete even when restricted to planar, bipartite graphs with maximum degree 3. Furthermore, the problem is $W[1]$-hard with respect to the parameter $k$, the size of a minimum TaRDiS, and is APX-hard on $r$-regular graphs for all integers $r\geq 3$.
\end{corollary}
\smallskip
\begin{corollary}
 \textsc{Nonstrict MaxMinTaRDiS} restricted to inputs with $\tau=2$ is in NP, admits a PTAS on planar graphs and is tractable on interval graphs, trapezoid graphs and circular arc graphs.
\end{corollary}

\section{Parameterized complexity results for \textsc{TaRDiS}}\label{sec:param-tardis}
In Section \ref{sec:classic-tardis}, we showed that the variants of \textsc{TaRDiS} are NP-complete and W[2]-hard with respect to $k$. This rules out $k$ as a candidate parameter for an fpt algorithm. We begin by showing inclusion of \textsc{TaRDiS} in FPT with respect to locally earliest edges and tractability when the input is heavily restricted. We then give an algorithm which solves each variant of \textsc{TaRDiS} on a nice tree decomposition of the footprint graph. This generalises the tree algorithm given in Section~\ref{sec:classic-tardis}.

A problem $\Pi$ is said to be \emph{fixed-parameter tractable} (fpt) with respect to some parameter $k$ if there is an algorithm solving $\Pi$ in time $f(k) \cdot \mathrm{poly}(n)$ (where $n$ is the size of the instance of $\Pi$). The complexity class FPT consists of all problem-parameter pairs admitting an fpt algorithm. 

\subsection{FPT results with a restricted temporal assignment}
We begin with two FPT results that require the temporal assignment to be restricted in some way to recover tractability of \textsc{TaRDiS}. In the first, we consider \textsc{Happy TaRDiS} and \textsc{Nonstrict TaRDiS} when parameterised by the number of locally earliest edges. Following that, we consider \textsc{Strict TaRDiS} when each component of the input is restricted in some way.

% contains stuff formerly in the prelims
% may have outside subsection
\smallskip
\begin{lemma}\label{lem:fpt_lee}
    \textsc{Happy TaRDiS} and \textsc{Nonstrict TaRDiS} are in FPT with respect to the number of locally earliest edges and weakly locally earliest edges, respectively.
\end{lemma}
\begin{proof}
    It suffices to observe that any instance with $t$ (weakly) locally earliest edges trivially admits a TaRDiS of cardinality $t$ (Corollary \ref{cor:lee_balev}). If $k\geq t$, we must have a yes-instance. 
    Otherwise there are $\binom{t}{k}$ possibilities for a (weakly) canonical TaRDiS, which can be checked by brute force in polynomial time when $t$ is bounded. 
  \end{proof}

\begin{lemma}\label{lem:fpt_finite_lang_tardis}
     \textsc{Strict TaRDiS} is in FPT parameterized by maximum degree in $\mathcal{G}_\downarrow$, lifetime and $k$, and \textsc{Happy TaRDiS} is in FPT with parameters lifetime and $k$.
\end{lemma}

\begin{proof}
    Any decidable language consisting of words of length at most some constant $c$ can be decided in constant time, hence is in $P$. Any strict temporal path in a temporal graph $\mathcal G$ has length at most $\tau$. This entails that, for all $v$, $|R_v^{<}|\leq 2 \Delta^\tau$ where $\Delta$ is the maximum degree of $\mathcal G_\downarrow$. Hence any strict TaRDiS $S$ in $\mathcal G$ must satisfy $|S| \geq \frac{V(\mathcal{G})}{2 \Delta^\tau}$.  Therefore, no instance $(\mathcal G, k)$ satisfying $|V(\mathcal{G})|\geq 2k\Delta^\tau$ can be in the language \textsc{Strict TaRDiS}. Note that, in a happy temporal graph, we can can apply the property that $\tau \geq \Delta$.
    Therefore, both problems restricted as above have constantly many yes-instances, each of size bounded by a constant. 
  \end{proof}

\subsection{Preliminaries: treewidth and tree decompositions}
%LLJ: Think we don't need below especially with reference to fpt book in the following paragraph
% \begin{definition}[Definition 1.2 \cite{cygan_parameterized_2015_bugfree}]
%     A parameterized problem $L\subseteq \Sigma^*\times \mathbb{N}$ is called \emph{fixed parameter tractable} if there exists an algorithm $\mathcal{A}$, a computable function $f:\mathbb{N}\to \mathbb{N}$, and a constant $c$ such that, given $(x,k)\in \Sigma^*\times \mathbb{N}$, the algorithm $\mathcal{A}$ correctly decides whether $(x,k)\in L$ in time bounded by $f(k)\cdot|(x,k)|^c$. We refer to the $\mathcal{A}$ as an fpt algorithm. The class containing all fixed-parameter tractable problems is FPT. 
% \end{definition}

We refer the interested reader to Chapter 10 of Niedermeier's \emph{Invitation to Fixed-Parameter Algorithms} \cite{niedermeier_invitation_2006_bugfree} for a fuller introduction. Definitions and results in this subsection are taken or adapted from that work.

\smallskip
\begin{definition}[Tree Decomposition, Treewidth]
    We say a pair $(T,B)$ is a \emph{tree decomposition} of $G$ if $T$ is a tree and $B = \{B(s) : s\in V(T)\}$ is a collection of subsets of $V(G)$, called \emph{bags}, satisfying:
    \begin{enumerate}
        \item $V(G)=\cup_{s\in V(T)}{B(s)}$.
        \item $\forall (u,v) \in E(G): \exists s \in V(T): \{u,v\}\in B(s)$. That is, for each edge in the graph, there is at least one bag containing both of its endpoints.
        \item $\forall v\in V(G): T[\{s:v\in B(s)\}]$ is connected; for each vertex, the subgraph obtained by deleting every node not containing $v$ in its bag from $T$ is connected.
    \end{enumerate}
The \emph{width} of a tree decomposition is defined to be $\max\{|B(s)|:s\in V(T)\}-1$. The \emph{treewidth} of a graph $G$ is the minimum $\omega$ such that $G$ has a tree decomposition of with $\omega$. 

\end{definition}

For a given tree decomposition $(T,B)$ of graph $G$, we denote by $V_s \subseteq V(G)$ the set of vertices in $G$ that occur in bags of the subtree of $T$ rooted at $s$. It is a well-known result by Bodlaender that finding a tree decomposition of width at most $\omega$, if one exists, is in FPT with parameter $\omega$. 

\begin{theorem}[Bodlaender \cite{bodlaender_linear-time_1996_bugfree}]\label{thm:tree-decomp-linear}
    For all fixed constants $\omega\in \mathbb N$, there exists a linear time algorithm that tests whether a given graph $G=(V,E)$ has treewidth at most $\omega$, and if so outputs a tree decomposition of $G$ with treewidth at most $\omega$. 
\end{theorem}

For ease, we describe our \textsc{TaRDiS} algorithm on a tree decomposition with additional structural properties.

\smallskip
\begin{definition}[Nice Tree Decomposition, Join/Introduce/Forget/Leaf Node]
A tree decomposition $(T,B)$ is called a \emph{nice} tree decomposition if:
\begin{itemize}
    \item every node of $T$ has at most two children;
    \item if a node $s$ has two children $s_l$ and $s_r$, then $B(s)=B(s_l)=B(s_r)$, and we call $s$ a \emph{join node};
    \item if node $s$ has one child $s'$ then either:
    \begin{itemize}
        \item $|B(s)|=|B(s')|+1$ and $B(s) \supset B(s')$ ($s$ is an \emph{introduce node}), or
        \item $|B(s)|=|B(s')|-1$ and $B(s) \subset B(s')$ ($s$ is a \emph{forget node});
    \end{itemize}
    \item if node $s$ has no children then $B(s)=\emptyset$, and we call $t$ a \emph{leaf node}.
\end{itemize}
\end{definition}
\begin{figure}
    \centering
    \includegraphics[width=\textwidth, page=2]{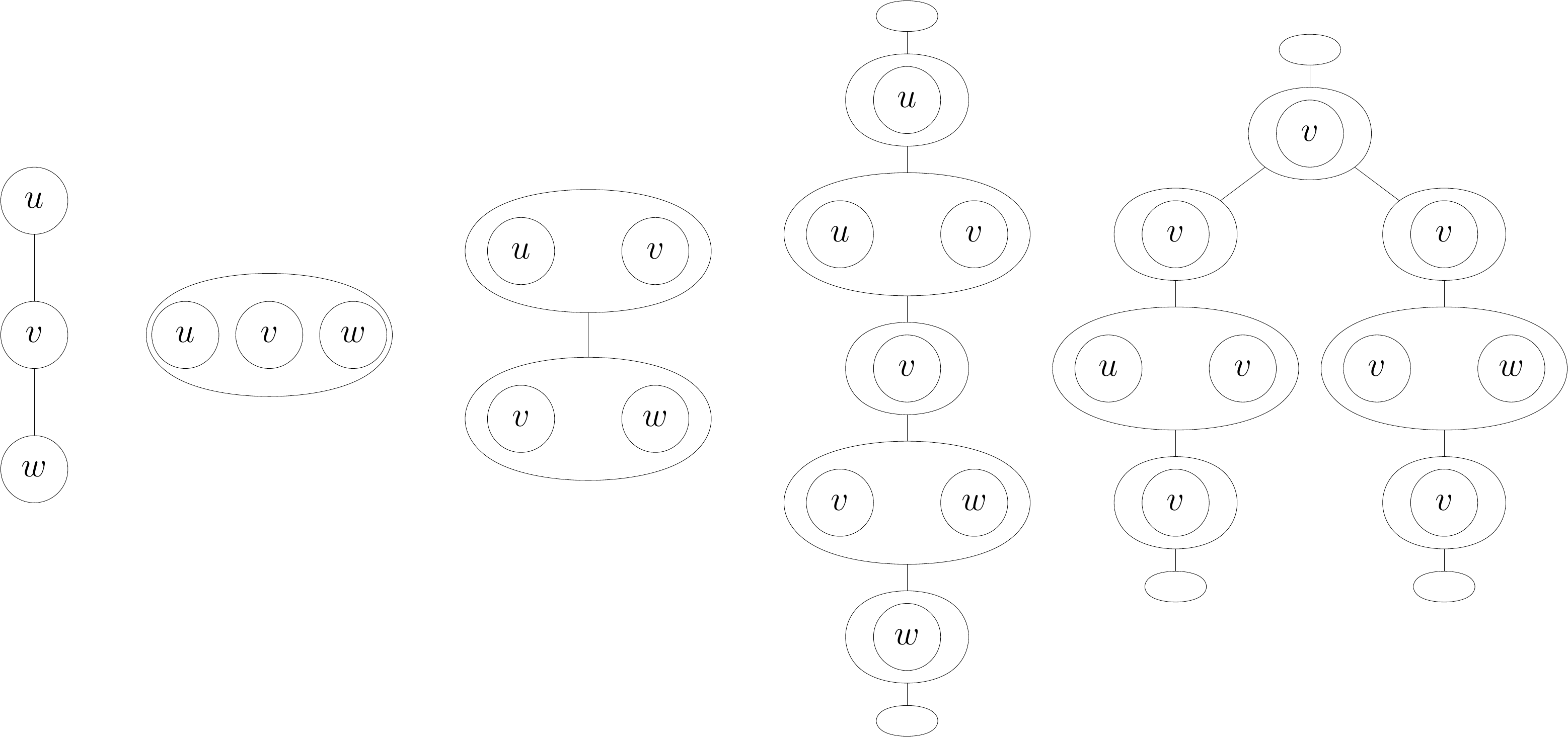}
    \caption{Graph $P_3$ and four of its tree decompositions. From left to right: $P_3$ ; the trivial decomposition of width $2$; a decomposition of width $1$; two possible nice tree decompositions of width $1$. Note that every graph admits the trivial tree decomposition of width $V(G)-1$, where there is only one bag $B(s)=V(G)$.}
    \label{fig:tree_decomposition_intro}
\end{figure}

An example of a graph and different tree decompositions of it can be found in Figure~\ref{fig:tree_decomposition_intro}. It is a known result that a tree decomposition $(T,B)$ of a graph $G$ of width $\omega$ on $n$ nodes can be efficiently processed to obtain a nice tree decomposition $G$ of width $\omega$. We use this result as given below by Cygan, Fomin, Kowalik, Lokshtanov, Marx, Pilipczuk, Pilipczuk and Saurabh \cite{cygan_parameterized_2015_bugfree}.
\smallskip
\begin{lemma}[Cygan \cite{cygan_parameterized_2015_bugfree}]\label{lem:nice-convert-linear}
    If a static graph $H$ admits a tree decomposition of width at most $\omega$, then it also admits a nice tree decomposition of width at most $\omega$. Moreover, given a tree decomposition $(T,B)$ of width at most $\omega$, on can compute a nice tree decomposition of $H$ of with at most $\omega$ that has at most $O(\omega V(H))$ nodes in time $O(\omega^2\max (V(T),V(H))$.
\end{lemma}

\subsection{Algorithm for \textsc{TaRDiS} parameterized by treewidth and lifetime}\label{sec:tree-decomp}

The following gives an algorithm for \textsc{TaRDiS} given a nice tree decomposition $(T,B)$ of the footprint $\mathcal G_\downarrow$ of the input temporal graph $\mathcal G$. Specifically, it computes the cardinality of a minimum TaRDiS $S$. Note that we use the word ``vertex'' when referring to the original graph and ``node'' when discussing the decomposition graph. We use the symbols $\prec$ and $\succ$ as place-holders for strict/nonstrict inequalities. More specifically, in \textsc{Strict TaRDiS} we use $<$ and $>$ in the place of $\prec$ and $\succ$ respectively and for \textsc{Nonstrict TaRDiS} we use $\leq$ and $\geq$. By substituting the correct inequalities, it can be seen that the algorithm described is correct for each of \textsc{Strict TaRDiS} and \textsc{Nonstrict TaRDiS}.  

Denote by $G_s$ the subgraph $(V_s,E_s)$, where $V_s$ (respectively $E_s$) is the set of all vertices (resp. edges) introduced in the subtree rooted at node $s$ of the decomposition tree. Intuitively, the algorithm works from leaf nodes to the root node and finds, at each node $s$, a partial solution consisting of a minimal TaRDiS $S$ for the subgraph $G_s$.

We define the \emph{arrival time} of a temporal path at a vertex $v$ to be the time of the final time-edge of the path. For the trivial path from a vertex to itself, we say that the time of arrival is 0. Similarly, we define the \emph{departure time} of a temporal path as the time of the first time-edge in the path. For example, the temporal path consisting of a single time-edge has the same departure and arrival time. We note that, for a path to be a strict temporal path, the arrival time at a vertex $v$ must be strictly before the departure time from $v$. We allow the arrival and departure times to be equal in nonstrict temporal paths of any length. We refer to a temporal path from a vertex $v$ to a vertex $u$ with earliest arrival time as a \emph{foremost} temporal path. We call the arrival time of a foremost path the \emph{foremost arrival time}. A foremost path from a set of vertices $S$ to a single vertex $v$ is a foremost path from a vertex $u$ in $S$ to $v$ such that the arrival time of a path from any other vertex in $S$ to $v$ is the same or later.

\subsubsection{States}
A state of a bag $B(s)$ is a mapping $\psi: B(s)\to ([0,\tau]\cup \perp) \times ([0,\tau]\cup \perp)$ where $\psi(v)=(t_a(v),t_p(v))$, such that $t_a(v) \geq t_p(v)$ if both values are integers and $t_p(v)=\bot$ only if $t_a(v)=\bot$. Conceptually, we think of $t_a(v)$ as the arrival time of some path to $v$ from our partial TaRDiS and $t_p(v)$ as the time by which we ``promise'' arrival of such a path from the eventual full TaRDiS (necessary to reach forgotten vertices). Denote the set of all states of a node $s$ by $\Psi(s)$. 
 
In our definition of a temporal graph, we assume that all edges are active at strictly positive times. Therefore, in both the strict and nonstrict variants of the problem, the earliest time any vertex can be temporally reachable from a vertex in the TaRDiS is 1, unless it is in the TaRDiS itself. Our intention is that, if $t_a(v) \neq \bot$ then $t_a(v)$ is exactly the time that $v$ is reached from some TaRDiS vertex. If $t_a(v)=0$ for a vertex $v$ in the bag, this corresponds to $v$ being included in the partial TaRDiS $S$. We use the notation $t_{a}^{-1}$, $t_{p}^{-1}$ to denote the preimage of the functions $t_{a}$ and $t_{p}$ respectively. That is, $t_{a}^{-1}(x)$ is the set of all vertices $v$ which are assigned $(x,t_p(v))$ under a state $\psi$. 

We call a state $\psi$ of a bag $B(s)$ \emph{consistent} if and only if there exists a set of vertices $S\subseteq V(G_s)$ such that
\begin{itemize}
    \item for all $v\in B(s) \setminus t_a^{-1}(\bot)$ the foremost temporal path from some vertex in $S$ arrives at $t_a(v)$ exactly;
    \item for any vertex $w$ in $V(G_s)\setminus B(s)$ which is not reachable from some vertex in $S$, there is a temporal path to $w$ from some vertex $u$ in $B(s)$ with departure time $t$ such that $t\succ t_p(u)$.
\end{itemize}
We call such a set $S$ a set that \emph{supports} $\psi$ of $s$. Note that, if $t_a(v)=\perp$, then it is possible to have a valid state where $v$ is reachable from some vertex in $S$.
%We add a dummy vertex $v_{\text{source}}$ to $S$. For each vertex $v$ with $t_p(v)\neq \perp$, we add the time-edge $(vv_{\text{source}},t_p(v))$. This is so we can assume that all vertices in the bag with $t_p(v)\neq\perp$ are reached by a vertex in $S$ by the time $t_p$. For ease, $v_{\text{source}}$ is not added to the bag $B(s)$. 
\subsubsection{Signature}
For a state $\psi$ of a node $s$, we define the \emph{signature} $c(s,\psi)$ to be the cardinality of the smallest set $S^*$ which supports $\psi$ of $s$. If there is no such $S^*$, then we say that $c(s,\psi)=\infty$. We say that such a set $S^*$ \emph{supports} the signature $c(s,\psi)$ if $S^*$ supports $\psi$ of $s$ and $|S^*|=c(s,\psi)$. The signature is a data structure with size $O(\tau^{2(\omega+1)})$ where $\omega$ is the width of the nice tree decomposition of $\mathcal{G}_\downarrow$. Note that two bags may contain the same (possibly empty) set of vertices. Therefore, the signature must be both a function of the state $\psi$ and of the node $s$ which $\psi$ is a state of. Examples of the signatures of different states can be seen in Figure~\ref{fig:signature_intro}.
% We define the signature $c(s,\psi)$ of a state of $B(s)$ as $|S|-1$ where $S$ is a set of vertices $v\in G_s$ of smallest cardinality such that 
% \begin{itemize}
%     \item $v_{\text{source}}\in S$,
%     \item $S\cap B(s)=t_a^{-1}(0)$,
%     \item all vertices $v$ in $B(s)\setminus t_a^{-1}(\perp)$ are reached by a temporal path from $S\setminus \{v_{\text{source}}\}$ arriving by $t_a(v)$, and
%     \item all vertices in $G_s\setminus B(s)$ are temporal reachability dominated by $S$.
% \end{itemize}

\begin{figure}
    \centering
    \includegraphics[width=\textwidth, page=9]{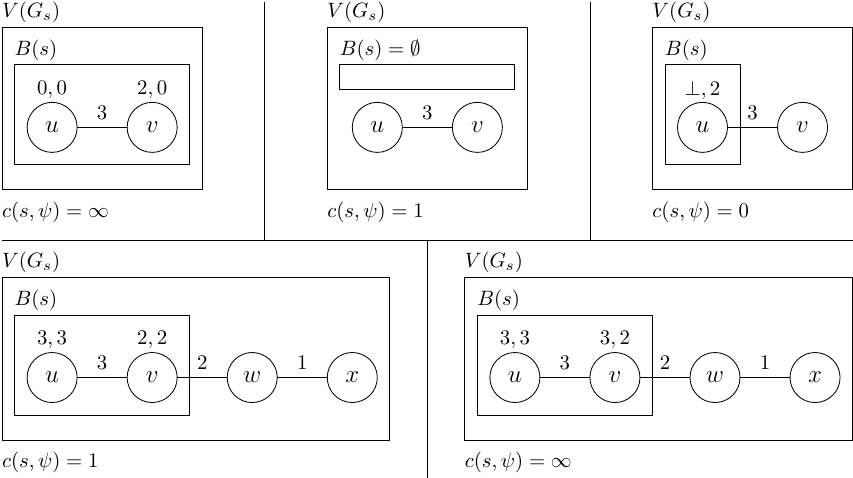}
    \caption{Five example bags, labelled with states and signatures. In the bottom right example, the infinite signature reflects the fact there exists no TaRDiS reaching $v$ at time $3$ exactly.}
    \label{fig:signature_intro}
\end{figure}

We use the convention that the root node $r$ is empty, i.e. $B(r)=\emptyset$. Consequently there is only one possible state for the root node, namely the empty function, and hence $\Psi(r)=\{\emptyset\}$.  Therefore, we have a yes-instance of TaRDiS if and only if $c(r,\emptyset)\leq k$. Note that $\emptyset$ is always a consistent state for $r$. It is supported by the trivial TaRDiS $S=V(G)$. 

We now discuss how we iteratively calculate the signature of a state for each type of node in a nice tree decomposition.
\subsubsection*{Leaf nodes} 
We assume that leaf nodes $l$ are empty, so there is only one trivial state, namely the empty function $\emptyset$. The signature of this state for leaves $l$ is $c(l,\emptyset)=0$.

\subsubsection*{Introduce Nodes}
Let $s$ be an introduce node with child $s'$. Then we must have that $B(s)=B(s')\cup\{v\}$ for some $v\notin B(s')$. To describe how to find $c(s,\psi)$ for an introduce node $s$ and state $\psi$, we must define some new notation.

Let $\psi|_{B(s')}$ be the restriction of some state $\psi$ of $s$ to the bag $B(s')$. For a state $\psi$ and functions $g,f$, let the state $\psi^{t_a(A)\to g, t_p(B)\to h}$ be defined
\[
    \psi^{t_a(A)\to g, t_p(B)\to h}(x):= 
    \begin{cases}
        (g(x),h(x))& \text{if } x\in A\cap B\\
        (g(x),t_p(x))& \text{if } x\in A\setminus B\\
        (t_a(x),h(x))& \text{if } x\in B\setminus A\\
        (t_a(x),t_p(x)),             & \text{otherwise.}
    \end{cases}
    \] 
% We will combine this function with the restriction of $\psi$ to get $\psi|_{B(s')}^{t_a(w)\to t}$, the restriction of $\psi$ with only the values of $t_a(w)$ and $t_p(w)$ changed when for a vertex $w\in B(s')$. 

For an introduce node $s$ with introduced vertex $v$ and state $\psi$ of $s$, we define $a_v$ to be the time of the earliest time-edge in $\{((v,u),t)\,|\, u \in B(s)$ and $t_a(u)\prec t$ under $\psi\}$. If no such edge exists, let $a_v=\infty$. Notionally, this is a time before which $v$ cannot be reached from a vertex in $S$ unless $v$ is itself in $S$. 
% Let $A_v$ be the set of neighbours $u$ of $v$ such that there exists the time-edge $(vu,a_v)$ and $t_a(u)\prec a_v$.
% Similarly, we define $b_v$ to be the time of the latest time-edge $(vu,t)$ incident to $v$ such that, under $\psi$, $t<t_a(u)$ where $u$ is the other endpoint of the time-edge. If no such time-edge exists, let $b_v=\infty$. If $v$ is reached by a vertex in $S$ after this time, it will have no effect on the foremost arrival time of a path from a vertex in $S$ to any of its neighbours. 
% Let $\lambda_{\succ t}(uv)$ be the time $t'$ of the earliest appearance of the edge $uv$ such that $t' \succ t$.  We assign $\lambda_{\succ t}(uv)$ the value $\perp$ if there is no appearance $t'$ of $uv$ such that $t' \succ t$. We define $t_{\prec}(uv,t^*)$ to be the latest time $t$ such that there exists an appearance of the edge $uv$ at time $t'$ where $t\prec t'\leq t^*$. If no such $t'$ exists, let $t_{\prec}(uv,t^*)=\perp$. Then $t_{\prec}(uv,t^*)$ is the latest time a temporal path could reach either endpoint \todo[inline]{just say $u$, have $v$ by symmetry anyway?} and be appended by a time-edge to give a temporal path that arrives at $t^*$. 
We define $R^t_v$ to be the set of vertices that are temporally reachable from $v$ by temporal paths which depart at a time $t'\succ t$. We use the convention that $R^{\infty}_w=R^\bot_w=\emptyset$ for all vertices $w$. For a vertex $u$ in $R^t_v$, define $\text{foremost}^t_v(u)$ to be the arrival time of a foremost path from $v$ to $u$ which departs at a time $t'\succ t$. An example of a valid state of an introduce node can be seen in Figure~\ref{fig:tree_decomposition_introduce}.
% We define the function $r_{s,\psi,v,t}$ to be \todo[inline]{could def $r_t$, other inputs always are $s,\psi,v$}
% $$r_{s,\psi,v,t}(u)=\left\{
%     \begin{array}{ll}
%       \psi(u) & \text{if }\lambda_{\succ t}(uv)=\perp,\\
%       \left(t_a(u), \min\{t_p(u), \lambda_{\succ t}(vu)\}\right) & \text{otherwise.}
%     \end{array} \right.$$
%     \todo[inline]{LLJ: we could be taking the min of $\perp$. thoughts on rectifying?}
% Note that we only change the $t_p$ component of the pairs, allowing us to restrict the function to $s'$ without contradicting the requirement that every vertex is reachable by time $t_a$ from a set $S$ supporting the state.
\begin{figure}
    \centering
    \includegraphics[width=.6 \textwidth, page=4]{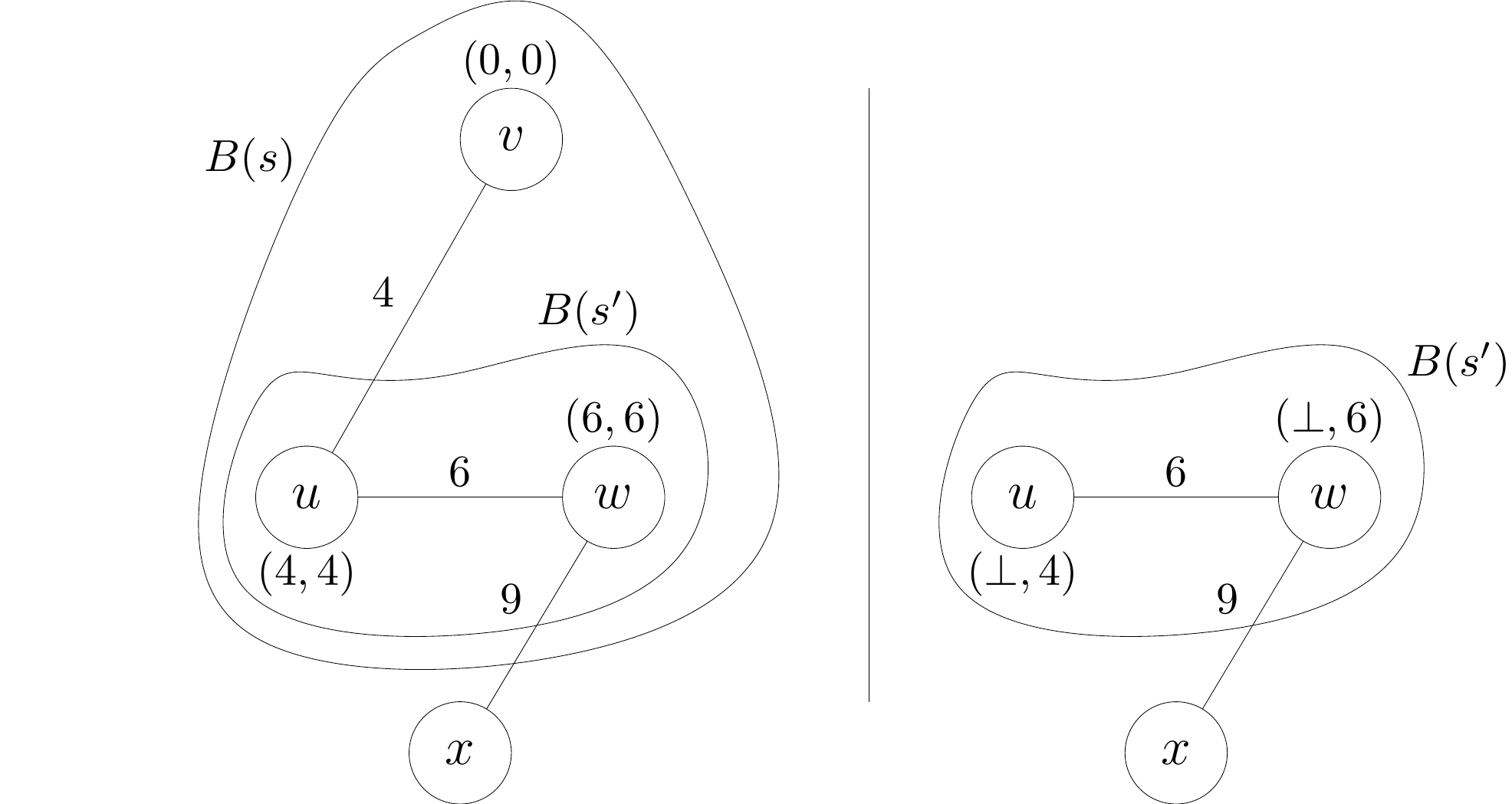}
    \caption{A valid state of an introduce node and a state of its child which are supported by the same partial TaRDiS. Here $x$ is a forgotten vertex.}
    \label{fig:tree_decomposition_introduce}
\end{figure}

\smallskip
\begin{lemma}\label{lem:decomp-intro}
    Let $\psi$ be a state of an introduce node $s$ with child $s'$ where $v$ is the vertex introduced. Define $Z$ to be the set $(R^{t_a(v)}_v(\mathcal{G})\cap B(s))\setminus \{u: \text{foremost}^{t_a(v)}_v(u)> t_a(u)\}$, and define the function $f(w):=\min\{t_p(w),\text{foremost}^{t_p(v)}_v(w)\}$. Then
    \begin{enumerate} 
        \item if, for any $u \in R^{a_v}_v(\mathcal{G})\cap B(s)$, $t_a(u) \neq \bot$ and $\textnormal{foremost}^{a_v}_v(u)< t_a(u)$, then $c(s, \psi)=\infty$ (a foremost path arrives before it is prescribed to); 
        \item else, if $t_a(v)=0$, $c(s,\psi)=1+c\left(s',\psi|_{B(s')}^{t_a(Z)\to \perp,t_p(B(s'))\to f(B(s'))}\right)$ (the introduced vertex is added to the partial TaRDiS);
        \item else, if we allow nonstrict paths (i.e. if $1 \prec 1$), $a_v = t_a(v)$ and there exists a nonempty set $W$ of neighbours of $v$ where, for all $w\in W$, $t_a(w)=t_a(v)= \lambda(vw)$ then $c(s,\psi)=\min_{w\in W}\left\{c\left(s',\psi|_{B(s')}^{t_a(Z')\to \perp,t_p(B(s'))\to f(B(s'))}\right)\right\}$ for $Z'=Z\setminus \{w\}$ (a pair of neighbours each rely on the other to be reached on time); 
        \item else, if $t_a(v)\in \{a_v, \bot\}$ then $c(s,\psi)=c\left(s',\psi|_{B(s')}^{t_a(Z)\to \perp,t_p(B(s'))\to f(B(s'))}\right)$ (the introduced vertex could result in a foremost path with an earlier arrival time);
        \item else,
        $c(s,\psi)=\infty$ (we have an inconsistent state).
        %\min\left\{\infty, \{c\left(s',\psi|_{B(s')}^{t_a(w)\to t_{\prec}(wv,t_a(v))}\right): w\in N(v)_{G_s}\text{ and }t_{\prec}(wv,t_a(v))\neq \perp\}\right\}  where $N(v)_{G_s}$ is the neighbourhood of $v$ in $G_s$;
    \end{enumerate}
    % $$c(s,\psi)=\left\{
    % \begin{array}{ll}
    %   \min\left\{\infty, \{c\left(s',\psi|_{B(s')}^{t_a(w)\to t_{\prec}(wv,t_a(v))}\right): w\in N(v)_{G_s}\text{ and }t_{\prec}(wv,t_a(v))\neq \perp\}\right\} & \text{if } 0<t_a(v)< a_v, \\
    %   c(s',\psi|_{B(s')}) & \text{if }t_a(v)=\perp\text{ or } t_a(v)\nprec b_v,\\
    %   c(s',r_{s,\psi,v,t_a(v)}) & \text{if }a_v \leq t_a(v)\prec b_v\\
    %   1+c(s',r_{s,\psi,v,0}) & \text{if }t_a(v)=0.
    % \end{array} \right.$$
\end{lemma}
\begin{proof}
    We begin by noting that for any set $S$ supporting $\psi$ of $s$ if the foremost arrival time to $v$ from a vertex in $S$ is $t_a(v)$, then $\text{foremost}^{t_a(v)}_v(u)<t_a(u)$ contradicts that $t_a(u)$ gives the earliest time of arrival from a vertex in $S$ to $u$. Otherwise, we could append the temporal path to $v$ from $S$ with the path from $v$ to $u$ to find an temporal path with an earlier arrival time. Hence, in case (1) of the Lemma, $\psi$ is an inconsistent state.

    We have characterised the value of $c(s,\psi)$ based on the value of $t_a(v)$ under $\psi$. This gives us four further cases ((2)-(5)) to consider, accounting for all possible values of $t_a(v)$. 
    We show that our calculation of the signature $c(s,\psi)$ is correct for each possible value of $t_a(v)$ for the introduced vertex $v$. Let $S$ be a set that supports the signature of the state $\psi$ of node $s$. That is, it is a set of minimal cardinality such that $t^{-1}_a(0)=S\cap B(s)$; all vertices $w$ in $B(s)\setminus t^{-1}_a(\perp)$ are temporally reachable from $S$ by a foremost path arriving at $t_a(w)$; and each vertex in $G_s\setminus B(s)$ not reachable from $S$ is temporally reachable from a vertex $u$ in $B(s)$ departing at some time $t\succ t_p(u)$. Note that, for a consistent state, if there is a forgotten vertex reachable from a vertex $w$ in $B(s)$ by a path departing at time $t\succ t_p(w)$, the change in $t_p$ values ensures that this remains true in the child states over which we take the minimum value.

    In case (2), $t_a(v)=0$. Additionally, for all $u\in R^t_v$, $t_a(u)\leq \text{foremost}^{t_a(v)}_v(u)$ or $t_a(u)=\perp$, since otherwise we would have case (1). Hence for any set $S$ supporting $\psi$, $v\in S$ and all vertices in $R^0_v$ are reached by time $\text{foremost}^0_v(u)$. We claim that a set $S$ supports $\psi$ if and only if $S\setminus \{v\}$ supports the state $\psi'= (t_a',t_p'):=\psi|_{B(s')}^{t_a(Z)\to \perp,t_p(B(s'))\to f(B(s'))}$ of $s'$. That is, every vertex $w$ in $B(s)\setminus R^0_v$ is temporally reached by a vertex in $S\setminus \{v\}$ at $t_a'(w)$ and every vertex in $V(G_{s})\setminus B(s)$ that is not temporally reachable from $S\setminus\{v\}$ is temporally reachable from a vertex $y$ in $B(s')$ by a path departing at time $t'\succ t_p'(y)$. It is clear that, if states $\psi$ and $\psi'$ are valid, the former statement is true. This is because the value of $t_a$ for a vertex in $B(s)\setminus R^0_v$ is the same under both states.
    
    To verify the latter statement, note that if a vertex $y'$ in $V(G_{s})\setminus B(s)$ is reachable from a vertex $y$ in $B(s')$ by a path departing at $t'\succ t_p'(y)$ under $\psi'$, then the same must be true under $\psi$. By construction of a nice tree decomposition, there are no vertices in $V(G_s)\setminus B(s)$ that are adjacent to $v$. Therefore, any path from $v$ to a vertex in $V(G_s)\setminus B(s)$ must traverse another vertex $x$ in $B(s)$ and depart $x$ at time $t''\succ t_p'(x)$ under $\psi'$. Thus, if a vertex $y'$ in $V(G_{s})\setminus B(s)$ is reachable from a vertex $y$ in $B(s')$ by a path departing at $t'\succ t_p'(y)$ under $\psi$, then $y'$ must be reachable from a (possibly different) vertex $y''$ under $\psi'$. Therefore, a set $S$ supports $\psi$ if and only if $S\setminus \{v\}$ supports the state $\psi'=\psi|_{B(s')}^{t_a(Z)\to \perp,t_p(B(s'))\to f(B(s'))}$ of $s'$ and our calculation of $c(s,\psi)$ is correct.

    Intuitively, cases (3) and (4) deal with the cases where the introduction of $v$ could change the earliest arrival time of a path from a vertex in $S$ to a vertex in $B(s)\setminus \{v\}$. In case (3) we deal with the possibility of nonstrict temporal paths when $t_a(w)=\text{foremost}^{t_a(v)}_v(w)=t_a(v)$ for some vertex $w$. In this case, a child state wherein $t_a(w)=\perp$ for every such $w$ may be supported by a set $S$ which does not support $\psi$. For this reason, we take the minimum over the signatures of states where the $t_a$ value of each such neighbour is unchanged. A set $S$ supports $\psi$ if and only if there exists a $w\in W$ such that $S$ supports $\psi|_{B(s')}^{t_a(Z')\to \perp,t_p(B(s'))\to f(B(s'))}$. The forward implication is clear from our description. Now suppose, for contradiction, that $S$ supports only $\psi|_{B(s')}^{t_a(Z')\to \perp,t_p(B(s'))\to f(B(s'))}$. Then there is either a vertex $u$ in $R^{t_a(v)}_v$ which is not reached from $S$ at time $t_a(u)$ or there is a forgotten vertex which is not reachable from $S$ or a vertex $w$ in $B(s)$ departing at $t'\succ t_p(w)$. Neither case is possible, and thus we have a contradiction. 

    In case (4), the above does not apply. This is either because we are in the strict setting or because there is no such nonempty set $W$, and $t_a(v)\in\{a_v,\bot\}$ under $\psi$. Then, for all vertices $w$ in $Z=(R^{a_v}_v(\mathcal{G})\cap B(s))\setminus\{u: \text{foremost}^{t_a(v)}_v(u)> t_a(u)\}$, there exists a foremost path from a set $S$ that supports the state which traverses $v$. Therefore, any set $S$ which supports $\psi$ must support a child state where these vertices $w\in Z'$ have $t_a$ value $\bot$ and $t_p$ is updated as mentioned earlier. In addition, it is clear that any set which supports $\psi|_{B(s')}^{t_a(Z)\to \perp,t_p(B(s'))\to f(B(s'))}$ must support $\psi$. Since we have not added any vertices to $S$, the signature of $s$ under $\psi$ must be exactly the signature of $s'$ under $\psi|_{B(s')}^{t_a(Z)\to \perp,t_p(B(s'))\to f(B(s'))}$.

    In case (5) we deal with all remaining possibilities, namely when $t_a(v)$ for the introduced vertex $v$ has a nonzero value and $t_a(v)\neq a_v$. Then $t_a(v)$ must be strictly before or after the earliest time-edge $((v,w),t)$ incident to $v$ such that, for the other endpoint $w$, $t_a(w)\prec t$. Therefore, for any set $S$ that supports $\psi$, if the foremost temporal path from a vertex in $S$ to each neighbour $w$ of $v$ arrives at $t_a(w)$, the foremost temporal path to $v$ from $S$ must arrive at a time which is not equal to $t_a(v)$. This implies that the state is inconsistent and must therefore have an infinite signature. $ $
\end{proof}

\subsubsection*{Forget Nodes}
Let $s$ be a forget node with child $s'$ such that $B(s)=B(s')\setminus\{v\}$. Let $\Psi_\text{strong}\subset\Psi(s')$ be the set of states which extend $\psi$ to $B(s')$ where $t_a(v)\in[0,\tau]$ and $t_a(v)= t_p(v)$. Intuitively, these are the states where our partial TaRDiS already reaches the forgotten node $v$ by the promised time $t_p(v)$, i.e. the promise is \emph{strongly} satisfied. 

Let $t'$ be the earliest time such that there exists a vertex $w$ in $B(s)$ and a temporal path in $B(s')$ from $w$ to $v$ departing at some time $t$ such that $t_p(w)\prec t$ under $\psi$ and arriving by $t'$. Let $\Psi_\text{weak}\subset\Psi(s')$ be the set of states which extend $\psi$ to $B(s')$ where either $t' \leq t_p(v)$ or $t_p(v)=\perp$. Intuitively, $\Psi_\text{weak}$ is the set of states where the the requirement that a forgotten vertex is reached by a path departing at time $t\succ t_p(v)$ is automatically satisfied by a path from $w$ departing at $t\succ t_p(w)$ for some $w\neq v$. That is, the promise is \emph{weakly} satisfied. An example of a valid state of a forget node can be seen in Figure~\ref{fig:tree_decomposition_forget}.

% rededi \alpha definition - old version here: Let $\alpha_v^{t'}\subset\Psi(s')$ be the set of states which extend $\psi$ to $B(s')$ where either $t'\prec t_p(v)$ \todo[inline]{here too $t'\leq t_p(v)$, not $\prec$ (?)}or $t_p(v)=\perp$. 
\begin{figure}
    \centering
    \includegraphics[width=.6\textwidth, page=5]{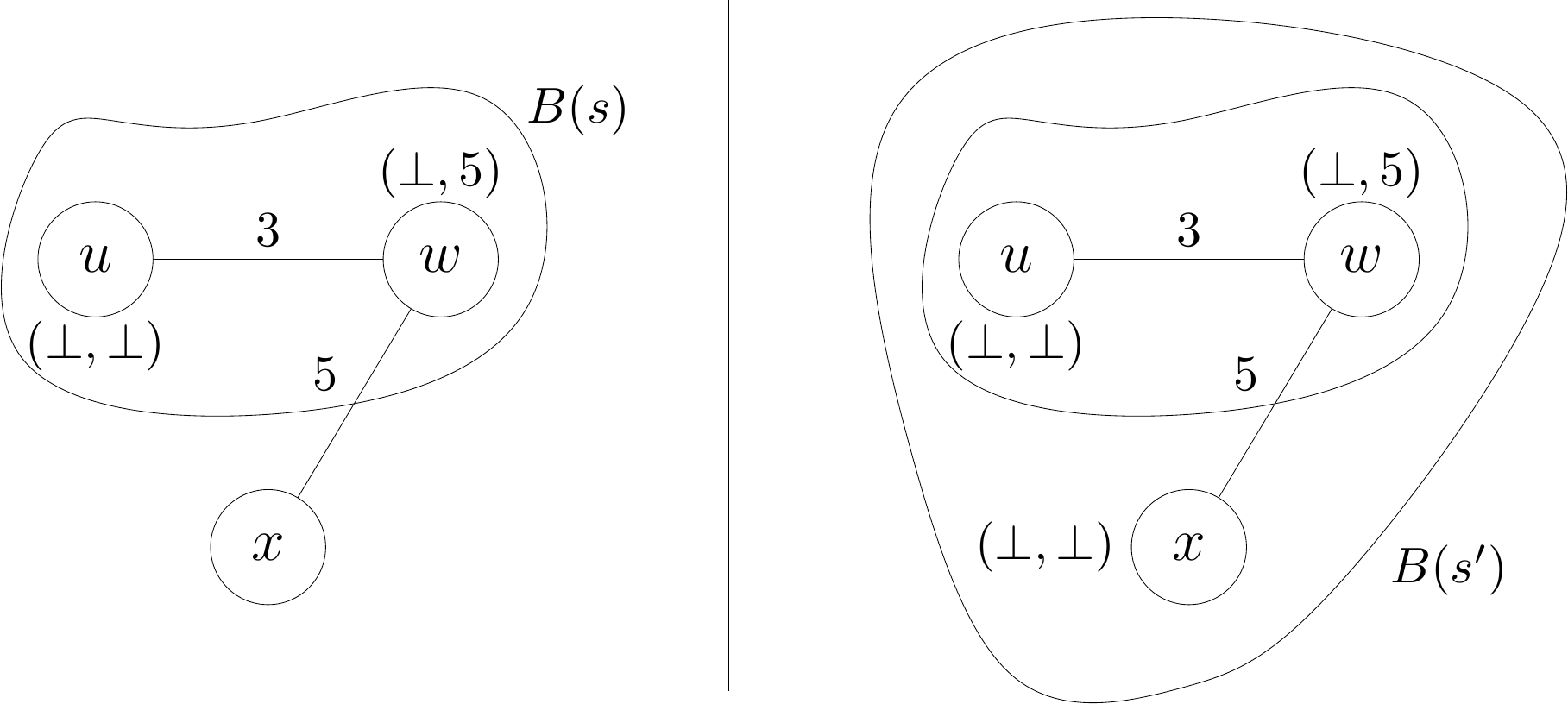}
    \caption{A forget node with a valid state and its extension to a state of its child supported by the same partial TaRDiS.}
    \label{fig:tree_decomposition_forget}
\end{figure}

\smallskip
\begin{lemma}\label{lem:decomp-forget}
    Let $s$ be a forget node with state $\psi$ and child $s'$ where $v$ is the vertex forgotten at $s$.  Then
    $$c(s,\psi)=\min\left ( \{c(s',\psi'):\psi'\in \Psi_\text{weak} \cup \Psi_\text{strong} \}  \cup \{\infty\}\right).$$
\end{lemma}

\begin{proof}
We begin by showing that if a set $S$ supports any state $\psi' \in \Psi_\text{weak} \cup \Psi_\text{strong}$ of $s'$ then the same set $S$ supports $\psi$. This shows us that the signature of $\psi$ is at most what we have calculated. Following this, we show that there is no smaller set that supports $\psi$ and hence that we have calculated the signature correctly.

Recall that set $S$ supports a state $\psi$ of a node $s$ if and only if, under $\psi$:

\begin{itemize}
    \item For all vertices $u$ in $B(s)\setminus t^{-1}_a(\perp)$, the foremost path from $S$ arrives at time $t_a(u)$, and
    \item for all vertices $x$ in $G_s\setminus B(s)$ which are not reachable from $S$, there is a temporal path to $x$ which departs from a vertex $w$ in $B(s)$ at some time $t \succ t_p(w)$.
\end{itemize}

% \todo[inline]{Case 1: we have $\psi'\in \alpha$}
To start, consider the case where $S$ supports $\psi'\in\Psi_\text{weak}$. That is,
there exists a vertex $w\in B(s')$ and a temporal path from $w$ to $v$ that departs at some time $t' \succ t_p(w)$ and arrives by time $t_p(v)$. As a result, any forgotten vertex $x$ temporally reachable from $v$ by a path departing at some time $t'' \succ t_p(v)$ must be temporally reachable from $w$ by a path departing at time $t'$ by prefixing the path from $v$ to $x$ with the temporal path from $w$ to $v$. If $\psi'$ is consistent, then all vertices $y$ in $B(s')\setminus\{v\}=B(s)$ that are reachable from a vertex in $S$ are reached by time $t_a(y)$ if $t_a(y)\neq \perp$. Therefore, the restriction of $\psi'$ to $B(s)$ is consistent and supported by $S$.

% \todo[inline]{Case 2: we have $\psi''\in A$}

Now consider the case where $S$ supports a state $\psi'\in \Psi_\text{strong}$. Recall that $\Psi_\text{strong}$ is the set of states of $s'$ where $t_a(v)\in[0,\tau]$ and $t_a(v)= t_p(v)$. If such a state $\psi'$ is consistent for $s'$, then $v$ is reached from a vertex in $S$ by time $t_a(v)$. Using that $t_p(v)=t_a(v)$, we obtain that any forgotten vertex $x$ reachable from $B(s)$ by a path from $v$ departing at some time $t'\succ t_p(v)$ is reachable from $S$ as well. Namely, we can prefix the temporal path from $v$ to $x$ which departs at time $t'$ with the temporal path from $S$ to $v$ which arrives by time $t_a(v)= t_p(v)$. Hence we obtain that $\psi$ is consistent for $s$ and supported by $S$. 

% \todo[inline]{Either case: our signature is correct}
We now suppose for contradiction that our calculation of the signature is incorrect. We have already shown that the signature must be at most $|S|$ for some set $S$ supporting a state in $\Psi_\text{strong}\cup \Psi_\text{weak}$. Therefore, we suppose that there is a smaller set $S^*$ which supports $\psi$ whose intersection with $B(s)$ is exactly the set of vertices $t_a^{-1}(0)$ under $\psi$. 

Consider the state $\psi^*$ extending $\psi$ to $B(s')$ in which $t_a(v)$ is the arrival time of the foremost path from $S^*$ to $v$  (or $\bot$ if there is no such path) and $t_p(v)=\tau$. Clearly $\psi^* \in \Psi_\text{weak} \cup \Psi_\text{strong}$, and $S^*$ supports $\psi^*$. This contradicts our assumption that $S^*$ was smaller than any set supporting a state in $\Psi_\text{weak} \cup \Psi_\text{strong}$.
 \end{proof}

\subsubsection*{Join Nodes}
Let $s$ be a join node with children $s_1$ and $s_2$. 
\smallskip
\begin{lemma}\label{lem:decomp-join}
    Let $\psi_1=(t_{a,1},t_{p,1})$ and $\psi_2=(t_{a,2},t_{p,2})$ be states of the children $s_1$, $s_2$ of a join node $s$. We say that $\psi_1$ and $\psi_2$ \emph{coincide} with the state $\psi$ if
\begin{itemize}
    \item $\psi(v)=(0,0)$ if and only if $\psi_1(v)=(0,0)$ and $\psi_2(v)=(0,0)$;
    \item for $i\in \{a,p\}$ and all $v\in B(s)$, 
    \[
   t_i(v)= 
    \begin{cases}
        \min\{t_{i,1}(v),t_{i,2}(v)\}& \text{if } t_{i,1}(v),t_{i,2}(v)\neq \perp,\\
        t_{i,1}(v)& \text{if } t_{i,2}(v)=\perp\neq t_{i,1}(v),\\
        t_{i,2}(v)& \text{if } t_{i,1}(v)=\perp\neq t_{i,2}(v),\\
        \perp,             & \text{otherwise.}
    \end{cases}
    \] 
\end{itemize}
Then we calculate $c(s,\psi)$ by
$$c(s,\psi)=\min_{\psi_1,\psi_2}\{c(s_1,\psi_1)+c(s_2,\psi_2)-|t_a^{-1}(0)|\}.$$
where the minimum is taken over all pairs of states $\psi_1$, $\psi_2$ which coincide with $\psi$.
\end{lemma}
\begin{proof}
    We begin by showing that a state $\psi$ of a join node $s$ is consistent if and only if there are consistent states $\psi_1$ and $\psi_2$ of its children which coincide with $\psi$. We then give a proof of correctness of our calculation of the signature.

    Suppose that the states $\psi_1$ and $\psi_2$ are consistent for the nodes $s_1$ and $s_2$ respectively and they coincide with the state $\psi$. Then, let $S$ be the union of sets $S_1$ and $S_2$ which support $\psi_1$ and $\psi_2$ respectively. By our assumption of consistency, all vertices in $V(G_s)\setminus B(s)$ are temporally reachable from a vertex in $S \cup B(s)$ and a foremost temporal path from a vertex in $S$ arrives at each vertex $u$ in $B(s)$ at time $t_a(u)$. In addition, if any forgotten vertex $w$ which is not reached from $S$ is temporally reachable from a vertex $v\in B(s)$ by a path departing at some time $t\succ t_{p,j}(v)$ for $j\in\{1,2\}$, then $w$ must be temporally reached from $v$ by a temporal path departing at some time $t\succ t_p(v)$. Therefore, $\psi$ must be a consistent state of $s$.

    We now assume that $\psi$ is a consistent state of $s$ supported by the set $S$. Suppose, for a contradiction, that $S$ supports $\psi$ and there are no child states supported by $S$ which coincide with $\psi$. There must be states for which $t_a$ describes the earliest time of arrival of a path from $S$ to any vertex in $B(s_1)$ and $B(s_2)$. Therefore, if there do not exist states $\psi_1$ and $\psi_2$ supported by $S$, then there must be a vertex $w$ in $G_{s_1}\setminus B(s_1)$ or $G_{s_2}\setminus B(s_2)$ which is not temporally reachable from a vertex in $S \cup B(s)$. If this is the case, then $\psi$ must also be inconsistent. This is because, by construction of a tree decomposition, the only (temporal) path from a vertex $G_{s_1}\setminus B(s_1)$ to $G_{s_2}\setminus B(s_2)$ or vice versa must traverse at least one vertex in $B(s)=B(s_1)=B(s_2)$. Therefore, if there is a temporal path from $S$ or $B(s)$ to each vertex in $G_s\setminus B(s)$, there must be a temporal path from $S\cup B(s_1)$ and $S\cap B(s_2)$ to each vertex in $G_{s_1}\setminus B(s_1)$ and $G_{s_2}\setminus B(s_2)$ respectively. Therefore, there exist child states supported by $S$.
    %
    % Suppose a vertex $u$ is reached at time $t_a(u)$ from a vertex in $S$. Then, the foremost arrival time 
    % We can then construct consistent states $\psi_1$ and $\psi_2$ of $s_1$ and $s_2$ respectively which are consistent with $\psi$. We find $S_1$ and $S_2$ by taking $S\cap G_{s_1}$ and $S\cap G_{s_2}$ respectively. Then, for $i\in\{1,2\}$ and all $u\in B(s)$, $\psi_i(u)=(t_a^i(u),t_p(u))$ where $t_a^i(u)$ is the time of arrival of a foremost temporal path from a vertex in $S_i$ if one exists and $\perp$ otherwise. There must be at least one of $i\in\{1,2\}$ such that $t_a(u)=t_a^i(u)$ and at most one such that $t_a(u)<t_a^i(u)$ or $t_a^i(u)=\bot$. If any forgotten vertex which is not reached from $S$ is temporally reachable from a vertex $u\in B(s)$ by a path departing at or after $t_p(u)$, then this is still the case when we consider the children of $s$. Therefore, $\psi_1$ and $\psi_2$ coincide with $\psi$. By consistency of $\psi$, $\psi_1$ and $\psi_2$ must also be consistent.

    The smallest set $S$ that supports a consistent state of $s$ must therefore be the smallest set that is the union of sets $S_1$ and $S_2$ which support consistent states $\psi_1$, $\psi_2$ respectively of its children $s_1$ and $s_2$ which coincide. To find the cardinality of this set, we employ the inclusion-exclusion principle. The only vertices which $G_{s_1}$ and $G_{s_2}$ have in common are those in $B(s)$. Therefore, $c(s,\psi)=\min_{\psi_1,\psi_2}\{c(s_1,\psi_1)+c(s_2,\psi_2)-|t_a^{-1}(0)|\}$ where the minimum is taken over pairs of states $\psi_1, \psi_2$ which coincide with $\psi$.
 \end{proof}

\subsubsection{Running Time and Extensions}

We note that for every vertex $v$ in a bag, there are strictly fewer than $(\tau+2)^2$ values of $\psi(v)$. The number of vertices in a bag is bounded by treewidth $\omega$ plus 1. Thus, there are less than $(\tau+2)^{2{(\omega+1)}}$ states per node. We now explore the running time of calculating the signature at each type of node. For leaf nodes, the signature can only be one value; therefore this can be found in constant time. If we have the signatures of all child states of an introduce node and we want to verify validity of a given state of that node, we need to compute the value of $a_v$ for the introduced vertex $v$, the reachability set $R_v^{t_a(v)}$ and the earliest time of arrival at each other vertex in the bag from $v$ by a path departing at times $t \succ t_a(v)$ and $t'\succ t_p(v)$. 

The reachability set $R_v^t(\mathcal G)$ can be computed in time $O(n^2)$ by a modified breadth first search algorithm for any $v$ and $t$. Hence it takes $O(\tau n^3)$ time to calculate $R_v^t(\mathcal G)$ for every vertex $v\in V$ and time $t\in [\tau]$. This can be done as a preprocessing step before we begin working up the tree decomposition.

The earliest time of arrival of the two paths can be computed using a variation of breadth-first search and thus takes $O(\omega^2)$ time. The value $a_v$ is computable in $O(\omega)$ time. In addition, for some cases in \textsc{Nonstrict TaRDiS} we find the minimum signature of restrictions of the state where we change the value of $t_a$ for some neighbours of $v$. Assuming all signatures of descendant nodes are calculated, this takes $O(\omega)$ time. Therefore, computing the signature of all states of an introduce node requires $O(\omega^2 (\tau)^{2{(\omega+1)}})$ time.

For forget nodes, we must compare the signatures of multiple child states. There are $O(\tau^2)$ of these extensions for each state of the forget node. In addition, we compute the earliest time of arrival of a path from any vertex in $B(s)$ to $v$ in the subgraph induced by $B(s)$ with restrictions on its departure time. This can be achieved by a variant of a breadth-first search which takes at most $O(\omega^3)$ time. Thus, finding the signature of all states of a forget node requires $O(\omega^3 (\tau)^{2\omega+4})$ time. 

Finally, to compute the signature of the state of a join node, we must compare the states of both child nodes which coincide with this state. For a given vertex in the bag of a join node, the number of values of $\psi_1$ and $\psi_2$ which coincide with $\psi(v)$ are bounded by $O(\tau^4)$. Therefore, there are $O(\tau^{4(\omega+1)})$ tuples of states to consider when calculating the signature of $\psi$ on $s$. Therefore, we calculate the signature of a state of a join node in $O(\tau^{4(\omega +1)}\cdot \omega)$ time. Note that this dominates the time needed to generate all possible states for a given bag.

We now combine the lemmas from this section to get the following theorem. We can find a tree decomposition of width at most $\omega$ if one exists in linear time by Theorem~\ref{thm:tree-decomp-linear}. Lemma~\ref{lem:nice-convert-linear} states that we can find a nice tree decomposition of width $\omega$ given a tree decomposition of width $\omega$ in linear time. We can recursively compute the signature of a given state of a node using Lemmas~\ref{lem:decomp-intro},~\ref{lem:decomp-forget} and~\ref{lem:decomp-join}. We solve \textsc{TaRDiS} by finding the signature $c(r,\psi)$ of the root where the state $\psi$ is the empty function which gives the cardinality of a minimal TaRDiS.
\smallskip
\begin{theorem}\label{thm:tree-decomp-tardis}
    The algorithm described takes as input a temporal graph $\mathcal{G}$ consisting of $n$ vertices with a nice tree decomposition of width at most $\omega$ and solves \textsc{Strict} and \textsc{Nonstrict TaRDiS} on $\mathcal{G}$ in time $O(\tau^{4({\omega+1})}\cdot \omega\cdot n + \tau n^3)$.
\end{theorem}
We emphasize that $\tau n^3$ is polynomial in the input size because $\tau \le \mathcal{E}$ (recall, no snapshot in a temporal graph is empty). The algorithm allows for edges to be active multiple times. That is, it is not restricted to simple temporal graphs.

\section{Parameterized complexity results for \textsc{MaxMinTaRDiS}}\label{sec:param-maxmintardis}

Having shown \textsc{MaxMinTaRDiS} is in $\Sigma^P_2$, finding instances of tractability is even more surprising than with the variants of the \textsc{TaRDiS} problem. We begin the following result, closely related to Lemma \ref{lem:fpt_finite_lang_tardis} which gives us tractability when each component of the input to \textsc{Strict MaxMinTaRDiS} is restricted.

\smallskip
\begin{lemma}\label{lem:fpt_finite_lang}
     \textsc{Strict MaxMinTaRDiS} is in FPT when parameterized by maximum degree $\Delta$ in $H$ and $k$.
\end{lemma}
\begin{proof}    
    Recall from Lemma \ref{lem:mmtar2ds} that $(H,k)$ is a yes-instance of \textsc{Strict MaxMinTaRDiS} if and only if $(H,k-1)$ is a no-instance of \textsc{Dominating Set}. Also, any pair $(H, k)$ where $H$ has maximum degree $\Delta$ satisfying $|V(H)| > k(\Delta+1)$ is trivially a no-instance of \textsc{Dominating Set}. Applying the same reasoning as in our proof of Lemma \ref{lem:fpt_finite_lang_tardis}, we obtain that \textsc{Strict MaxMinTaRDiS} is solvable in polynomial time when $\Delta$ and $k$ are bounded. If the input graph has at least $k(\Delta+1)$ vertices then output YES, and otherwise solve the problem (necessarily of bounded size) by brute-force.
  \end{proof}

\subsection{Existence of an fpt algorithm for \textsc{MaxMinTaRDiS}}\label{sec:courcelle}
We show tractability of our problems by expressing them in EMSO logic and applying the variant of Courcelle's theorem given by Arnborg, Lagergren and Seese \cite{arnborg_easy_1991_bugfree}. This result states that an optimisation problem which is definable in EMSO can be solved in polynomial time when parameterized by treewidth and length of the formula. The theorem is as follows.

\smallskip
\begin{theorem}[adapted from \cite{langer_practical_2014_bugfree}, Theorem 30]\label{thm:courcelle}
    Let $P$ be an EMSO-definable problem, then one can solve $P$ on graphs $G=(V,E)$ of order $n:=|V|$ and treewidth at most $w$ in time $O(f_{P}(w)\cdot \mathrm{poly}(n))$.
\end{theorem}

An extended monadic second order formula (EMSO) over a static graph $H$ is a formula that uses:
\begin{enumerate}
    \item the logical operators $\vee$, $\wedge$, $\neg$, $=$ and parentheses;
    \item a finite set of variables, each of which takes an element of $V(H)$ or $E(H)$;
    \item the quantifiers $\forall$ and $\exists$;
    \item a finite set of variables which take subsets of the sets of edges or vertices;
    \item integers.
\end{enumerate}
We make use of the predicates $\neq$, $\implies$, $\impliedby$, $\iff$, $\subseteq$, $\in$ and $\setminus$, which can be implemented using the above. We note that formulas consisting only of the first 3 components is a first-order formula (FO). A formula which is first order with the addition of the fourth bullet point is referred to as an MSO formula. A more formal definition of an EMSO-definable problem is given by a survey by Langer, Reidl, Rossmanith and Sikdar \cite{langer_practical_2014_bugfree}.

\smallskip
\begin{theorem}\label{thm:emso}
    \textsc{MaxMinTaRDiS} is fixed-parameter tractable when parameterized by lifetime, $k$ and treewidth of the graph.
\end{theorem}
\begin{proof}
    The formal definition of EMSO given by Langer Reidl, Rossmanith and Sikdar \cite{langer_practical_2014_bugfree} requires a weight function bounded by a constant. Our weight function will be the cardinality of a TaRDiS, which is bounded by $k$. Since the temporal assignment is not part of the input, we encode it as a partition of edges into sets which correspond to the time at which they are active. The EMSO formula is constructed using the following auxiliary subformulae
\begin{itemize}
    \item $\text{card}(k,X)$ tests whether a set $X$ has cardinality at least $k$:
$$\text{card}(k,X):=\exists x_1,\ldots,x_k \in X :\bigwedge_{1\leq i<j\leq k}x_i\neq x_j.$$
    \item $\text{geq}(X_1,X_2)$ tests whether $|X_1|\geq|X_2|$ for sets $X_1$ and $X_2$:
    $$\text{geq}(X_1,X_2):=\exists k: \text{card}(k,X_1)\wedge\neg\text{card}(k+1,X_2).$$
    \item $\text{part}(S_1,\ldots, S_{\tau})$ tests whether the sets of edges $S_1\ldots S_{\tau}$ partition the edges of $H$:
\begin{align*}
    &\text{part}(S_1,\ldots,S_{\tau}):= \forall e\in E: \\
    &\left(\bigvee_{1\leq i\leq \tau} e\in S_i\wedge\left(\bigwedge_{1\leq j<i} e\not\in S_j\right)\wedge\left(\bigwedge_{i<\ell\leq\tau } e\not\in S_{\ell}\right)\right).
\end{align*}
The two right-most brackets can be ignored if we do not require that the temporal assignment is simple. To enforce a happy temporal assignment, we can add the requirement that no two edges in the same set share and endpoint. Recall that tractability of \textsc{Happy MaxMinTaRDiS} is shown with respect to lifetime and $k$ combined in Lemma~\ref{lem:fpt_finite_lang}. That is a stronger result than what we have here since Courcelle's theorem only gives tractability on graphs with bounded treewidth. 

The following subformulae can be adapted to write \textsc{TaRDiS} in MSO logic.
\item $\text{mconn}(X,S_t)$ tests whether all vertices in a set $X$ are in the same connected component of $G_{t}$.
        $$\text{mconn}(X,S_t):= \forall Y \subset X, Y\neq\emptyset,\exists x\in X,\exists y\in Y\setminus X:xy\in S_t$$
        \item $\text{mvconn}(v,w,t)$ tests whether two vertices are in the same connected component of $G_t$:
        $$\text{mvconn}(v,w,t):= \exists X\subset V : v\in X\wedge w\in X\wedge \text{mconn}(X,S_t).$$
    \item $\text{mtadj}$ tests whether two vertices $v$, $w\in V$ are adjacent at time $t$:
$$\text{mtadj}(v,w,S_t):=\exists e\in S_t : v\in e\wedge w\in e.$$
    \item $\text{mpath}$ tests whether there is a temporal path from $v$ to $w\in V$ with latest time $\tau$:
    \begin{align*}
        \text{mpath}(v,w,S_1,\ldots,S_{\tau}):=\exists v_0,\dots,v_\tau\in V:\\
        v=v_0\wedge v_{\tau}=w\wedge\bigwedge^{\tau-1}_{t=0}(v_t=v_{t+1}\vee \text{a}(v_{t}, v_{t+1},S_{t+1}))
    \end{align*}
    where $\text{a}(v,w,S_t)$ can be substituted for $\text{mvconn}(v,w,t)$ or $\text{mtadj}(v,w,S_t)$ depending on whether we are testing for nonstrict or strict temporal paths respectively.
    \item $\text{TaRDiS}(S_1,\ldots,S_{\tau},X)$ which tests if every vertex is temporally reachable from $X$ by a temporal path of lifetime $\tau$:
$$\text{TaRDiS}(S_1,\ldots,S_{\tau},X):=\forall v\in V(H), \exists s\in X:\text{mpath}(s,v,S_1,\ldots,S_{\tau}).$$
    \item $m\text{TaRDiS}(S_1,\ldots,S_{\tau},X)$ tests whether a set $X$ is a minimum TaRDiS:
    \begin{align*}
        m\text{TaRDiS}(S_1,\ldots,S_{\tau},X):=& \forall X'\subset V(H):\textsc{TaRDiS}(S_1,\ldots,S_{\tau},X)\wedge\\
        &(\textsc{TaRDiS}(S_1,\ldots,S_{\tau},X')\implies \text{geq}(X',X)).
    \end{align*}
    \item $\textsc{MinTaRDiS}(H,\tau,k)$ which tests if there is temporal assignment with lifetime at most $\tau$ such that there exists a minimum TaRDiS of size at least $k$ on $H$:
    \begin{align*}
        &\textsc{MinTaRDiS}(H,\tau,k):=\exists X\subset V, \exists S_1,\ldots, S_{\tau}\subset E:\\
        &\text{part}(S_1,\ldots,S_{\tau})\wedge m\textsc{TaRDiS}(S_1,\ldots,S_{\tau},X)\wedge\text{card}(k,X).
    \end{align*}
\end{itemize}

Therefore \textsc{MaxMinTaRDiS} can be expressed in EMSO and the theorem holds.
  \end{proof}
\begin{corollary}\label{cor:MSOreach}
    Temporal reachability is expressible in MSO logic.
\end{corollary}
Furthermore, we can use the expression to express \textsc{TaRDiS} in EMSO. This, however, gives a weaker tractability result than Theorem~\ref{thm:tree-decomp-tardis}.

\section{Conclusions and open questions}\label{sec:conclusion}

In this paper, we introduce the \textsc{TaRDiS} and \textsc{MaxMinTaRDiS} problems and study their (parameterized) complexity. We show a bound on the lifetime $\tau$ and a restriction to planar inputs combined are insufficient to obtain tractability (Theorems~\ref{thm:happytardis},~\ref{thm:sigma}, and Corollary~\ref{cor:d3is}) and moreover tightly characterize the minimum lifetime $\tau$ for which each problem becomes intractable. Further, we give an algorithm on a nice tree decomposition of a temporal graph which gives tractability of \textsc{TaRDiS} with respect to lifetime and treewidth of the footprint of the graph. In addition, we show that $\tau$, $k$ and the treewidth of the input graph combined are sufficient to yield tractability in all cases of \textsc{MaxMinTaRDiS} by leveraging Courcelle's theorem. 

These results leave open the questions of the exact complexity of \textsc{Nonstrict MaxMinTaRDiS} with lifetime $\tau \geq 3$ and \textsc{Happy MaxMinTaRDiS} with lifetime $\tau \geq 4$. An interesting extension of our work would be to find approximability results for these problems. From the parameterized side, it remains to be shown whether parameterization by a structural parameter of the footprint (e.g. treewidth) alone is sufficient to obtain tractability for any of the considered variants. 
Another interesting dimension is the comparison of \textsc{Nonstrict MaxMinTaRDiS} and \textsc{Happy MaxMinTaRDiS} when $\tau$ is lower-bounded by a function of the number of edges $m$. With the constraint $\tau = m$ the two problems become equivalent, and their computational complexity in this case is an interesting open question.
Analogously to \textsc{$t$-Dominating Set} \cite{kneis_partial_2007_bugfree}, \textsc{$t$-TaRDiS}, in which $t$ individuals must be reached provides a natural generalisation of our problem and the potential for parameterization by $t$.
Recall that, in \textsc{Strict MaxMinTaRDiS}, it is always optimal to choose the constant function as our temporal assignment $\lambda$. It may be interesting to consider restrictions on $\lambda$ other than happiness which require the use of a non-constant temporal assignment (as in \cite{schoeters2023inefficiently}) to make the problem more interesting.

\section*{Declarations}
The authors have no relevant financial or non-financial interests to disclose.
The authors have no conflicts of interest to declare that are relevant to the content of this article.
All authors certify that they have no affiliations with or involvement in any organization or entity with any financial interest or non-financial interest in the subject matter or materials discussed in this manuscript.
The authors have no financial or proprietary interests in any material discussed in this article.

\backmatter 

% old bib
% \bibliographystyle{splncs04}
% \bibliography{mybibliography,buggybibliography}

\bibliography{algorithmica_tardis}

\end{document}